\documentclass[times,10pt]{shrinkarticle}
\usepackage{fullpage}
\usepackage{fixltx2e}
\usepackage{amsmath,amsthm}
\usepackage{amssymb}
\usepackage{times}
\usepackage{paralist}
\usepackage{mathrsfs}
\usepackage[latin1]{inputenc}
\usepackage{framed}
\usepackage[ruled,vlined]{algorithm2e}
\usepackage{multirow}
\usepackage{wrapfig}
\usepackage{tikz}
\usetikzlibrary{arrows,automata,shapes,decorations,calc,matrix,decorations.pathmorphing}
\usepackage{etoolbox}
\usepackage{hyperref}
\newcommand{\val}{\mbox{\rm val}}

\newcommand{\PSPACE}{\mbox{\sf PSPACE}}

\newcommand{\FNP}{\mbox{\sf FNP}}
\newcommand{\FP}{\mbox{\sf FP}}
\newcommand{\TFNP}{\mbox{\sf TFNP}}
\newcommand{\NP}{\mbox{\sf NP}}
\newcommand{\coNP}{\mbox{\sf coNP}}

\newcommand{\U}{\mathsf{U}}

\newcommand{\eps}{\varepsilon}
\newcommand{\abs}[1]{\ensuremath{\mathopen\lvert #1 \mathclose\rvert}}

\newcommand{\bit}{\ensuremath{\operatorname{bit}}}

\newcommand{\Tau}{\ensuremath{\mathcal{T}}}

\newcommand{\N}{\ensuremath{{\rm \mathbb N}}}
\newcommand{\R}{\ensuremath{{\rm \mathbb R}}}

\newcommand{\var}{\mathsf{var}}
\newcommand{\reach}{\mathsf{Reach}}
\newcommand{\safety}{\mathsf{Safety}}

\newcommand{\supp}{\mathrm{Supp}}

\newcommand{\M}[3]{\begin{pmatrix}
#2 & #3 & #3 & \dots & #3\\
#1 & #2 & #3 & \dots & #3\\
\vdots & #1 & \ddots & \ddots & \vdots\\
#1 & \vdots & \ddots & #2 & #3\\
#1  & #1 & \dots & #1 & #2
\end{pmatrix}}

\newcommand{\sco}{}
\newcommand{\ma}[1][xshift=0]{
\renewcommand{\sco}{#1}
\maA
}
\newcommand{\name}{}
\newcommand{\maA}[1]{
\renewcommand{\name}{#1}
\maB
}

\newcommand{\maB}[3][-NoValue-]{
\expandafter\maC\expandafter{\sco}{\name}{#1}{#2}{#3}
}
\newcommand{\maC}[5]{
\begin{scope}[#1,solid,-,shorten >=0,shorten <=0]
\expandafter\draw (0,0) node[rectangle, minimum height=#4 cm,minimum width=#5 cm,draw] (\name) {};

\begin{scope}[shift={($(0,0)!0.5!(-#5,-#4)$)}]

\foreach\i in {0,...,#5}{
\draw (\i,0) -- (\i,#4);
}
\foreach \j in {0, ..., #4} {
\draw (0,\j) -- (#5,\j);
}
\foreach\i in {1,...,#5}{
\foreach\j in {1,...,#4}{
\draw ($(0,1+#4)+(\i,-\j)-(0.5,0.5)$) node[rectangle, minimum height=1 cm,minimum width=1 cm,draw] (\name -\j -\i) {};
}}

\ifstrequal{#3}{-NoValue-}{
\node[draw=white] (name #2) at ($(-0.5 , 0)!.5!(-0.5 , #4)$) {#2};
}{
\node[draw=white] (name #2) at ($(-0.5 , 0)!.5!(-0.5 , #4)$) {#3};
}

\end{scope}
\end{scope}

}

\newcommand{\nloop}[2][-{stealth}]{
\renewcommand{\sco}{#1}
\renewcommand{\name}{#2}
\nloopA
}

\newcommand{\nloopA}[1][xscale=1]{
\begin{scope}[xscale=-1]
\begin{scope}[#1]
\draw[-{stealth},\sco](\name .center) to ($(\name .center)+(0,0.5cm)$) arc (0:270:0.5cm);
\end{scope}
\end{scope}
}

\newcommand{\z}{z}
\newcommand{\y}{y}

\makeatletter

\begingroup \catcode `|=0 \catcode `[= 1
\catcode`]=2 \catcode `\{=12 \catcode `\}=12
\catcode`\\=12 |gdef|@xcomment#1\end{comment}[|end[comment]]
|endgroup

\def\@comment{\let\do\@makeother \dospecials\catcode`\^^M=10\def\par{}}

\def\begincomment{\@comment\@xcomment}

\makeatother

\newenvironment{comment}{\begincomment}{}

 \newtheorem{theorem}{Theorem}
 
 \newtheorem{corollary}[theorem]{Corollary}
 \newtheorem{lemma}[theorem]{Lemma}
 \newtheorem{remark}[theorem]{Remark}

 \newtheorem{definition}[theorem]{Definition}

\newcommand{\dist}{d}

\title{Strategy Complexity of Concurrent Stochastic Games \\
with Safety and Reachability Objectives}
\author{
Krishnendu Chatterjee\thanks{IST Austria. Email: {\tt krish.chat@ist.ac.at}} \and 
Kristoffer Arnsfelt Hansen\thanks{Aarhus University. E-mail: {\tt arnsfelt@cs.au.dk}} \and 
Rasmus Ibsen-Jensen\thanks{IST Austria. Email: {\tt ribsen@ist.ac.at}}
}
\date{}

\begin{document}
\maketitle

\begin{abstract}
  We consider finite-state concurrent stochastic games, played by
  $k\geq2$ players for an infinite number of rounds, where in every
  round, each player simultaneously and independently of the other
  players chooses an action, whereafter the successor state is
  determined by a probability distribution given by the current state
  and the chosen actions. We consider reachability objectives that
  given a target set of states require that some state in the target set
  is visited, and the dual safety objectives that given a target set
  require that only states in the target set are visited. We are
  interested in the complexity of stationary strategies measured by
  their \emph{patience}, which is defined as the inverse of the
  smallest non-zero probability employed.

  Our main results are as follows: We show that in two-player zero-sum
  concurrent stochastic games (with reachability objective for one
  player and the complementary safety objective for the other player):
  (i)~the optimal bound on the patience of optimal and
  $\epsilon$-optimal strategies, for both players is doubly
  exponential; and (ii)~even in games with a single non-absorbing state
  exponential (in the number of actions) patience is necessary. In
  general we study the class of non-zero-sum games admitting
  %%stationary 
  $\eps$-Nash equilibria. 
  We show that if there is at least one player with reachability objective, 
  then doubly-exponential patience is needed in general for $\eps$-Nash 
  equilibrium strategies, whereas in contrast if all players have safety 
  objectives, then the optimal bound on patience for $\eps$-Nash equilibrium 
  strategies is only exponential.
\end{abstract}

%%% Local Variables: ***
%%% mode:latex ***
%%% TeX-master: "safetygames.tex"  ***
%%% End: ***

\section{Introduction}

\noindent{\bf Concurrent stochastic games.}
Concurrent stochastic games are played on finite-state graphs by 
$k$ players for an infinite number of rounds.
In every round, each player simultaneously and independently of the other 
players chooses moves (or actions). 
The current state and the chosen moves of the players determine a probability
distribution over the successor state. 
The result of playing the game (or a \emph{play}) is an infinite sequence 
of states and action vectors.
These games with two players were introduced in a seminal work by 
Shapley~\cite{Sha53}, and have been one of the most fundamental and 
well-studied game models in stochastic graph games.
Matrix games (or normal form games) can model a wide range problems 
with diverse applications, when there is a finite number of 
interactions~\cite{Owen95,vNM47}.
Concurrent stochastic games can be viewed as a finite set of matrix games, 
such that the choices made in the current game determine which game is 
played next, and is the appropriate model for many applications~\cite{FV97}.
Moreover, in analysis of reactive systems, concurrent games provide the 
appropriate model for reactive systems with components that interact 
synchronously~\cite{dAHM00,dAHM01,AHK02}.

\smallskip\noindent{\bf Objectives.} 
An objective for a player defines the set of desired plays for the player,
i.e., if a play belongs to the objective of the player, then the player 
wins and gets payoff~1, otherwise the player looses and gets payoff~0.
The most basic objectives for concurrent games are the \emph{reachability}
and the \emph{safety} objectives. 
Given a set $F$ of states, a reachability objective with target set $F$ 
requires that some state in $F$ is visited at least once, whereas the dual
safety objective with target set $F$ requires that only states in $F$ are
visited.
In this paper, we will only consider reachability and safety objectives.
A zero-sum game consists of two players (player~1 and player~2), and
the objectives of the players are complementary, i.e.,
a reachability objective with target set $F$ for one player and 
a safety objective with target set complement of $F$ for the other 
player.
In this work, when we refer to zero-sum games we will imply that
one player has reachability objective, and the other player has the 
complementary safety objective.
Concurrent zero-sum games are relevant in many applications.
For example, the synthesis problem in control theory (e.g., 
discrete-event systems as considered in~\cite{RW87}) 
corresponds to reactive synthesis of~\cite{PR89}. 
The synthesis problem for synchronous reactive systems is appropriately modeled as 
concurrent games~\cite{dAHM00,dAHM01,TCS:AlfaroHK07}. 
Other than control theory, concurrent zero-sum games also provide the appropriate
model to study several other interesting problems, such as two-player poker 
games~\cite{MS07}.

\smallskip\noindent{\bf Properties of strategies in zero-sum games.}
Given a zero-sum concurrent stochastic game, the player-1 \emph{value} 
$v_1(s)$ of the game at a state $s$ is the limit probability with 
which he can guarantee his objective against all strategies of 
player~2. The player-2 \emph{value} $v_2(s)$ is analogously the limit 
probability with which player~2 can ensure his own objective against
all strategies of player~1. 
Concurrent zero-sum games are determined~\cite{Eve57}, i.e., 
for each state $s$ we have $v_1(s)+v_2(s)=1$.
A \emph{strategy} for a player, given a history (i.e., finite prefix of a play) 
specifies a probability distribution over the actions. 
A \emph{stationary} strategy does not depend on the history, but only
on the current state. 
For $\epsilon\geq 0$, a strategy is $\epsilon$-optimal for a state $s$
for player~$i$ if it ensures his own objective with probability at least
$v_i(s) -\epsilon$ against all strategies of the opponent. 
A $0$-optimal strategy is an \emph{optimal} strategy.
In zero-sum concurrent stochastic games, there exist stationary optimal 
strategies for the player with safety 
objectives~\cite{BAMS:Parthasarathy71,PAMS:HPRV76}; whereas in contrast, 
for the player with reachability objectives optimal strategies do not exist 
in general, however, for every $\epsilon>0$ there exists stationary 
$\epsilon$-optimal strategies~\cite{Eve57}.

\smallskip\noindent{\bf The significance of patience and roundedness of strategies.}
The basic decision problem is as follows: given a zero-sum concurrent 
stochastic game and a rational threshold $\lambda$, decide whether 
$v_1(s) \geq \lambda$. 
The basic decision problem is in \PSPACE and is \emph{square-root
sum} hard~\cite{EY06}\footnote{The square-root sum problem 
is an important problem from computational geometry, where given a set of 
natural numbers $n_1,n_2,\ldots,n_k$, 
the question is whether the sum of the square roots exceed an integer $b$. 
The problem is not known to be in \NP.}.
Given the hardness of the basic decision problem, the next most relevant 
computational problem is to compute an approximation of the value.
The computational complexity of the approximation problem is closely related 
to the size of the description of $\epsilon$-optimal strategies.
Even for special cases of zero-sum concurrent stochastic games, 
namely \emph{turn-based} stochastic games, where in each state at most 
one player can choose between multiple moves, the best known complexity 
results are obtained by guessing an optimal strategy and computing the 
value in the game obtained after fixing the guessed strategy.
A strategy has patience $p$ if $p$ is the inverse of the smallest non-zero probability used by a distribution describing the strategy.
A rational valued strategy has roundedness $q$ if $q$ is the greatest denominator of the probabilities used by the distributions describing the strategy.
Note that if a strategy has roundedness $q$, then it also has patience at most $q$.
The description complexity of a stationary strategy can be bounded by the roundedness. 
A stationary strategy with exponential roundedness, can be described using polynomially many bits, whereas 
the explicit description of stationary strategies with doubly-exponential patience is not 
polynomial.
Thus obtaining upper bounds on the roundedness and lower bounds on the patience 
is at the heart of the computational complexity analysis of concurrent stochastic games.

\smallskip\noindent{\bf Strategies in non-zero-sum games and roundedness.}
In non-zero-sum games, the most well-studied notion of equilibrium is 
\emph{Nash equilibrium}~\cite{Nash50}, which is a strategy vector (one for each player),
such that no player has an incentive of unilateral deviation (i.e., if the 
strategies of all other players are fixed, then a player cannot switch strategy
and improve his own payoff).
The existence of Nash equilibrium in non-zero-sum concurrent stochastic games
where all players have safety objectives has been established in~\cite{SS02}.
It follows from the strategy characterization of the result of~\cite{SS02} and 
our Lemma~\ref{lem:player-stationary_deviation} that if such strategies have 
exponential roundness and forms an $\epsilon$-Nash equilibrium, for a constant 
or even logarithmic number of players, for $\epsilon>0$, then there will be 
polynomial-size witness for those strategies 
(and the approximation of a Nash equilibrium can be achieved in \TFNP, 
see Remark~\ref{rem:find_safety_only_nash}). 
Thus again the notion of roundedness is at the core of the computational complexity
of non-zero-sum games.

\smallskip\noindent{\bf Previous results and our contributions.} 
In this work we consider concurrent stochastic games (both zero-sum and 
non-zero-sum) where the objectives of the players are either reachability or 
safety.
We first describe the relevant previous results and then our contributions.

\noindent{\em Previous results.} 
For zero-sum concurrent stochastic games, the optimal bound on patience and roundedness for 
$\epsilon$-optimal strategies for reachability objectives, for $\epsilon>0$, 
is doubly exponential~\cite{HKM09,HIM11}.
The doubly-exponential lower bound is obtained by presenting a family of games
(namely, Purgatory) where the reachability player requires doubly-exponential 
patience (however, in this game the patience of the safety player 
is~1)~\cite{HKM09,HIM11}; 
whereas the doubly-exponential upper bound is obtained by expressing the 
values in the existential theory of reals~\cite{HKM09,HIM11}. 
In contrast to reachability objectives that in general do not admit optimal 
strategies, similar to safety objectives there are two related classes of 
concurrent stochastic games that admit optimal stationary strategies, namely, 
discounted-sum objectives, and ergodic concurrent games.
For both these classes the optimal bound on patience and roundedness for $\epsilon$-optimal 
strategies, for $\epsilon>0$, is exponential~\cite{CI14,I13}. 
The optimal bound on patience and roundedness for optimal and $\epsilon$-optimal strategies, 
for $\epsilon>0$, for safety objectives has been an open problem.

\noindent{\em Our contributions.} Our main results are as follows:
\begin{compactenum}

\item \emph{Lower bound: general.} 
We show that in zero-sum concurrent stochastic games, a lower bound on 
patience of optimal and $\epsilon$-optimal strategies, for
$\epsilon>0$, for safety objectives is doubly exponential (in contrast 
to the above mentioned related classes of games that admit stationary optimal strategies 
and require only exponential patience).
We present a family of games (namely, Purgatory Duel) where the optimal 
and $\epsilon$-optimal strategies, for $\epsilon>0$, for both players
require doubly-exponential patience. 

\item \emph{Lower bound: three states.}
We show that even in zero-sum concurrent stochastic games with 
three states of which two are absorbing (sink states with only self-loop
transitions) the patience required for optimal and $\epsilon$-optimal 
strategies, for $\epsilon>0$, is exponential (in the number of actions).
An optimal (resp., $\epsilon$-optimal, for $\epsilon>0$) strategy in a 
game with three states (with two absorbing states) is basically an optimal
(resp., $\epsilon$-optimal) strategy of a matrix game, where some entries
of the matrix game depends on the value of the non-absorbing state 
(as some transitions of the non-absorbing state can lead to itself).
In standard matrix games, the patience for $\epsilon$-optimal strategies,
for $\epsilon>0$, is only logarithmic~\cite{LMM03}; and perhaps surprisingly in 
contrast we show that the patience for $\epsilon$-optimal strategies in 
zero-sum concurrent stochastic games with three states is exponential 
(i.e., there is a doubly-exponential increase from logarithmic to exponential).

\item \emph{Upper bound.}
We show that in zero-sum concurrent stochastic games, an upper bound 
on the patience of optimal strategies and an upper bound on the patience and roundedness of $\eps$-optimal strategies, for $\eps>0$,
is as follows: (a)~doubly exponential in general; and (b)~exponential for the
safety player if the number of value classes (i.e., the number of different 
values in the game) is constant.
Hence our upper bounds on roundedness match our lower bound results for patience.
Our results also imply that if the number of value classes is constant, then 
the basic decision problem is in \coNP (resp., \NP) if player~1 has 
reachability (resp., safety) objective. 
%%{\bf KRISH TO RASMUS: PLS CHECK LAST SENTENCE. MODIFY AS THE RESULT.}

\item \emph{Non-zero-sum games.}
We consider non-zero-sum concurrent stochastic games with reachability and 
safety objectives.
First, we show that it easily follows from our example family of Purgatory Duel
that if there are at least two players and there is at least one player with 
reachability objective, then a lower bound on patience for $\epsilon$-Nash 
equilibrium is doubly exponential, for $\epsilon>0$, for \emph{all} players.
In contrast, we show that if all players have safety objectives, then the 
optimal bound on patience of strategies for $\epsilon$-Nash equilibrium is 
exponential, for $\epsilon>0$ (i.e., for upper bound we show that there always 
exists an $\epsilon$-Nash equilibrium where the strategy of each player requires at 
most exponential roundedness; and there exists a family of games, where for any 
$\epsilon$-Nash equilibrium the strategies of all players require at least 
exponential patience).
\end{compactenum}
In summary, we present a complete picture of the patience and roundedness required in 
zero-sum concurrent stochastic games, and non-zero-sum concurrent stochastic
games with safety objectives for all players.
Also see Section~\ref{subsec:imp-tech} %%(of the full version) 
for a discussion on important technical aspects of our results.

\smallskip\noindent{\bf Distinguishing aspects of safety and reachability.}
While the optimal bound on patience and roundedness we establish in zero-sum 
concurrent stochastic games for the safety player matches that for the 
reachability player, there are many distinguishing aspects for safety as 
compared to reachability in terms of the number of value classes (as shown in 
Table~\ref{tab:strategy-complexity}).
For the reachability player, if there is one value class, then the patience and roundedness
required is linear: it follows from the results of~\cite{Cha07} that if there 
is one value class then all the values must be either~1 or~0; and if all states 
have value~0, then any strategy is optimal, and if all states have value~1, 
then it follows from~\cite{TCS:AlfaroHK07,CdAH11} that there is an almost-sure 
winning strategy (that ensures the objective with probability~1) from all 
states and the optimal bound on patience and roundedness is linear.
The family of game graphs defined by Purgatory has two value classes, and 
the reachability player requires doubly exponential patience and roundedness, even for two 
value classes.
In contrast, if there are (at most) two value classes, then again the values 
are~1 and~0; and in value class~1, the safety player has an optimal strategy 
that is stationary and deterministic (i.e., a positional strategy) and has 
patience and roundedness~1~\cite{TCS:AlfaroHK07}, and in value class~0 any strategy is optimal.
While for two value classes, the patience and roundedness is~1 for the safety player, we show
that for three value classes (even for three states) the patience and roundedness is 
exponential, and in general the patience and roundedness is doubly exponential 
(and such a finer characterization does not exist for reachability objectives).
Finally, for non-zero-sum games (as we establish), if there are at least 
two players, then even in the presence of one reachability player, the 
patience required is at least doubly exponential, whereas if all players have 
safety objectives, the patience required is only exponential.

\begin{table}[h]
\begin{center}
\begin{tabular}{| c | c | c |}
\hline
\# Value classes     & Reachability & Safety \\
\hline
1 & Linear & One \\
\hline
2 & Double-exponential & One \\
\hline
3 & Double-exponential & {\bf Exponential} \\
  & & {\bf LB, Theorem~\ref{thm:small n}} \\
\hline 
Constant & Double-exponential & {\bf Exponential} \\
  %%& & {\bf UB, Theorem~\ref{thm:upper}} \\
  & & {\bf UB, Corollary~\ref{COR:RoundedEpsOptimal}} \\
\hline 
General & Double-exponential & {\bf Double-exponential} \\
  & & {\bf LB, Theorem~\ref{thm:big n}} \\
%%  & & {\bf UB, Theorem~\ref{thm:upper}} \\
  & & {\bf UB, Corollary~\ref{COR:RoundedEpsOptimal}} \\
\hline 
\end{tabular}
\end{center}
\caption{Strategy complexity (i.e., patience and roundedness of $\epsilon$-optimal strategies, for $\epsilon>0$) 
of reachability vs safety objectives depending on the number of value classes.
Our results are bold faced, and LB (resp., UB) denotes lower (resp., upper) bound on patience (resp., roundedness).
}\label{tab:strategy-complexity}
\end{table}

\smallskip\noindent{\bf Our main ideas.} 
Our most interesting results are the doubly-exponential and exponential 
lower bound on the patience and roundedness in zero-sum games. 
We now present a brief overview about the lower bound example.

The game of \emph{Purgatory} \cite{HKM09,HIM11} is a concurrent
reachability game \cite{TCS:AlfaroHK07} that was defined as an example
showing that the \emph{reachability} player must, in order to play
near optimally, use a strategy with non-zero %%behaviour 
probabilities that are \emph{doubly exponentially} small in the number of 
states of the game (i.e., the patience is doubly exponential). 
%%We call the inverse of the smallest such probability the \emph{patience} of the strategy.

In this paper we present another example of a reachability game
where this is the case for the \emph{safety} player as well. The game
Purgatory consists of a (potentially infinite) sequence of \emph{escape
  attempts}. In an escape attempt one player is given the role of the
\emph{escapee} and the other player is given the role as the
\emph{guard}. An escape attempt consists of at most $N$ rounds. In
each round, the guard selects and hides a number between $1$ and $m$,
and the escapee must try to guess the number. If the escapee
successfully guesses the number $N$ times, the game ends with the
escapee as the winner. If the escapee incorrectly guesses a number
which is strictly larger than the hidden number, the game ends with
the guard as the winner. Otherwise, if the escapee incorrectly guesses
a number which is strictly smaller than the hidden number, the escape
attempt is over and the game continues.

The game of Purgatory is such that the reachability player is always
given the role of the escapee, and the safety player is always given
the role of the guard. If neither player wins during an escape attempt
(meaning there is an infinite number of escape attempts) the safety
player wins. Purgatory may be modelled as a concurrent reachability
game consisting of $N$ non-absorbing positions in which each player has
$m$ actions. The value of each non-absorbing position is 1. This means
that the reachability player has, for any $\eps>0$, a stationary
strategy that wins from each non-absorbing position with probability at
least $1-\eps$ \cite{Eve57}, but such strategies must have 
doubly-exponential patience.  In fact for $N$ sufficiently large and $m
\geq 2$, such strategies must have patience at least $2^{m^{N/3}}$ for
$\eps=1-4m^{-N/2}$ \cite{HIM11}. For the safety player however, the
situation is simple: \emph{any} strategy is optimal.

We introduce a game we call the \emph{Purgatory Duel} in which the
safety player must also use strategies of doubly-exponential patience
to play near optimally. The main idea of the game is that it forces
the safety player to behave as a reachability player. We can describe
the new game as a variation on the above description of the Purgatory
game. The Purgatory Duel consists also of a (potentially infinite)
sequence of escape attempts. But now, before each escape attempt the
role of the escapee is given to each player with probability
$\frac12$, and in each escape attempt the rules are as described
above. The game remains asymmetric in the sense that if neither player
wins during an escape attempt, the safety player wins.

The Purgatory Duel may be modelled as a concurrent reachability game
consisting of $2N+1$ non-absorbing positions, in which each player has
$m$ actions, except for a single position where the players each have
just a single action.

\smallskip\noindent{\em Technical contribution.}
The key non-trivial aspects of our proof are as follows: first, is to 
come up with the family of games, namely, Purgatory Duel, where the 
$\epsilon$-optimal strategies, for $\epsilon\geq0$, for the players are 
symmetric, even though the objectives are complementary;
and then the precise analysis of the game needs to combine and extend 
several ideas, such as refined analysis of matrix games, and analysis of 
perturbed Markov decision processes (MDPs) which are one-player 
stochastic games.

\smallskip\noindent{\bf Highlights.} 
We highlight two features of our results, namely, the surprising aspects 
and the significance (see Section~\ref{subsec:imp-features} 
for further details).
\begin{compactenum}
\item {\em Surprising aspects.}
The first surprising aspect of our result is the doubly-exponential
lower bound for concurrent safety games. The properties of strategies
in concurrent safety games resemble concurrent disocunted games,
as in both cases optimal stationary strategies exist, and locally
optimal strategies are optimal.
We show that in contrast to concurrent discounted games where exponential
patience suffices for concurrent safety games doubly-exponential patience
is necessary.
The second surprising aspect is the lower bound example itself. 
The lower bound example is obtained as follows: (i)~given Purgatory we first 
obtain simplified Purgatory by changing the start state such that it deterministically
goes to the next state; (ii)~we then consider its dual where the roles of the players
are exchanged; and (iii)~Purgatory duel is obtained by merging the start states of
simplified Purgatory and its dual. 
Both in simplified Purgatory and its dual, there are only two value classes, and
positional optimal strategies exist for the safety player.
Surprisingly we show that a simple merge operation gives a game with linear number
of value classes and the patience increases from~1 to doubly-exponential.
Finally, the properties of strategies in concurrent reachability and safety games
differ substantially.
An important aspect of our lower bound example is that we show how to modify an 
example for reachability game to obtain the result for safety games.

\item {\em Significance.} Our most important results are the lower bounds, and
the main significance is threefold.
First, the most well-studied way to obtain computational complexity result in games 
is to explicitly guess strategies, and then verify the game obtained fixing the strategy.
The lower bound for concurrent reachability games by itself did not rule out that better 
complexity results can be obtained through better strategy complexity for safety games
(indeed, for constant number of value classes, we obtain a better complexity result than known before
due to the exponential bound on roundedness).
Our doubly-exponential lower bound shows that in general the method of explicitly guessing strategies
would require exponential space, and would not yield \NP\ or \coNP\ upper bounds.
Second, one of the most well-studied algorithm for games is the strategy-iteration algorithm. 
Our result implies that any natural variant of the strategy-iteration algorithm for the safety 
player that explicitly compute strategies require exponential space in the worst-case.
Finally, in games, strategies that are witness to the values and specify how to play the game, 
are as important as values, and our results establish the precise strategy 
complexity (matching upper bound of roundedness with lower bounds of patience).

\end{compactenum}

\smallskip\noindent{\em Related work.}
We have already discussed the relevant related works such 
as~\cite{BAMS:Parthasarathy71,PAMS:HPRV76,Eve57,EY06,HKM09,HIM11,TCS:AlfaroHK07} 
on zero-sum games. 
We discuss relevant related works for non-zero-sum games. 
The computational complexity of \emph{constrained} Nash equilibrium, 
which asks the existence of Nash (or $\epsilon$-Nash, for $\epsilon>0$) 
equilibrium that guarantees at least a payoff vector has been studied.
The constrained Nash equilibrium problem is undecidable even for turn-based
stochastic games, or concurrent deterministic games with randomized 
strategies~\cite{UW11,BMS14}.
The complexity of constrained Nash equilibrium in concurrent deterministic games
with pure strategies has been studied in~\cite{BBM10,BBMU12}.
In contrast, we study the complexity of computing some Nash equilibrium in 
randomized strategies in concurrent stochastic games, and our result on 
roundedness implies that with safety objectives for all players the approximation 
of some Nash equilibrium can be achieved in \TFNP.

\section{Definitions}
\smallskip\noindent{\bf Other number.}
Given a number $i\in\{1,2\}$ let $\widehat{i}$ be the other number, i.e., if $i=1$, then $\widehat{i}=2$ and if $i=2$, then $\widehat{i}=1$.

\smallskip\noindent{\bf Probability distributions.} A {\em probability distribution} $d$ over a finite set $Z$, is a map $d:Z\rightarrow [0,1]$, such that $\sum_{z\in Z} d(z)=1$. Fix a probability distribution $d$ over a set $Z$.
The distribution $d$ is {\em pure (Dirac)} if $d(z)=1$ for some $z\in Z$ and for convenience we overload the notation and let $d=z$. The {\em support} $\supp(d)$ is the subset $Z'$ of $Z$, such that $z\in Z'$ if and only if $d(z)>0$.
The distribution $d$ is {\em totally mixed} if $\supp(d)=Z$. The {\em patience} of $d$ is $\max_{z\in \supp(d)}\frac{1}{d(z)}$, i.e., the inverse of the minimum non-zero probability. The {\em roundedness} of $d$, if $d(z)$ is a rational number for all $z\in Z$,  is the greatest denominator of $d(z)$. Note that roundness of $d$ is always at least the patience of $d$. Given two elements $z,z'\in Z$, the probability distribution $d=\U(z,z')$ over $Z$ is such that $d(z)=d(z')=\frac{1}{2}$. Let $\Delta(Z)$ be the set of all probability distributions over $Z$.

\smallskip\noindent{\bf Concurrent game structure.}
A concurrent game structure for $k$ players, consists of (1)~a finite set of {\em states} $S$, of size $N$; and (2)~for each state~$s\in S$ and each player~$i$ a set $A_s^i$ of {\em actions} (and $A^i=\bigcup_{s} A_s^i$ is the set of all actions for player~$i$, for each $i$; and $A=\bigcup_{i} A^i$ is the set of all actions) such that $A_s^i$ consists of at most $m$ actions; and (3)~a stochastic {\em transition function} $\delta:S\times A^1\times A^2\times\dots\times A^k\to \Delta(S)$. Also, a state $s$ is {\em deterministic} if 
$\delta(s,a_1,a_2,\dots,a_k)$ is pure (deterministic), for all $a_i\in A^i_s$ and for all  $i$. A state $s$ is called {\em absorbing} if $A^i_s=\{a\}$ for all $i$ and  $\delta(s,a,a,\dots,a)=s$. The number $\delta_{\min}$ is \[\min_{s,a_1,\dots,a_k,s'\in \supp(\delta(s,a_1,a_2,\dots,a_k))}(\delta(s,a_1,a_2,\dots,a_k)(s'))\enspace ,\]
i.e., the smallest non-zero transition probability.

\smallskip\noindent{\bf Safety and reachability objectives.}
Each player~$i$, who has a safety or reachability objective, is identified by a pair $(t_i,S^i)$, where $t_i\in\{\reach,\safety\}$ and $S^i\subseteq S$.

\smallskip\noindent{\bf Concurrent games and how to play them.}
Fix a number $k$ of players. A concurrent game consists of a concurrent game structure for $k$ players and for each player~$i$ a pair $(t_i,S^i)$, identifying the type of that player.
The game $G$, starting in state $s$, is played as follows: initially a pebble is placed on $v_0:=s$. In each time step $T\geq 0$, the pebble is on some state $v_T$  and each player selects (simultaneously and independently of the other players, like in the game rock-paper-scissors) an action $a_{T+1}^i\in A_{v_T}^i$. Then, the game selects $v_{T+1}$ according to the probability distribution $\delta(v_T,a_{T+1}^1,a_{T+1}^2,\dots,a_{T+1}^k)$ and moves the pebble onto $v_{T+1}$. 
The game then continues with time step $T+1$ (i.e., the game consists of infinitely many time steps). For a round~$T$, let $a_{T+1}$ be the vector of choices of the actions for the players, i.e., $(a_{T+1})_i$ is the choice of player~$i$, for each~$i$.
Round~0 is identified by $v_0$ and round~$T>0$ is then identified by the pair $(a_{T},v_T)$. 
A {\em play} $P_{s}$, starting in state $v_0=s$, is then a sequence of rounds \[(v_0,(a_{1},v_1),(a_{2},v_2),\dots, (a_{T},v_T),\dots)\enspace ,\] and for each $\ell$ a prefix of $P^\ell_{s}$ of length $\ell$ is then \[(v_0,(a_{1},v_1),(a_{2},v_2),\dots, (a_{T},v_T),\dots,(a_{\ell},v_\ell))\enspace ,\] 
and we say that $P^{\ell}_{s}$ {\em ends in} $v_\ell$.
For each $i$, player~$i$ wins in the play $P_s$, 
if $t_i=\safety$ and $v_T\in S_i$ for all $T\geq0$; or 
if $t_i=\reach$ and $v_T\in S_i$, for some $T\geq0$. 
Otherwise, player $i$ loses. For each $i$, player $i$ tries to maximize the probability that he wins.

\smallskip\noindent{\bf Strategies.} 
Fix a player~$i$. A strategy is a recipe to choose a probability 
distribution over actions given a finite prefix of a play.
Formally, a strategy $\sigma_i$ for player~$i$ is a map from $P^\ell_{s}$, 
for a play $P_{s}$ of length $\ell$ starting at state $s$, to a distribution 
over $A_{v_\ell}^i$. 
Player~$i$ {\em follows} a strategy $\sigma_i$, if given the current prefix of 
a play is $P^\ell_{s}$, he selects $a_{\ell+1}$ according to 
$\sigma_i(P^\ell_{s})$, for all plays $P_{s}$ starting at $s$ and all 
lengths $\ell$.
A strategy $\sigma_i$ for player~$i$, is {\em stationary}, if for all $\ell$ 
and $\ell'$, and all pair of plays $P_{s}$ and $P'_{s'}$, starting at states 
$s$ and $s'$ respectively, such that $P^\ell_{s}$ and $(P')^{\ell'}_{s'}$ 
ends in the same state $t$, we have that $\sigma_i(P^{\ell}_{s})=
\sigma_i((P')^{\ell'}_{s'})$; and we write $\sigma_i(t)$ for the unique 
distribution used for prefix of plays ending in $t$. 
The {\em patience} (resp., {\em roundedness}) of a strategy $\sigma_i$ is the supremum of the patience  (resp., roundedness)
of the distribution $\sigma_i(P^\ell_{s})$, over all plays $P_{s}$ starting at 
state $s$, and all lengths $\ell$.
Also, a strategy $\sigma_i$ is {\em pure} (resp., {\em totally mixed}) if 
$\sigma_i(P^{\ell}_{s})$ is pure (resp., totally mixed), for all plays 
$P_{s}$ starting at $s$ and all lengths $\ell$. 
A strategy is {\em positional} if it is pure and stationary. 
For each player~$i$, let $\Sigma^i$ be the set of all strategies for the
respective player. 
 
 \smallskip\noindent{\bf Strategy profiles and Nash equilibria.} A {\em strategy profile} $\sigma=(\sigma_i)_{i}$ is a vector of strategies, one for each player.  A strategy profile $\sigma$ defines a unique probability measure on plays, denoted $\Pr_{\sigma}$, when the players follow their respective strategies~\cite{Var85}.  Let $u(G,s,\sigma,i)$ be the probability that player~$i$ wins the game $G$ when the players follow $\sigma$ and the play starts in $s$ (i.e., the utility or payoff for player~$i$). 
Given a strategy profile $\sigma=(\sigma_i)_{i}$ and a strategy $\sigma_i'$ for player~$i$, the strategy profile $\sigma[\sigma_i']$ is the strategy profile where 
the strategy for player~$i$ is $\sigma_i'$  and the strategy for player~$j$ is $\sigma_j$ for $j\neq i$. 
Fix a state $s$ and $\eps\geq 0$. 
 A strategy profile $\sigma$ forms an {\em $\eps$-Nash equilibrium from state $s$} if for all $i$ and all strategies $\sigma_i'$ for player~$i$, we have that $u(G,s,\sigma,i)\geq u(G,s,\sigma[\sigma_i'],i)-\eps$. A strategy profile $\sigma$ forms an {\em $\eps$-Nash equilibrium} if it forms an $\eps$-Nash equilibrium from all states $s$.
 Also a strategy profile forms a {\em Nash equilibrium} (resp., from state $s$, for some $s$) if it forms a $0$-Nash equilibrium (resp., from state $s$). We say that a strategy profile has a property (e.g., is stationary) if each of the strategies in the profile has that property.

\subsection{Zero-sum concurrent stochastic games}
A zero-sum game consists of two players with complementary objectives.
Since we only consider reachability and safety objectives, a zero-sum 
concurrent stochastic game consists of a two-player concurrent stochastic game 
with reachability objective for player~1 and the complementary safety 
objective for player~2 (such a game is also referred to as concurrent
reachability game).

\smallskip\noindent{\bf Concurrent reachability game.} A concurrent reachability game is a concurrent game with two players, identified by $(\reach,S^1)$ and $(\safety,S\setminus S^1)$. Observe that in such games, exactly one player wins each play (this implies that the games are zero-sum). Note that for all 
strategy profiles $\sigma$ we have $u(G,s,\sigma,1)+u(G,s,\sigma,2)=1$. For ease of notation and tradition, we write $u(G,s,\sigma_1,\sigma_2)$ for $u(G,s,\sigma_1,\sigma_2,1)$, for all concurrent reachability games $G$, states $s$, and strategy profiles $\sigma=(\sigma_1,\sigma_2)$. 
Also if the game $G$ is clear from context we drop it from the notation.

\smallskip\noindent{\bf Values of concurrent reachability games.}
Given a concurrent reachability game $G$, the {\em upper value} of $G$ starting in $s$  is \[\overline{\val}(G,s)=\sup_{\sigma_1\in \Sigma^1}\inf_{\sigma_2\in \Sigma^2}u(G,s,\sigma_1,\sigma_2)\enspace ;\] 
and % Given a concurrent reachability game $G$, 
the {\em lower value} of $G$ starting in $s$ is
\[\underline{\val}(G,s)=\inf_{\sigma_2\in \Sigma^2}\sup_{\sigma_1\in \Sigma^1}u(G,s,\sigma_1,\sigma_2)\enspace .\]
As shown by~\cite{Eve57} we have that \[\val(G,s):=\overline{\val}(G,s)=\underline{\val}(G,s) \enspace ;\] 
and this common number is called the {\em value} of $s$. 
We will sometimes write $\val(s)$ for $\val(G,s)$ if $G$ is clear from the context. We will also write $\val$ for the vector where $\val_s=\val(s)$.

\smallskip\noindent{\bf ($\eps$-)optimal strategies for concurrent reachability games.}
For an $\eps\geq 0$, a strategy $\sigma_1$ for player~$1$ (resp., $\sigma_2$ for player~2) is called {\em $\eps$-optimal} if for each state $s$ we have that $\val(s)-\eps\leq \inf_{\sigma_2\in \Sigma^2} u(s,\sigma_1,\sigma_2)$ (resp., $\val(s)+\eps\geq \sup_{\sigma_1\in \Sigma^1} u(s,\sigma_1,\sigma_2)$). For each $i$, a strategy $\sigma_i$ for player~$i$ is called {\em optimal} if it is 0-optimal.
There exist concurrent reachability games in which player~$1$ does not have optimal strategies, see~\cite{Eve57} for an example\footnote{note that it is not because that we require the strategy to be optimal for each start state, since if there was one for each start state separately then there would be one for all, since this is not just for stationary strategies}. On the other hand in all games $G$ player~$1$ has a stationary $\eps$-optimal strategy for each $\eps>0$. In all games player~2 has an optimal stationary strategy (thus also an $\eps$-optimal stationary strategy for all $\eps>0$)~\cite{BAMS:Parthasarathy71,PAMS:HPRV76}.
%%Player~2 always have an optimal strategy as shown by~\cite{BAMS:Parthasarathy71,PAMS:HPRV76}.
Also, given a stationary strategy $\sigma_1$ for player~1 we have that there exists a positional strategy $\sigma_2$, such that 
$u(s,\sigma_1,\sigma_2)=\inf_{\sigma_2'\in \Sigma^2} u(s,\sigma_1,\sigma_2')$, i.e., we only need to consider positional strategies for player~2. Similarly, we only need to consider positional strategies for player~1, if we are given a stationary strategy for player~2. 

\smallskip\noindent{\bf ($\eps$-)optimal strategies compared to ($\eps$-)Nash equilibria.}
It is well-known and easy to see that for concurrent reachability games, a strategy profile $\sigma=(\sigma_1,\sigma_2)$ is optimal if and only if $\sigma$ forms a Nash equilibrium. Also, if $\sigma_1$ is $\eps$-optimal and $\sigma_2$ is $\eps'$-optimal, for some $\eps$ and $\eps'$, then $\sigma=(\sigma_1,\sigma_2)$ forms an $(\eps+\eps')$-Nash equilibrium. Furthermore, if $\sigma=(\sigma_1,\sigma_2)$ forms an $\eps$-Nash equilibrium, for some $\eps$, then $\sigma_1$ and $\sigma_2$ are $\eps$-optimal\footnote{observe that the two latter properties implies the former, but all are included to make it clear that there is a strong connection}.
%%% KRISH: PARAGRAPH CAN BE OMITTED FOR SHORT VERSION.

\smallskip\noindent{\bf Markov decision processes and Markov chains.}
For each player $i$, a {\em Markov decision process (MDP) for player~$i$} is a concurrent game where the size of $A_s^{j}$ is 1 for all $s$ and $j\neq i$. 
A {\em Markov chain} is an MDP for each player (that is the size of $A_s^j$ is~1 for all $s$ and $j$). 
A {\em closed recurrent set} of a Markov chain $G$ is a maximal (i.e., no closed recurrent set is a subset of another) set $S'\subseteq S$ such that for all pairs of states $s,s'\in S$, the play starting at $s$ reaches state $s'$ eventually with probability~1 (note that it does not depend on the choices of the players as we have a Markov chain). For all starting states, eventually a closed recurrent set is reached with probability~1, and then plays stay in the closed recurrent set.
Observe that fixing a stationary strategy for all but one player in a concurrent game, the resulting game is an MDP for the remaining player. Hence, fixing a stationary strategy for each player gives a Markov chain.

\subsection{Matrix games and the value iteration algorithm} %% for concurrent reachability games}
A (two-player, zero-sum) matrix game consists of a matrix $M\in \R^{r\times c}$. We will typically let $M$ refer to both the matrix game and the matrix and it should be clear from the context what it means.
A matrix game $M$ is played as follows: player~1 selects a row $a_1$ and at the same time, without knowing which row was selected by player~1, player~2 selects a column $a_2$. The {\em outcome} is then $M_{a_1,a_2}$. Player~1 then tries to maximize the outcome and player~2 tries to minimize it.

\smallskip\noindent{\bf Strategies in matrix games.}
A {\em strategy} $\sigma_1$ (resp., $\sigma_2$) for player~1 (resp., player~2) is a probability distribution over the rows (resp., columns) of $M$. A strategy profile $\sigma=(\sigma_1,\sigma_2)$ is a pair of strategies, one for each player.
Given a strategy profile $\sigma=(\sigma_1,\sigma_2)$ the payoff $u(M,\sigma_1,\sigma_2)$ under those strategies is the expected outcome if player~1 picks row $a_1$ with probability $\sigma_1(a_1)$ and player~2 picks column $a_2$ with probability $\sigma_2(a_2)$ for each $a_1$ and $a_2$, i.e., \[u(M,\sigma_1,\sigma_2)=\sum_{a_1}\sum_{a_2} M_{a_1,a_2}\cdot \sigma_1(a_1)\cdot \sigma_2(a_2)\enspace .\]

\smallskip\noindent{\bf Values in matrix games.}
The {\em upper value} of a matrix game is $\overline{\val}(M)=\sup_{\sigma_1}\inf_{\sigma_2}u(M,\sigma_1,\sigma_2)$.
The {\em lower value} of a matrix game is $\underline{\val}(M)=\inf_{\sigma_2}\sup_{\sigma_1}\sum_{a_1}u(M,\sigma_1,\sigma_2)$.
One of the most fundamental results in game theory, as shown by~\cite{vNM47}, is that $\val(M):=\overline{\val}(M)=\underline{\val}(M)$. 
This common number is called the {\em value}. 

\smallskip\noindent{\bf ($\eps$-)optimal strategies in matrix games.}
A strategy $\sigma_1$ for player~1 is {\em $\eps$-optimal}, for some number $\eps\geq 0$ if $\val(M)-\eps\leq \inf_{\sigma_2}u(M,\sigma_1,\sigma_2)$. Similarly, a strategy $\sigma_2$ for player~2 is {\em $\eps$-optimal}, for some number $\eps\geq 0$ if $\val(M)+\eps\geq \sup_{\sigma_1}u(M,\sigma_1,\sigma_2)$. A strategy is {\em optimal} if it is $0$-optimal. There exists an optimal strategy for each player in all matrix games~\cite{vNM47}.
Given an optimal strategy $\sigma_1$ for player~1, consider the vector $\overline{v}$, such that $\overline{v}_j=u(M,\sigma_1,j)$ for each column $j$. 
Then we have that $\overline{v}_j=\val(M)$ for each $j$ such that there exists an optimal strategy $\sigma_2$ for player~2, where $\sigma_2(j)>0$. Similar analysis holds for optimal strategies of player~2. This also shows that given an optimal strategy $\sigma_1$ for player~1 we have that $u(M,\sigma_1,\sigma_2)$ is minimized for some pure strategy $\sigma_2$ and similarly for optimal strategies $\sigma_2$ for player~2.
Given a matrix game $M$, an optimal strategy for each player and the value of $M$ can be computed in polynomial time using linear programming.

\smallskip\noindent{\bf The matrix game $A^s[\overline{v}]$ and $A^s$.}
Fix a concurrent reachability game $G$.
Given a vector $\overline{v}$ in $\R^{S}$ and a state $s$ (in $G$), the matrix game $A^s[\overline{v}]=[a_{i,j}]$ is the matrix game where $a_{i,j}=\sum_{s'\in S}\delta(s,i,j)(s')\cdot\overline{v}_{s'}$. Given a state $s$, the matrix game $A^s$ is the matrix game $A^s[\val]$. As shown by~\cite{BAMS:Parthasarathy71,PAMS:HPRV76}, each optimal stationary strategy $\sigma_2$ for player~2 in $G$ is such that for each state $s$  the distribution $\sigma_2(s)$ is an optimal strategy in the matrix game $A^s$. Also, conversely, if $\sigma_2(s)$ is an optimal strategy in $A^s$ for each $s$, then $\sigma_2$ is an optimal stationary strategy in $G$. Furthermore, also as shown by~\cite{BAMS:Parthasarathy71,PAMS:HPRV76}, we have that $\val(s)=\val(A^s)$ for each state $s$.

\smallskip\noindent{\bf The value iteration algorithm.}
The conceptually simplest algorithm for concurrent reachability games is the \emph{value iteration} algorithm, which is an iterative approximation algorithm. The idea is as follows: Given a concurrent reachability game $G$, consider the game $G^t$ where a {\em time-limit} $t$ (some non-negative integer) has been introduced. The game $G^t$ is then played as $G$, except that player~2 wins if the time-limit is exceeded (i.e., he wins after round $t$ unless a state in $S^1$ has been reached before that). 
(The game $G^t$ has a value like in the above definition of matrix games since the game only has a finite number of pure strategies and thus can be reduced to a matrix game).
The value of $G^t$ starting in state $s$ then converges to the value of $G$ starting in $s$ as $t$ goes to infinity as shown by~\cite{Eve57}. 
More precisely, the algorithm is defined on a vector $\overline{v}^t$ which is the vector where $\overline{v}^t_s$ is the value of $G^t$ starting in $s$. We can compute $\overline{v}^t_s$ recursively for increasing $t$ as follows 
\[
\overline{v}^t_s=\begin{cases}
1 &\text{if }s\in S^1\\
0 &\text{if }s\not\in S^1\text{ and }t=0\\
\val(A^s[\overline{v}^{t-1}]) &\text{if }s\not\in S^1\text{ and }t\geq 1\enspace .
\end{cases}
\]
We have that $\overline{v}^t_s\leq \overline{v}^{t+1}_s\leq \val(s)$ for all $t$ and $s$, and 
for all $s$ we have $\lim_{t\to\infty} \overline{v}^t_s =\val(s)$,
as shown by~\cite{Eve57}. As shown by~\cite{HIM11,HKLMT11} the smallest time-limit~$t$ such that $\overline{v}^t_s\geq \val(s)-\eps$ can be as large as $\eps^{-m^{\Omega(N)}}$ for some games (of $N$ states and at most $m$ actions in each state for each player) and $s$, for $\eps>0$. On the other hand it is also at most $\eps^{-m^{O(N^2)}}$ as shown by~\cite{HIM11}.

\section{Zero-sum Concurrent Stochastic Games: Patience Lower Bound}
In this section we will establish the doubly-exponential lower bound
on patience for zero-sum concurrent stochastic games.
First we define the game family, namely, \emph{Purgatory Duel} and we
also recall the family \emph{Purgatory} that will be used in our proofs.
We split our proof about the patience in Purgatory Duel in three parts.
First we present some refined analysis of matrix games, and use the 
analysis to first prove the lower bound for optimal strategies, and then
for $\eps$-optimal strategies, for $\eps>0$.

\smallskip\noindent{\bf The Purgatory Duel.} In this paper we specifically focus on the following concurrent reachability game, the {\em Purgatory Duel}\footnote{To allow a more compact notation, we have here exchanged the criterias for when the safety player wins as a guard and when the escape attempt ends, as compared to the textual description of the game given in the introduction.}, defined on a pair of parameters $(n,m)$. The game consists of $N=2n+3$ states, namely $\{v_1^1,v_2^1,\dots, v_n^1,v_1^2,v_2^2,\dots, v_n^2, v_s,\top,\bot\}$ and all but $v_s$ are deterministic. 
To simplify the definition of the game, let $v_0^1=v_{n+1}^2=\bot$ and $v_0^2=v_{n+1}^1=\top$.
The states $\top$ and $\bot$ are absorbing.
For each $i\in \{1,2\}$ and $j\in\{1,\dots,n\}$, the state $v^i_j$ is such that $A_{v^i_j}^1=A_{v^i_j}^2=\{1,2,\dots, m\}$ and for each $a_1,a_2$ we have that
\[\delta(v^i_j,a_1,a_2)=\begin{cases}
v_s & \text{if }a_1>a_2\\
v_0^i & \text{if }a_1<a_2\\
v_{j+1}^i & \text{if }a_1=a_2\enspace .
\end{cases}
\]
Finally, $A_{v_s}^1=A_{v_s}^2=\{a\}$ and $\delta(v_s,a,a)=\U(v_1^1,v_1^2)$.
Furthermore, $S^1=\{\top\}$. There is an illustration of the Purgatory Duel with $m=n=2$ in Figure~\ref{fig:pd}.

\begin{figure}
\begin{center}
\begin{tikzpicture}[node distance=3cm,-{stealth},shorten >=2pt]
\ma{top}[$\top$]{1}{1};

\ma[shift={($(top)+(0cm,-3cm)$)}]{v21}[$v_2^1$]{2}{2};
\ma[shift={($(v21)+(0cm,-4.5cm)$)}]{v11}[$v_1^1$]{2}{2};

\ma[shift={($(v11)+(3cm,-3cm)$)}]{vs}[$v_s$]{1}{1};

\ma[shift={($(v11)+(6cm,0)$)}]{v12}[$v_1^2$]{2}{2};
\ma[shift={($(v12)+(0,4.5cm)$)}]{v22}[$v_2^2$]{2}{2};

\ma[shift={($(v22)+(0,3cm)$)}]{bot}[$\bot$]{1}{1};
\nloop{top};
\nloop{bot};

\draw (v21-2-2.center) -- (top);
\draw (v21-1-1.center) -- (top);
\draw (v21-1-2.center) -- (bot);
\draw (v21-2-1.center) to[out=-70,in=90] (vs);

\draw (v11-2-2.center) -- (v21);
\draw (v11-1-1.center) -- (v21);
\draw (v11-1-2.center) -- (bot);
\draw (v11-2-1.center) to[out=-90,in=135] (vs);

\draw (v22-2-2.center) -- (bot);
\draw (v22-1-1.center) -- (bot);
\draw (v22-1-2.center) to[out=110,in=-25] (top);
\draw (v22-2-1.center) to[out=-110,in=45] (vs);

\draw (v12-2-2.center) -- (v22);
\draw (v12-1-1.center) -- (v22);
\draw (v12-1-2.center) to[out=120,in=-45] (top);
\draw (v12-2-1.center) to[out=-90,in=0] (vs);

\draw[dashed] (vs-1-1.center) to[out=90,in=-45]  (v11);
\draw[dashed] (vs-1-1.center) to[out=90,in=-135] (v12);

\end{tikzpicture}
\end{center}
\caption{An illustration of the Purgatory Duel with $m=n=2$. The two dashed edges have probability~$\frac{1}{2}$ each.\label{fig:pd}}
\end{figure}
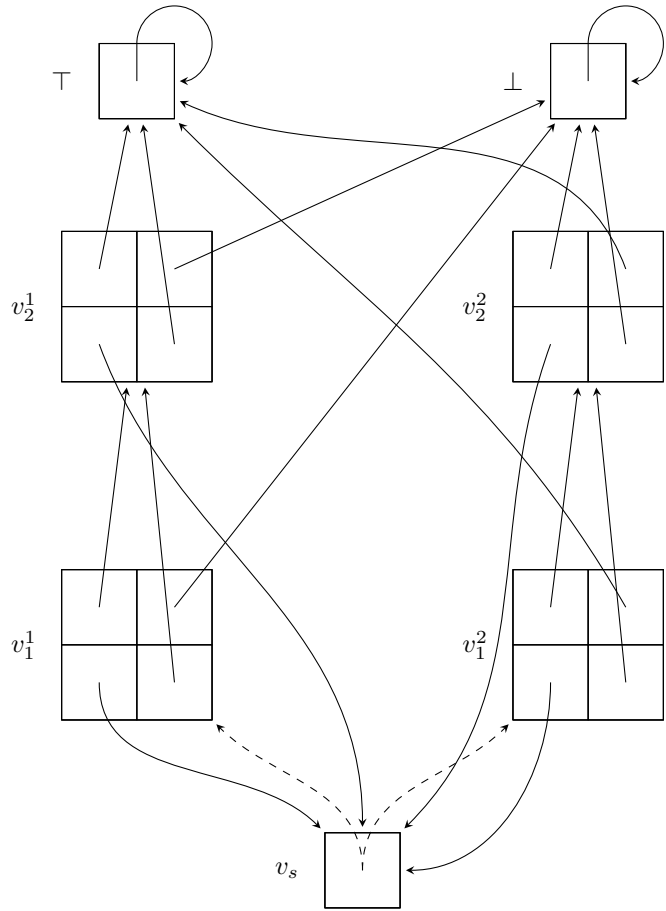

\smallskip\noindent{\bf The game Purgatory.} We will also use the game {\em Purgatory} as defined by~\cite{HIM11} (and also in~\cite{HKM09} for the case of $m=2$). Purgatory is similar to the Purgatory Duel and hence the similarity in names. Purgatory is also defined on a pair of parameters $(n,m)$. The game consists of $N=n+2$ states, namely, $\{v_1,v_2,\dots, v_n,\top,\bot\}$ and each state is deterministic. 
To simplify the definition of the game, let $v_{n+1}=\top$.
For each %%$i$ and 
$j\in\{1,\dots,n\}$, the state $v_j$ is such that $A_{v_j}^1=A_{v_j}^2=\{1,2,\dots, m\}$ and for each $a_1,a_2$ we have that
\[\delta(v_j,a_1,a_2)=\begin{cases}
v_1 & \text{if }a_1>a_2\\
\bot & \text{if }a_1<a_2\\
v_{j+1} & \text{if }a_1=a_2\enspace .
\end{cases}
\]
The states $\top$ and $\bot$ are absorbing.
Furthermore, $S^1=\{\top\}$. There is an illustration of Purgatory with $m=n=2$ in Figure~\ref{fig:p}.

\begin{figure}
\center
\begin{tikzpicture}[node distance=3cm,-{stealth},shorten >=2pt]
\ma{top}[$\top$]{1}{1};

\ma[shift={($(top)+(0cm,-3cm)$)}]{v2}[$v_2$]{2}{2};
\ma[shift={($(v2)+(0cm,-4.5cm)$)}]{v1}[$v_1$]{2}{2};

\ma[shift={($(v2)+(3cm,0)$)}]{bot}[$\bot$]{1}{1};
\nloop{top};
\nloop{bot};

\draw (v2-2-2.center) -- (top);
\draw (v2-1-1.center) -- (top);
\draw (v2-1-2.center) to[out=0,in=160] (bot);
\draw (v2-2-1.center) to[out=-100,in=90] ($(v1-1-1.north west)!0.5!(v1-1-1.north)$);

\draw (v1-2-2.center) -- (v21);
\draw (v1-1-1.center) -- (v21);
\draw (v1-1-2.center) -- (bot);
\nloop{v1-2-1}[xscale=-1,yscale=-1]

\end{tikzpicture}
\caption{An illustration of Purgatory with $m=n=2$.\label{fig:p}}
\end{figure}
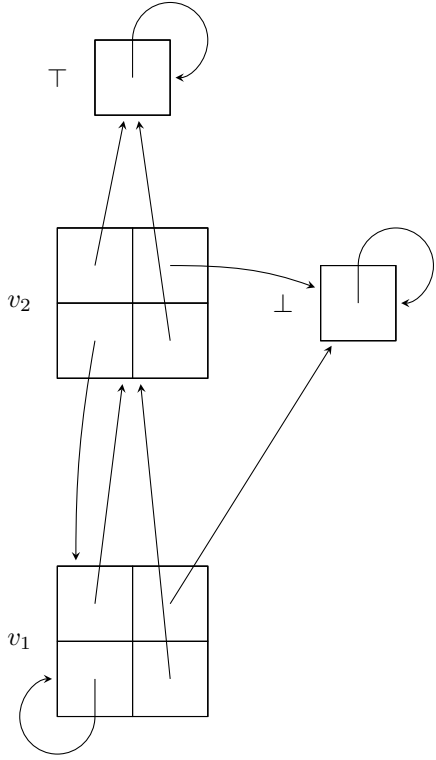

\subsection{Analysis of matrix games}
In this section we present some refined analysis of some 
simple matrix games, which we use in the later sections to find optimal strategies 
for the players and the values of the states in the Purgatory Duel.
\begin{definition}
Given a positive integer $m$ and reals $x$, $\y$ and $\z$, let $M^{x,\y,\z,m}$ be the $(m\times m)$-matrix with  $x$ below the diagonal, $\y$ in the diagonal and $\z$ above the diagonal, i.e., 
\[M^{x,\y,\z,m}=\M{x}{\y}{\z} \enspace .
\]
\end{definition}

We first explain the significance of the matrix game $M^{x,\y,\z,m}$ in relation to
Purgatory Duel. 
Consider the Purgatory Duel defined on parameters $(n,m)$, for some $n$.
We will later establish that for any $j$, 
let $v$ (resp., $v'$) be state $v^1_j$ (resp., $v^2_j$) of the Purgatory Duel, then we have that $A^v=M^{0,\val(v^1_{j+1}),\val(v_s),m}$
(resp., $A^{v'}=M^{1,\val(v^2_{j+1}),\val(v_s),m}$). 
In this section we show that for $0<\z<\y$ we have that $M=M^{0,\y,\z,m}$ is such that $\val(M)>\z$ 
and each optimal strategy for either player is totally mixed. 
Similarly, for $1>\z'>\y'$ we show that  $M'=M^{1,\y',\z',m}$ is such that $\val(M')<\z$ 
and each optimal strategy for either player is totally mixed. 
We also compute the value and the patience of each optimal strategy in the matrix game 
$M^{0,\frac{1}{2}+\eps,\frac{1}{2},m}$ 
(since we will establish in the next section, using the results of this section, 
that $\val(v_s)=\frac{1}{2}$ and $\val(v^1_j)>\val(s)$ for all $j$).

\begin{lemma}\label{lem:mat_val_above}
For all positive integers $m$ and reals $\y$ and $\z$ such that $0<\z<\y$, the matrix game $M=M^{0,\y,\z,m}$ has value strictly above $\z$.
\end{lemma}

\begin{proof}
Let $\eps>0$ be some number to be defined later.
Consider the probability distribution $\sigma_1^\eps$ given by \[
\sigma_1^\eps(a)=\begin{cases}
\eps^{a-1}-\eps^{a} & \text{if }1\leq a\leq m-1\\
\eps^{m-1} & \text{if }a=m \enspace .
\end{cases}
\]
If player~2 plays column~$a$ against $\sigma_1$, for $a\leq m-1$, then the payoff $u(M,\sigma_1,a)$ is $\y\cdot (\eps^{a-1}-\eps^{a})+\y\cdot (1-\eps^{a-1})$; 
and if player~2 plays column~$m$, then the payoff $u(M,\sigma_1,m)$ is $\y\cdot (\eps^{m-1})+\z\cdot (1-\eps^{m-1})$. 
For any $\eps$ such that $\y\cdot (1-\eps)>\z$, the payoff is strictly greater than $\z$ implying that the value of $M$ is strictly greater than $\z$.
\end{proof}

\begin{lemma}\label{lem:mat_p1_totally_mixed}
For all positive integers $m$ and reals $\y$ and $\z$ such that $0<\z<\y$, each optimal strategy for player~1 in the matrix game $M^{0,\y,\z,m}$ is totally mixed.
\end{lemma}
\begin{proof}
Consider some strategy $\sigma_1$ for player~1 in $M^{0,\y,\z,m}$ which is not totally mixed. Thus there exists some row $a$, where $\sigma_1(a)=0$. Consider the pure strategy $\sigma_2$ that plays column $a$ with probability~1. Playing $\sigma_1$ against $\sigma_2$ ensures that each outcome is either $\z$ or $0$, i.e., the payoff is at most $z$ which is strictly less than the value by Lemma~\ref{lem:mat_val_above}.
\end{proof}

\begin{lemma}\label{lem:mat_p2_totally_mixed}
For all positive integers $m$ and reals $\y$ and $\z$ such that $0<\z<\y$, each optimal strategy for player~2 in the matrix game $M=M^{0,\y,\z,m}$ is totally mixed.
\end{lemma}
\begin{proof}
Given a strategy $\sigma_1$ for player~1 and two rows $a'$ and $a''$, let the strategy $\sigma_1[a'\rightarrow a'']$ be the strategy where the probability mass on $a'$ is moved to $a''$, i.e.,
\[
\sigma_1[a'\rightarrow a''](a)=\begin{cases}
\sigma_1(a) & \text{if }a'\neq a\neq a''\\
0 & \text{if }a=a'\\
\sigma_1(a')+\sigma_1(a'') & \text{if }a=a''\enspace .
\end{cases}
\]

Consider some optimal totally mixed strategy $\sigma_1$ for player~1, which exists by Lemma~\ref{lem:mat_p1_totally_mixed} and let $v$ be the value of $M$. Consider some strategy $\sigma_2$ for player~2 such that $u(M,\sigma_1,\sigma_2)=v$, but $\sigma_2$ is not totally mixed. We will argue that $\sigma_2$ is not optimal. This shows that any optimal strategy $\sigma_2^*$ is totally mixed, since any optimal strategy $\sigma_2$ is such that $u(M,\sigma_1,\sigma_2)=v$.

 Let $b'$ be the first column such that $\sigma_2'(b)=0$. 
There are two cases, either $b'=1$ or $b'>1$. If $b'=1$ let $b''$ be the first action such that $\sigma_2(b'')>0$. Let $\sigma_1'=\sigma_1[b'\rightarrow b'']$. 
The payoff $u(M,\sigma_1',\sigma_2)$ of playing  $\sigma_1'$ against $\sigma_2$ is strictly more than the payoff $u(M,\sigma_1,\sigma_2)$ of playing $\sigma_1$ against $\sigma_2$. This is because the payoff $u(M,\sigma_1',b'')$ is such that  
\begin{align*}
u(M,\sigma_1',b'')&=\sigma_1'(b'')\cdot \y +\z\cdot \sum_{a=1}^{b''-1}\sigma_1'(a)\\
&=\sigma_1'(b'')\cdot \y +\z\cdot \sum_{a=2}^{b''-1}\sigma_1'(a)\\
&=(\sigma_1(b'')+\sigma_1(1))\cdot \y +\z\cdot \sum_{a=2}^{b''-1}\sigma_1'(a) \\
&> \sigma_1(b'')\cdot \y +\z\cdot \sum_{a=1}^{b''-1}\sigma_1(a) \\
&= u(M,\sigma_1,b'')\enspace ,
\end{align*}
where the second equality comes from that $\sigma_1'(1)=0$. The inequality comes from that $\y>\z$.
 Also, the payoff $u(M,\sigma_1',b)$, for $b>b''$ is such that  
\begin{align*}
u(M,\sigma_1',b)&=\sigma_1'(b)\cdot \y +\z\cdot \sum_{a=1}^{b-1}\sigma_1'(a)\\
&=\sigma_1(b)\cdot \y +\z\cdot \sum_{a=1}^{b-1}\sigma_1(a)=u(M,\sigma_1,b)\enspace ,
\end{align*} 
because $\sigma_1'$ is not different from $\sigma_1$ on those actions.
We can then find the payoff $u(M,\sigma_1',\sigma_2)$ as follows 
\begin{align*}
& u(M,\sigma_1',\sigma_2)=\sum_{b=1}^m \sigma_2(b)\cdot u(M,\sigma_1',b) \\
& =\sum_{b=b''}^m \sigma_2(b)\cdot u(M,\sigma_1',b) \\
& =\sigma_2(b'')\cdot u(M,\sigma_1',b'') + \sum_{b=b''+1}^m \sigma_2(b)\cdot u(M,\sigma_1',b) \\
& >\sigma_2(b'')\cdot u(M,\sigma_1,b'') + \sum_{b=b''+1}^m \sigma_2(b)\cdot u(M,\sigma_1,b) \\
&=u(M,\sigma_1,\sigma_2) \enspace ,
\end{align*} 
where the second equality comes from that  $b''$ is the first action $\sigma_2$ plays with positive probability. 
Since the payoff $u(M,\sigma_1,\sigma_2)$ is the value, by definition of $\sigma_2$, and the payoff $u(M,\sigma_1',\sigma_2)$ is strictly more, the strategy $\sigma_2$ cannot be optimal.
This completes the case where $b'=1$.

The case where $b'\neq 1$ follows similarly but considers $\sigma_1''=\sigma_1[b'\rightarrow 1]$ instead of $\sigma_1'$.
\end{proof}

\begin{lemma}\label{lem:precise}
For all positive integers $m$ and $0<\eps\leq \frac{1}{2}$, 
the matrix game $M=M^{0,\frac{1}{2}+\eps,\frac{1}{2},m}$ has the following properties:
\begin{itemize}
\item{\em \bf Property 1.} The patience of any optimal strategy is (i)~at least $(2\eps)^{-m+1}$ and (ii)~decreasing in $\eps$.
\item{\em \bf Property 2.} The value is (i)~at most $\frac{1}{2}+\eps\cdot (2\eps)^{m-1}$ and (ii)~increasing in $\eps$.
\item{\em \bf Property 3.} Any optimal strategy $\sigma_1$ for player~1 (resp., $\sigma_2$ for player~2) is such that $\sigma_1(1)>\frac{1}{2}$ (resp., $\sigma_2(m)>\frac{1}{2}$).
\item{\em \bf Property 4.} For $\eps=\frac{1}{2}$, the value is $\val(M)=\frac{1}{2}+\frac{1}{2^{m+1}-2}$  and the patience of any optimal strategy is $2^{m}-1$.

\end{itemize}
\end{lemma}

\begin{proof}
Let $\sigma_i$ be an optimal strategy for player~$i$ in $M$, for each $i$.
By Lemma~\ref{lem:mat_p1_totally_mixed} and Lemma~\ref{lem:mat_p2_totally_mixed} the strategy $\sigma_i$ is totally mixed for each $i$. 
We can therefore consider the vector $\overline{v}$. Recall that $\overline{v}_j=u(M,\sigma_1,j)$ and that for each $j$ such that $\sigma_2(j)>0$ we have that $\overline{v}_j=\val(M)$. Hence, since $\sigma_2$ is totally mixed, all entries of $\overline{M}$ are $\val(M)$.  
For any row $a'<m$, that $\overline{v}_{a'}=\overline{v}_{a'+1}$ implies that 
\begin{align*}
(\frac{1}{2}+\eps)\cdot \sigma_1(a') +\frac{1}{2}\cdot \sum_{a=1}^{a'-1}\sigma_1(a)&
\\=(\frac{1}{2}+\eps)\cdot \sigma_1(a'+1) +\frac{1}{2}\cdot \sum_{a=1}^{a'}\sigma_1(a) & \Rightarrow \\
\eps\cdot \sigma_1(a') =(\frac{1}{2}+\eps)\cdot \sigma_1(a'+1) & \Rightarrow\\
\sigma_1(a') =\frac{\frac{1}{2}+\eps}{\eps}\cdot \sigma_1(a'+1) & \enspace ,
\end{align*}
indicating that $\sigma_1(a')>\sigma_1(a'+1)$ and thus the patience is $1/\sigma_1(m)$.
Also, since $\sigma_1$ is a probability distribution 
\begin{align*}
1&=\sum_{a=1}^m\sigma_1(a)\\
&=\sigma_1(m)\cdot \sum_{a=1}^m\left(\frac{\frac{1}{2}+\eps}{\eps}\right)^{m-a}
\end{align*}
We then get that \[
\sigma_1(m)=\frac{1}{\sum_{a=1}^m\left(\frac{\frac{1}{2}+\eps}{\eps}\right)^{m-a}}\]
We have that $\frac{\frac{1}{2}+\eps}{\eps}=1+\frac{1}{2\eps}$ is decreasing in $\eps$.
This indicates that $\sigma_1(m)$ is increasing in $\eps$ and thus the patience is decreasing in $\eps$. This shows (ii) of Property~1 for player~1.
We also have that $\val(M)=\overline{v}_{m}$ indicating that \begin{align*}
\val(M)&=\overline{v}_{m} \\
&=\sigma_1(m)\cdot(\frac{1}{2}+\eps)+\frac{1}{2}\cdot\sum_{a=1}^{m-1}\sigma_1(a)\\
&=\eps\cdot \sigma_1(m)+\frac{1}{2}
\end{align*}
and thus, the value is increasing in $\eps$ (because $\eps$ and $\sigma_1(m)$ both are). This shows (ii) of Property~2.

Also, we get that, \begin{align*}
\sigma_1(m)&=\frac{1}{\sum_{a=1}^m\left(\frac{\frac{1}{2}+\eps}{\eps}\right)^{m-a}} \\
&=\frac{\eps^{m-1}}{
\sum_{a=1}^m (\frac{1}{2}+\eps)^{m-a}\cdot \eps^{a-1}} \\
&=\frac{\eps^{m-1}}{
\left(\frac{1}{2}\right)^{m-1}+\eps\cdot p(\eps)} \enspace ,
\end{align*}
where $p$ is some polynomial of degree $m-1$ in which all terms have a positive sign ($p$ is found by multiplying out $\sum_{a=1}^m (\frac{1}{2}+\eps)^{m-a}\cdot \eps^{a-1}$). 
Hence, we have that $\sigma_1(m)$ is at most \[\sigma_1(m)=\frac{\eps^{m-1}}{
\left(\frac{1}{2}\right)^{m-1}+\eps\cdot p(\eps)}<(2\eps)^{m-1}\enspace .\]
Thus, the patience is at least $(2\eps)^{-m+1}$. This shows (i) of Property~1 for player~1.
Using that $\val(M)=\eps\cdot \sigma_1(m)+\frac{1}{2}$ from above, we get that $\val(M)<\frac{1}{2}+\eps\cdot (2\eps)^{m-1}$. This shows (i) of Property~2.

Furthermore, we can also consider the vector $\overline{v'}$ such that $\overline{v'}_j=u(M,j,\sigma_2)$ for all $j$ (which like $\overline{v}$ has all entries equal to $\val(M)$). Since the expression, when $\sigma_2$ is taken to be an unknown vector, for the $j$'th entry of $\overline{v'}$ is the same as for the $m+1-j$'th entry of $\overline{v}$, when $\sigma_1$ is taken to be an unknown vector, we see that $\sigma_1(a)=\sigma_2(m+1-a)$, implying that the patience of player~2's optimal strategies is also at least $(2\eps)^{-m+1}$ and that it is decreasing in $\eps$. This shows Property~1 for player~2. 

Observe that since the value is above $\frac{1}{2}$, by Lemma~\ref{lem:mat_val_above}, we have that $\sigma_1(1)>\frac{1}{2}$ (because otherwise, if player~2 plays $1$ with probability~1, the payoff will not be above $\frac{1}{2}$) and thus also $\sigma_2(m)>\frac{1}{2}$. This shows Property~3.

Also, for $\eps=\frac{1}{2}$ we see that \begin{align*}
\sigma_1(m)&=\frac{1}{\sum_{a=1}^m\left(\frac{\frac{1}{2}+\eps}{\eps}\right)^{m-a}}  \\
&=\frac{1}{\sum_{a=1}^m 2^{m-a}}  \\
&=\frac{1}{2^{m}-1} \enspace .
\end{align*}
Similarly to above, we also get that $\sigma_2(m)=\frac{1}{2^{m}-1}$ and that $\val(M)=\frac{1}{2}+\frac{1}{2^{m+1}-2}$. This shows Property~4 and completes the proof.
\end{proof}

\begin{lemma}\label{lem:mat_wh}
Given a positive integer $m$ and reals $\y$ and $\z$ such that $1>\z>\y$, the matrix game $M=M^{1,\y,\z,m}$ has the following properties:
\begin{itemize}
\item The value $\val(M)<\z$.
\item Each optimal strategy $\sigma_i$ for player~$i$ is such that there exists an optimal strategy $\sigma_{\widehat{i}}$ for player~$\widehat{i}$ in $M^{0,1-\y,1-\z,m}$ where $\sigma_{i}(j)=\sigma_{\widehat{i}}(m-j+1)$.
\end{itemize}
\end{lemma}

\begin{proof}
Let a positive integer $m$ and reals $\y$ and $\z$ such that $1>\z>\y$ be given.
Consider $M$ and let $v$ be the value of $M$. Exchange the roles of the players by exchanging the rows and columns and multiply the matrix by~$-1$. We get the matrix \[M^1=
\M{-\z}{-\y}{-1} \enspace .
\]
We then have that each optimal strategy $\sigma_1$ in $M$ is an optimal strategy for player~2 in $M^1$ and similarly, each optimal strategy $\sigma_2$ for player~2 in $M$ is an optimal strategy for player~1 in $M^1$ (and vice versa).  Also, the value $v_1$ of $M^1$ is $v_1:=-v$.

Let $M^2$ be the matrix where $M^2_{a,b}=M^1_{m+1-a,m+1-b}$, i.e., 
\[M^2=\M{-1}{-\y}{-\z} \enspace .
\]
For each $i$, and for any optimal strategy $\sigma_i$ for player~$i$ in $M^1$ the strategy $\sigma_i'$ is optimal for player~$i$ in $M^2$, where $\sigma_i'(a)=\sigma_i(m+1-a)$ for each $a$ (and vice versa). Also, the value $v_2$ of $M^2$ is $v_2:=v_1=-v$. 

Next, let $M^3$ be the matrix $M^2$ where we add 1 to each entry, i.e.,
\[M^3=\M{0}{1-\y}{1-\z} \enspace .
\]
For each $i$, it is clear that an optimal strategy in $\sigma_i$ for player~$i$ in $M^2$ is an optimal strategy for player~$i$ in $M^3$ and that the value $v_3$ is $v_3:=1+v_2=1-v$. Also, we see that $M^3=M^{0,1-\y,1-\z,m}$ and that $0<1-\z<1-\y$.

We then get that $1-v>1-\z$ from Lemma~\ref{lem:mat_val_above} and thus $v<\z$. 
\end{proof}

\subsection{The patience of optimal strategies} %%%in the Purgatory Duel}
In this section we present an approximation of the values of the states 
and the patience of the optimal strategies in the Purgatory Duel. 
%%In this section, the only concurrent reachability game we consider is the Purgatory Duel. 
We first show that the values of the states (besides $\top$ and $\bot$) are strictly between 0 and 1. 

\begin{lemma}\label{lem:pos val}
Each state \[v\in \{v_1^1,v_2^1,\dots, v_n^1,v_1^2,v_2^2,\dots, v_2^1, v_s\}\] is such that $\val(v)\in [\frac{1}{m^{n+2}},1-\frac{1}{m^{n+2}}]$
\end{lemma}
\begin{proof}
Fix $v\in \{v_1^1,v_2^1,\dots, v_n^1,v_1^2,v_2^2,\dots, v_2^1, v_s\}$.
The fact that $\val(v)\geq \frac{1}{m^{n+2}}$ follows from that if player~1 plays uniformly at random all actions in every state $v^i_j$ for all $i,j$, 
then against all strategies for player~2 there is a probability of at least $\frac{1}{m}$ to go 
(1)~from $v_j^1$ to $v_{j+1}^1$, for all $j$; and 
(2)~from $v_s$ to $v_1^1$; and 
(3)~from $v_j^2$ to $v_s$, for all $j$. 
By following such steps for at most $n+2$ steps, the state $v_{n+1}^1=\top$ is reached.
Similarly that  $\val(v)\leq 1-\frac{1}{m^{n+2}}$ follows from player~2 playing uniformly at random all actions 
in every state $v^i_j$ for all $i,j$ (and using that $\top$ cannot be reached from $\bot$).
\end{proof}

Next we show that every optimal stationary strategy for player~2 must be totally mixed.

\begin{lemma}\label{lem:opt_totally_mixed}
Let $\sigma_2$ be an optimal stationary strategy for player~2. The distribution $\sigma_2(v_j^i)$ is totally mixed and $\val(v^1_j)>\val(v_s)>\val(v^2_j)$, for all $i,j$.
\end{lemma}
\begin{proof}
Let $v=v_j^i$ for some $i,j$. We will use that $\val(v)=\val(A^v)$.
For $i=1$ we have that $A^v=M^{0,\val(v_{j+1}^1),\val(v_s),m}$ and for $i=2$ we have that $A^v=M^{1,\val(v_{j+1}^2),\val(v_s),m}$. 

Consider first $i=1$.
We will show using induction in $j$ (with base case $j=n$ and proceeding downwards), that $\val(v_{j}^1)>\val(v_s)$ and that the distribution $\sigma_2(v_j^1)$ is totally mixed.

{\bf Base case, $j=n$:}
We have that $A^v=M^{0,1,\val(v_s),m}$. By Lemma~\ref{lem:pos val} we have that $1>\val(v_s)>0$ and thus, that $\val(v)>\val(v_s)$ follows from Lemma~\ref{lem:mat_val_above}. That $\sigma_2(v)$ is totally mixed follows from Lemma~\ref{lem:mat_p2_totally_mixed}.

{\bf Induction case, $j\leq n-1$:}
We have that $A^v=M^{0,\val(v_{j+1}^1),\val(v_s),m}$. By Lemma~\ref{lem:pos val} we have that $\val(v_s)>0$ and by induction we have that $\val(v_{j+1}^1)>\val(v_s)$ and thus, that $\val(v)>\val(v_s)$ follows from Lemma~\ref{lem:mat_val_above}. That $\sigma_2(v)$ is totally mixed follows from Lemma~\ref{lem:mat_p2_totally_mixed}.

The argument for $i=2$ is similar but uses Lemma~\ref{lem:mat_wh} together with Lemma~\ref{lem:mat_p1_totally_mixed}, instead of Lemma~\ref{lem:mat_p2_totally_mixed} and Lemma~\ref{lem:mat_val_above}.
\end{proof}

Next, we show that if either player follows a stationary strategy that is totally mixed on at least one side (that is, if there is an $i'$, such that for each $j$ the stationary strategy plays totally mixed in $v^{i'}_j$), then eventually either $\top$ or $\bot$ is reached with probability~1.

\begin{lemma}\label{lem:partial_mixed}
For any $i$ and $i'$, let $\sigma_i$ be a stationary strategy for player~$i$, such that $\sigma_i(v^{i'}_j)$ is totally mixed for all $j$. Let $\sigma_{\widehat{i}}$ be some positional strategy for the other player. 
Then, each closed recurrent set in the Markov chain defined by the game and $\sigma_i$ and $\sigma_{\widehat{i}}$ consists of only the state $\top$ or only the state $\bot$.  
\end{lemma}
\begin{proof}
In the Markov chain defined by the game and $\sigma_i$ and $\sigma_{\widehat{i}}$, we have that there are at most two closed recurrent sets, namely, the one consisting of only $\top$ and the one consisting of only $\bot$. The reasoning is as follows:
If either $\top$ or $\bot$ is reached, then the respective state will not be left. Also, for each $j$, since $\sigma_i$ is totally mixed there is a positive probability to go to either $v^{i'}_0$ or $v^{i'}_{j+1}$ from $v^{i'}_j$ (the remaining probability goes to $v_s$). The probability to go from $v_s$ to $v^{i'}_{1}$ in one step is $\frac{1}{2}$. Also if neither $\top$ nor $\bot$ has been reached, then $v_s$ is visited after at most $n+1$ steps.
Hence, in every $n+1$ steps there is a positive probability that in the next $n+1$ steps either $\top$ or $\bot$ is reached (i.e., from $v_s$ there is a positive probability that the next states are either (i)~$v_{1}^{i'},\dots, v_{j}^{i'},v_{0}^{i'}$; or (ii)~$v_{1}^{i'},\dots, v_{n}^{i'},v_{n+1}^{i'}$). This shows that eventually either $\top$ or $\bot$ is reached with probability~1.
\end{proof}

\begin{remark}
Note that Lemma~\ref{lem:partial_mixed} only requires that the strategy $\sigma_i$ is totally mixed on one ``side'' of the Purgatory Duel. For the purpose of this section, we do not use that it only requires one side to be totally mixed, since we only use the result for optimal strategies for player~2, which are totally mixed by Lemma~\ref{lem:opt_totally_mixed}. However the lemma will be reused in the next section, where the one sidedness property will be useful.  
\end{remark}

The following definition basically \emph{``mirrors''} a strategy $\sigma_i$ for player~$i$, for each $i$ and gives it to the other player. 
We show (in Lemma~\ref{lem:mirror}) that if $\sigma_2$ is optimal for player~2, then the mirror strategy is optimal for player~1. 
We also show that if $\sigma_2$ is an $\eps$-optimal strategy for player~2, for $0<\eps<\frac{1}{3}$, then so is the mirror strategy 
for player~1 (in Lemma~\ref{thm:epsilon_mirror}). 
\begin{definition}[Mirror strategy]
Given a stationary strategy $\sigma_i$ for player~$i$, for either $i$, let the mirror strategy $\sigma_{\widehat{i}}^{\sigma_i}$ for player~$\widehat{i}$ 
be the stationary strategy where $\sigma_{\widehat{i}}^{\sigma_i}(v_j^{\widehat{i'}})=\sigma_i(v_j^{i'})$ for each $i'$ and $j$.
\end{definition}

We next show that player~1 has optimal stationary strategies in the Purgatory Duel 
and give expressions for the values of states.

\begin{lemma}\label{lem:mirror}
Let $\sigma_2$ be some optimal stationary strategy for player~2. 
Then the mirror strategy $\sigma_1^{\sigma_2}$ is optimal for player~1. 
We have $\val(v_s)=\frac{1}{2}$ and $\val(v^i_j)=1-\val(v^{\widehat{i}}_j)$, 
for all $i,j$.
\end{lemma}
\begin{proof}
Consider some optimal stationary strategy $\sigma_2$ for player~2. It is thus totally mixed, by Lemma~\ref{lem:opt_totally_mixed}. Let $\sigma_1=\sigma_1^{\sigma_2}$ 
be the mirror strategy for player~1. 

Playing $\sigma_1$ against $\sigma_2$ and starting in $v_s$ we see that we have probability~$\frac{1}{2}$ to reach $\top$ and probability~$\frac{1}{2}$ to reach $\bot$, by symmetry and Lemma~\ref{lem:partial_mixed}. This shows that the value is at least $\frac{1}{2}$ because $\sigma_2$ is optimal. On the other hand, consider some stationary strategy $\sigma_1'$ for player~1, and the mirror strategy  $\sigma_2'=\sigma_2^{\sigma_1'}$ for player~2. If player~2 plays $\sigma_2'$ against $\sigma_1'$, then the probability to eventually reach $\bot$ is equal to the probability to eventually reach $\top$ and then there is some probability~$p$ (perhaps~0) that neither will be reached. The payoff $u(v_s,\sigma_1',\sigma_2',1)$ is then $\frac{1-p}{2}\leq \frac{1}{2}$. This shows that player~1 cannot ensure value strictly more than $\frac{1}{2}$, which is then the value of $v_s$.
Finally, we argue that $\sigma_1$ is optimal. If not, then consider $\sigma_2^*$ such that 
$u(v_s,\sigma_1,\sigma_2^*,1) <1/2$, and then the mirror strategy $\sigma_1^*=\sigma_1^{\sigma_2^*}$ ensures
that $u(v_s,\sigma_1^*,\sigma_2,1)>1/2$ contradicting optimality of $\sigma_2$.
%% and indicates that $\sigma_1$ is optimal if the play starts in $v_s$.

Similarly, for any $i,j$, playing $\sigma_1$ against $\sigma_2$ and starting in $v_j^i$ we see that the probability with which we reach $\top$ is equal to the probability of reaching $\bot$ starting in $v_j^{\widehat{i}}$ and vice versa, by symmetry. 
Also, by Lemma~\ref{lem:partial_mixed} the probability to eventually reach either $\bot$ or $\top$ is~1. 
Observe that the probability to reach $\bot$ starting in $v_j^{\widehat{i}}$ is at least $1-\val(v_j^{\widehat{i}})$, by optimality of $\sigma_2$ and that with probability~1 either $\bot$ is reached or $\top$ is reached. Also, again because $\sigma_2$ is optimal, the probability to reach $\top$ starting in $v_j^i$ is at most $\val(v_j^i)$.
This shows that $\val(v_j^i)\geq 1-\val(v_j^{\widehat{i}})$. Using an argument like the one above, we obtain that  $\val(v_j^i)= 1-\val(v_j^{\widehat{i}})$ and that $\sigma_1$ is optimal if the play starts in $v_j^i$.
\end{proof}

Finally, we give an approximation of the values of states in the Purgatory Duel and a lower bound on the patience of any optimal strategy of $2^{(m-1)^2 m^{n-2}}$.

\begin{theorem}
For each $j$ in $\{1,\dots,n\}$, the value of state $v^1_j$ in the Purgatory Duel is less than $\frac{1}{2}+ 2^{(1-m)\cdot m^{n-j}-1}$ and for any optimal stationary strategy $\sigma_i$ for either player~$i$, the patience of $\sigma_i(v^1_j)$ is at least $2^{(m-1)^2 m^{n-j-1}}$.
\end{theorem}

\begin{proof}
Consider some optimal stationary strategy $\sigma_2$ for player~2.
We will show using induction in $j$ that $\val(v^1_j)$ is less than $\frac{1}{2}+ 2^{(1-m)\cdot m^{n-j}-1}$ and that the patience of $\sigma_2(v^1_j)$ is at least $2^{(m-1)^2 m^{n-j-1}}$. Note that using Lemma~\ref{lem:mirror}, a similar result holds for optimal strategies for player~1. Let $v=v^i_j$.

{\bf Base case, $j=n$:} We see that the matrix $A^v$ is $M^{0,1,\frac{1}{2},m}$ and thus, by Lemma~\ref{lem:precise} (Property~1 and~2) we have that the value
\begin{align*}
\val(v)&=\val(A^v)\\
&=\frac{1}{2}+\frac{1}{2^{m+1}-2}\\
&<\frac{1}{2}+2^{-m}\\
&=\frac{1}{2}+2^{(1-m)\cdot m^0-1} \enspace ,
\end{align*}
 and $\sigma_2(v)$ has patience $2^{m}-1>2^{(m-1)^2\cdot m^{-1}}$.

{\bf Induction case, $j\leq n-1$:} We see that the matrix $A^v$ is $M=M^{0,\val(v^i_{j+1}),\frac{1}{2},m}$. By induction we have that $\val(v^i_{j+1})<\frac{1}{2}+ 2^{(1-m)\cdot m^{n-j-1}-1}$. Let  $\eps=2^{(1-m)\cdot m^{n-j-1}-1}$ and consider $M'=M^{0,\frac{1}{2}+\eps,\frac{1}{2},m}$. By Lemma~\ref{lem:precise} (Property~1 and~2) we get that $\val(M')\geq \val(M)$ and that the patience of $M'$ is smaller than the one for $M$. Also, we get that
\begin{align*}
\val(M')&<\frac{1}{2}+\eps\cdot (2\eps)^{m-1}\\
&=\frac{1}{2}+2^{m-1}\cdot 2^{(1-m)\cdot m^{n-j}-m}\\
&=\frac{1}{2}+2^{(1-m)\cdot m^{n-j}-1}\enspace ,
\end{align*}
and that the patience of $M'$ (and thus $M$) is at least 
\begin{align*}
(2\eps)^{-m+1}&=2^{m-1}\cdot 2^{(1-m)^2 \cdot m^{n-j-1}-m+1}\\
&=2^{(1-m)^2 \cdot m^{n-j-1}}\enspace .
\end{align*}
This completes the proof.
\end{proof}

\begin{remark}
It can be seen using induction that the value of each state in the Purgatory Duel is a rational number. First notice that $v^1_n$ and $v^2_n$ are the value of a matrix game with numbers in $\{0,\frac{1}{2},1\}$ and hence are rational. Similarly, using induction in $i$, we see that for $j\in \{1,2\}$ the number $v^j_i$ is rational, since it is the value of a matrix game with numbers in $\{v^j_0,\frac{1}{2},v^{j}_{i+1}\}$ (recall that $v^1_0=0$ and $v^2_0=1$).
\end{remark}

\subsection{The patience of $\eps$-optimal strategies} %in the Purgatory Duel}
\label{sec:eps_patience}
In this section we consider the patience of $\eps$-optimal strategies for $0<\eps<\frac{1}{3}$. 
First we argue that each such strategy for player~2 is totally mixed on one side.

\begin{lemma}
\label{lem:eps_opt_partial_mixed}
For all $0<\eps<\frac{1}{2}$, each $\eps$-optimal stationary strategy $\sigma_2$ for player~2 is such that $\sigma_2(v^2_j)$ is totally mixed, for all $j$.
\end{lemma}
\begin{proof}
Fix $0<\eps<\frac{1}{2}$ and fix some stationary strategy $\sigma_2$ such that there exists $j$ such that $\sigma_2(v^2_j)$ is not totally mixed. We will show that $\sigma_2$ is not $\eps$-optimal.

Let $\eta$ be such that $0<\eta<\frac{1}{2}-\eps$. Let $a$ be an action such that $\sigma_2(v^2_j)(a)=0$.
Let $\sigma_1^\eta$ be an $\eta$-optimal strategy in {\em Purgatory} (not the Purgatory Duel) (with the same parameters $n$ and $m$).
Let $\sigma_1$ be the strategy such that (i)~$\sigma_1(v_{j'}^2)(1)=1$ for each $j'$; and (ii)~$\sigma_1(v_{j}^2)(a)=1$; and (iii)~$\sigma_1(v_{j}^1)=\sigma_1^\eta(v_j)$.
Consider a play starting in $v_s$. Whenever the play is in state $v^2_{j'}$, for some $j'\neq j$ in each step there is a probability of either going back to $v_s$ or going to $v^2_{j'+1}$. Thus, the play either reaches $v_j^2$ or has gone back to $v_s$. If it reaches $v_j^2$, then the next state is either $v_s$ or $\top$ (i.e., $v_{j+1}^2$ cannot be reached). If the play is in $v_1^1$, then there is a positive probability to reach $\top$ before going back to $v_s$, which is at least $\frac{1-\eta}{\eta}$ times the probability to reach $\bot$ before going back to $v_s$, since $\sigma_1$ follows an $\eta$-optimal strategy in Purgatory. Hence, the probability to eventually reach $\top$ is at least $1-\eta>\frac{1}{2}+\eps$ and thus $\sigma_2$ is not $\eps$-optimal, since the value of $v_s$ is $\frac{1}{2}$ by Lemma~\ref{lem:pos val}. 
\end{proof}

We now show that if we mirror an $\eps$-optimal strategy, then we get an $\eps$-optimal strategy.

\begin{lemma}\label{thm:epsilon_mirror}
For all $0<\eps<\frac{1}{3}$, each $\eps$-optimal stationary strategy $\sigma_2$ for player~2 in the Purgatory Duel, is such that the mirror strategy 
$\sigma_1^{\sigma_2}$ is $\eps$-optimal for player~1.
\end{lemma}
\begin{proof}
Fix $0<\eps<\frac{1}{3}$ and let $\sigma_2$ be some $\eps$-optimal stationary strategy for player~2. Also, let $\sigma_1=\sigma_1^{\sigma_2}$ be the mirror strategy.

By Lemma~\ref{lem:eps_opt_partial_mixed} the strategy $\sigma_2$ is such that $\sigma_2(v^2_j)$ is totally mixed, for all $j$. We can then apply Lemma~\ref{lem:partial_mixed} and get that either $\top$ or $\bot$ is reached with probability~1. Hence, since $\sigma_2$ is $\eps$-optimal we reach $\bot$ with probability at least $1-\val(v)-\eps$ starting in $v$ against all strategies for player~1, for each $v$.
It is clear that any play $P$ of $\sigma_2$ against any given strategy $\sigma_1'$ for player~1 starting in $v$ corresponds, by symmetry, to a play $P'$ of $\sigma_2^{\sigma_1'}$ against $\sigma_1$ starting in $f(v)$, where \[f(v)=\begin{cases}v_s&\text{if }v=v_s\\
v_{j}^{\widehat{i}} &\text{if }v=v_j^i\\
\bot &\text{if }v=\top\\
\top &\text{if }v=\bot\enspace ,
\end{cases}
\]
such that in round $i$ we have that $P_i=f(P_i')$ and the plays are equally likely.
Thus, the probability to reach $f(\bot)=\top$, starting in state $f(v)$, for each $v$ is at least $1-\val(v)-\eps=\val(f(v))-\eps$, where the equality follows from Lemma~\ref{lem:mirror}. Hence, $\sigma_1$ is $\eps$-optimal for player~1.
\end{proof}

Next we give a definition and a lemma, which is similar to Lemma~6 in~\cite{IM12}. 
The purpose of the lemma is to identify certain cases where one can change the transition function of an MDP in a specific way and 
obtain a new MDP with larger values. 
%%(for the time-limited case, and then we present Corollary~\ref{cor:change_child} that shows the values without time limit also increase). 
We cannot simply obtain the result from Lemma~6 in~\cite{IM12}, since the direction is opposite (i.e., Lemma~6 in~\cite{IM12} considers some cases where one can change the transition function and obtain a new MDP with {\em smaller} values) and our lemma is also for a slightly more general class of MDPs. 
%%The proof follows the same idea as the proof of Lemma~6 in \cite{IM12}, though. 
%% KRISH: OMITTED LAST SENTENCE. PARAGRAPH STARTS WITH THAT.

\begin{definition}
Let $G$ be an MDP with safety objectives. A {\em replacement set} is a set of triples of states, actions and distributions over the states $Q=\{(s_1,a_1,\delta_1),\dots,(s_\ell,a_\ell,\delta_\ell)\}$. Given the replacement set $Q$, the MDP $G[Q]$ is an MDP over the same states as $G$ and with the same set of safe states, but where the transition function $\delta'$ 
is \[
\delta'(s,a)=\begin{cases}
\delta_i & \text{if }s=s_i\text{ and }a=a_i\text{ for some }i\\
\delta(s,a) &\text{otherwise}
\end{cases}
\]
\end{definition}

\begin{lemma}\label{lem:change_child}
Let $G$ be an MDP with safety objectives.
Consider some replacement set  
\[
Q=\{(s_1,a_1,\delta_1),\dots,(s_\ell,a_\ell,\delta_\ell)\}\enspace ,\] such that for all $t$ and $i$ we have that 
\[\sum_{s\in S}(\delta(s_i,a_i)(s)\cdot \overline{v}^t_s)\leq 
\sum_{s\in S}(\delta_i(s)\cdot \overline{v}^t_s)\enspace .\] 
Let $\overline{v'}^t$ be the value vector for $G[Q]$ with finite horizon $t$.
(1)~For all states $s$ and time limits $t$ we have that \[\overline{v}^t_s\leq \overline{v'}^t_s\enspace .\]
(2)~For all states $s$, we have that \[\val(G,s)\leq \val(G[Q],s)\enspace .\]
\end{lemma}

\begin{proof}
We first present the proof of first item.
We will show, using induction in $t$, that 
$\overline{v}^t_s\leq \overline{v'}^t_s$ for all $s$. Let $\delta'$ be the transition function for $G[Q]$.

{\bf Base case, $t=0$:} Consider some state $s$. Clearly we have that $\overline{v}^t_s= \overline{v'}^t_s$ because we have not changed the safe states.

{\bf Induction case, $t\geq 1$:} 
The induction hypothesis state that $\overline{v}^{t-1}_s\leq \overline{v'}^{t-1}_s$ for all $s$.
Consider some state $s$. Consider any action $a'$ such that there is an $i$ such that $s=s_i$ and $a=a_i$. 
We have that 
\[\sum_{s'} (\delta(s,a')(s')\cdot \overline{v}^{t-1}_{s'})\leq \sum_{s'} (\delta'(s,a')(s')\cdot \overline{v}^{t-1}_{s'})\] by definition for such $a'$ (the statement is true for all time limits and thus also for $t-1$). 
For all other actions $a''$ we have that 
\[\sum_{s'} (\delta(s,a'')(s')\cdot \overline{v}^{t-1}_{s'})= \sum_{s'} (\delta'(s,a'')(s')\cdot \overline{v}^{t-1}_{s'})\enspace ,\]
since $\delta(s,a'')=\delta'(s,a'')$.
Hence, \[\min_a\sum_{s'} (\delta(s,a)(s')\cdot \overline{v}^{t-1}_{s'})\leq \min_a\sum_{s'} (\delta'(s,a)(s')\cdot \overline{v}^{t-1}_{s'})\]

We then have, using the recursive definition of $\overline{v}^{t}_s$, that
\begin{align*}
\overline{v}^{t}_s&=
\min_a \sum_{s'} (\delta(s,a)(s')\cdot \overline{v}^{t-1}_{s'})\\
&\leq\min_a \sum_{s'} (\delta'(s,a)(s')\cdot \overline{v}^{t-1}_{s'})\\
&\leq\min_a\sum_{s'} (\delta'(s,a)(s')\cdot \overline{v'}^{t-1}_{s'}) \\
 &=\overline{v'}^{t}_s \enspace .
\end{align*}
where we just argued the first inequality; and the second inequality comes from the induction hypothesis and that each factor is positive. (Note that the optimal strategy for player~2 in a matrix game $A^s[\overline{v}^{t-1}]$  of 1~row is to pick one of the columns with the smallest entry with probability~1 and thus $\overline{v}^{t}_s=\val(A^s[\overline{v}^{t-1}])=\min_a \sum_{s'} (\delta(s,a)(s')\cdot \overline{v}^{t-1}_{s'})$ and similarly for $\overline{v'}^{t}_s$).
This completes the proof of the first item.
The second item follows from the first item and since the value of a time limited game goes 
to the value of the game without the time limit as the time limit grows to $\infty$, as shown by~\cite{Eve57}.
\end{proof}

%%% KRISH MERGED NEXT COROLLARY TO PREVIOUS LEMMA TO SAVE REPEAT AND SPACE.
%%\begin{corollary}\label{cor:change_child}
%Let $G$ be an MDP with safety objectives.
%Consider some set of triples of states, actions and distributions over the states \[A=\{(s_1,a_1,\delta_1),\dots,(s_\ell,a_\ell,\delta_\ell)\}\enspace ,\] such that for all $t$ and $i$ we have that \[\sum_{s\in S}(\delta(s_i,a_i)(s)\cdot \overline{v}^t_s)\leq 
%\sum_{s\in S}(\delta_i(s)\cdot \overline{v}^t_s)\enspace .\] 
%Then for each state $s$, we have that \[\val(G,s)\leq \val(G[A],s)\enspace .\]
%\end{corollary}
%\begin{proof}
%The corollary follows from Lemma~\ref{lem:change_child} since the value of a time limited game goes to the value of the game without the time limit as the time limit grows, as shown by~\cite{Eve57}.
%\end{proof}

We next show that for player~1, the patience of $\eps$-optimal strategies is high.
\begin{lemma}\label{lem:player_1_eps_patience}
For all $0<\eps<\frac{1}{3}$, each $\eps$-optimal stationary strategy $\sigma_1$ for player~1 in the Purgatory Duel has patience at least $2^{m^{\Omega(n)}}$. For $N=5$ the patience is $2^{\Omega(m)}$.
\end{lemma}

\begin{proof}
Consider some $\eps$-optimal stationary strategy $\sigma_1$ for player~1 in the Purgatory Duel. Fixing $\sigma_1$ for player~1 in the Purgatory Duel we obtain an MDP $G'$ for player~2. Let $\overline{v}^t$ be the value vector for $G'$ with finite horizon (time-limit) $t$ and let $\delta$ be the transition function for $G'$.
For each $i$, let \[
\delta_i(s)=\begin{cases}
\delta(v_n^2,i)(s) & \text{if }v_s\neq s\neq \bot\\
\delta(v_n^2,i)(\bot)+\delta(v_n^2,i)(v_s) & \text{if }v_s=s\\
0 & \text{if }\bot=s\\
\end{cases}
\]
(Note that $\delta_i$ is the same probability distribution as $\delta(v_n^2,i)$, except that the probability mass on $\bot$ is moved to $v_s$.)
Consider the replacement set $Q=\{(v_n^2,1,\delta_1),\dots,(v_n^2,m,\delta_m)\}$ and the MDP $G'[Q]$. We have for all $t$ and $i$ that \[\sum_{s\in S}(\delta(v_n^2,i)(s)\cdot \overline{v}^t_s)\leq 
\sum_{s\in S}(\delta_i(s)\cdot \overline{v}^t_s)\] because  \[\overline{v}^t_{\bot}=\overline{v}^t_{v^2_{n+1}}=0\leq \overline{v}^t_{v_s}\] for all $t$ and the only difference between $\delta(v_n^2,i)$ and $\delta_i$ is that the probability mass on $\bot$ is moved to $v_s$. We then get from Lemma~\ref{lem:change_child}(2) that $\val(G',v_s)\leq \val(G'[Q],v_s)$. Let $\sigma_2$ be an optimal positional strategy in $G'[Q]$. It is easy to see that $\sigma_2$ plays action~1 in $v_j^2$ for all $j$, because the best player~2 can hope for is to get back to $v_s$ since $\bot$ cannot be reached from $v_j^2$ in $G'[Q]$ for any $j$ and if he plays some action which is not~1, then there is a positive probability that $\top$ will be reached in one step. Thus, the MDP $G'[Q]$ corresponds to the MDP one gets by fixing the strategy $\sigma_1'$ where $\sigma_1'(v_i)=\sigma_1(v_i^1)$ for player~1 in Purgatory.
But the probability to reach $\top$ in $G'[Q]$ is at least $\frac{1}{2}-\eps$ and hence $\sigma_1'$ is $(\frac{1}{2}+\eps)$-optimal in Purgatory (note that this is Purgatory and not Purgatory Duel).
As shown by~\cite{HIM11} any such strategy requires patience $2^{m^{\Omega(n)}}$. Thus, any $\eps$-optimal stationary strategy for player~1 in the Purgatory Duel 
requires patience $2^{m^{\Omega(n)}}$.

It was shown by~\cite{HIM11} that the patience of $\eps$-optimal strategies for Purgatory with $n=1$ Purgatory state is $2^{\Omega(m)}$, and thus similarly for the Purgatory Duel with $N=5$. 
\end{proof}

We are now ready to prove the main theorem of this section.
\begin{theorem}\label{thm:big n}
For all $0<\eps<\frac{1}{3}$, every $\eps$-optimal stationary strategy, for either player, 
in the Purgatory Duel (that has $N=2n+3$ states and at most $m$ actions for each player at 
all states) has patience $2^{m^{\Omega(n)}}$. 
For $N=5$ the patience is $2^{\Omega(m)}$.
\end{theorem}

\begin{proof}
The statement for strategies for player~1 follows from Lemma~\ref{lem:player_1_eps_patience}. By Lemma~\ref{thm:epsilon_mirror}, for each $\eps$-optimal strategy for player~2, there is an $\eps$-optimal strategy for player~1 (i.e., the mirror strategy) with the same patience. Thus the result follows for strategies for player~2.
\end{proof}

\section{Zero-sum Concurrent Stochastic Games: Patience Lower Bound for Three States}
%%\section{Exponential patience with 3 states}
In this section we show that the patience of all $\eps$-optimal strategies, 
for all $0<\eps<\frac{1}{3}$, for both players in a concurrent reachability game $G$ with 
three states of which two are absorbing, and the non-absorbing state has $m$ actions 
for each player, can be as large as $2^{\Omega(m)}$.
The proof consists of two phases, first we show the lower bound in a game with at most 
$m^2$ actions for each player; and second, we show that all but $2m-1$ actions can be removed 
for both players in the game without changing the patience.

The first game, the {\em 3-state Purgatory Duel}, is intuitively speaking the Purgatory Duel for $N=5$, where we replace the states $v_1^1$, $v_1^2$ and $v_s$ with a state $v_s'$ while in essence keeping the same set of $\eps$-optimal strategies. The idea is to ensure that one step in the 3-state Purgatory Duel corresponds to two steps in the Purgatory Duel with $N=5$, by having the players pick all the actions they might use in the next two steps at once. 
The game is formally defined as follows:

The 3-state Purgatory Duel consists of $N=3$ states, named $v_s',\top'$ and $\bot'$ respectively.
The states $\top'$ and $\bot'$ are absorbing.
The state $v_s'$ is such that \[A_{v_s'}^1=A_{v_s'}^2=\{(i,j)\mid 1\leq i,j\leq m\}\enspace .\] 
Also, let $\delta'$ be the transition function for the Purgatory Duel with $N=5$. Let $p$ be the function that given a state in $\{v_s,\bot,\top\}$ in the Purgatory Duel for $i=1$ outputs the primed state (which is then a state in the 3-state Purgatory Duel). Recall that $\U(s,s')$ is the uniform distribution over $s$ and $s'$. Observe that the deterministic distributions $\delta'(v_1^1,a_1,a_2)$ and $\delta'(v_1^2,a_1,a_2)$ are in $\{v_s,\top,\bot\}$ for all $a_1$ and $a_2$.
For each pair of actions $(a_1^1,a_1^2)\in A_{v_s'}^1$ and $(a_2^1,a_2^2)\in A_{v_s'}^2$ in the 3-state Purgatory Duel, we have that 
\begin{align*}
&\delta(v_s',(a_1^1,a_1^2),(a_2^1,a_2^2))=\\
&\U(p(\delta'(v_1^1,a_1^1,a_2^1)),p(\delta'(v_1^2,a_1^2,a_2^2)))\enspace .
\end{align*}
To make the game easier to understand on its own, we now give a more elaborate description of the transition function 
$\delta$ without using the transition function for the Purgatory Duel. 
To make the pattern as clear as possible we write $\U(s,s)$ instead of $s$ for all $s$.
\begin{align*}
&\delta(v_s',(a_1^1,a_1^2),(a_2^1,a_2^2))=\\
&\begin{cases}
\U(\bot',\top') & \text{if }a_1^1>a_2^1\text{ and } a_1^2>a_2^2\\
\U(\bot',\bot') & \text{if }a_1^1>a_2^1\text{ and } a_1^2=a_2^2\\
\U(\bot',v_s') & \text{if }a_1^1>a_2^1\text{ and } a_1^2<a_2^2\\
\U(\top',\top') & \text{if }a_1^1=a_2^1\text{ and } a_1^2>a_2^2\\
\U(\top',\bot') & \text{if }a_1^1=a_2^1\text{ and } a_1^2=a_2^2\\
\U(\top',v_s') & \text{if }a_1^1=a_2^1\text{ and } a_1^2<a_2^2\\
\U(v_s',\top') & \text{if }a_1^1<a_2^1\text{ and } a_1^2>a_2^2\\
\U(v_s',\bot') & \text{if }a_1^1<a_2^1\text{ and } a_1^2=a_2^2\\
\U(v_s',v_s') & \text{if }a_1^1<a_2^1\text{ and } a_1^2<a_2^2\enspace .
\end{cases}
\end{align*}
Furthermore, $S^1=\{\top'\}$.
We will use $\tau_i$ for strategies in the 3-state Purgatory Duel to distinguish them from strategies in the Purgatory Duel. There is an illustration of the Purgatory Duel with $N=5$ and $m=2$ in Figure~\ref{fig:p5} and the corresponding 3-state Purgatory Duel in Figure~\ref{fig:3sp}.

\begin{figure}
\center
\begin{tikzpicture}[node distance=3cm,-{stealth},shorten >=2pt]
\ma{top}[$\top$]{1}{1};

\ma[shift={($(top)+(0cm,-3cm)$)}]{v11}[$v_1^1$]{2}{2};

\ma[shift={($(v11)+(3cm,-3cm)$)}]{vs}[$v_s$]{1}{1};

\ma[shift={($(v11)+(6cm,0)$)}]{v12}[$v_1^2$]{2}{2};

\ma[shift={($(v12)+(0,3cm)$)}]{bot}[$\bot$]{1}{1};
\nloop{top};
\nloop{bot};

\draw (v11-2-2.center) -- (top);
\draw (v11-1-1.center) -- (top);
\draw (v11-1-2.center) -- (bot);
\draw (v11-2-1.center) to[out=-90,in=135] (vs);

\draw (v12-2-2.center) -- (bot);
\draw (v12-1-1.center) -- (bot);
\draw (v12-1-2.center) to[out=120,in=-45] (top);
\draw (v12-2-1.center) to[out=-90,in=0] (vs);

\draw[dashed] (vs-1-1.center) to[out=90,in=-45]  (v11);
\draw[dashed] (vs-1-1.center) to[out=90,in=-135] (v12);
\end{tikzpicture}
\caption{An illustration of the Purgatory Duel with $N=5$ and $m=2$. The two dashed edge have probability~$\frac{1}{2}$ each.\label{fig:p5}}
\center
\begin{tikzpicture}[node distance=3cm,-{stealth},shorten >=2pt]
\ma{top}[$\top'$]{1}{1};

\ma[shift={($(top)+(0cm,-5cm)$)}]{vsb}[]{4}{4};
\fill[lightgray] (vsb-1-2.north west) rectangle (vsb-3-4.south east);
\ma[shift={($(top)+(0cm,-5cm)$)}]{vs}[$v_s'$]{4}{4};

\ma[shift={($(vs)+(0,-5cm)$)}]{bot}[$\bot'$]{1}{1};
\nloop[dashed]{top};
\nloop[dashed]{bot}[yscale=-1];
%1
\draw (vs-1-1.center) to[out=180,in=-160] (top);
\draw (vs-1-1.center) to[out=180,in=160] (bot);

\draw (vs-2-2.center) to[out=180,in=-150] (top);
\draw (vs-2-2.center) to[out=180,in=140] (bot);

\draw (vs-3-3.center) to[out=0,in=-30] (top);
\draw (vs-3-3.center) to[out=0,in=40] (bot);

\draw (vs-4-4.center) to[out=0,in=-10] (top);
\draw (vs-4-4.center) to[out=0,in=30] (bot);

%2(1,1),(1,2),(2,1),(2,2)
\draw (vs-1-2.center) to (vs-1-2.north) arc (0:180:0.5cm);
\draw (vs-1-2.center) to[out=90,in=-110] (top);

\draw[xscale=-1,rotate=90] (vs-3-4.center) to (vs-3-4.east) arc (0:180:0.5cm);
\draw (vs-3-4.center) to[out=0,in=-20] (top);

%3

\nloop{vs-1-3}[yscale=-1,rotate=90];
\draw (vs-1-3.center) to[out=0,in=20] (bot);

\draw[rotate=-90] (vs-2-4.center) to (vs-2-4.east) arc (0:180:0.43cm);

\draw (vs-2-4.center) to[out=0,in=10] (bot);

%4
\nloop[dashed]{vs-1-4};

%5
\draw[dashed](vs-2-1.center) to[out=80,in=-120] (top);

\draw[dashed](vs-4-3.center) to[out=110,in=-100] (top);

%6

\draw[xscale=-1] (vs-2-3.center) to[out=90,in=-90] (vs-1-2.north east) arc (0:180:0.5cm);
\draw (vs-2-3.center) to[out=90,in=-90] (top);

%7
\draw[dashed] (vs-3-1.center) to[out=-80,in=110] (bot);

\draw[dashed] (vs-4-2.center) to[out=-80,in=100]  (bot);

%8
\draw[yscale=-1] (vs-3-2.center) to[out=90,in=-90] (vs-4-2.south east) arc (0:180:0.45cm);
\draw (vs-3-2.center) to[out=-90,in=80](bot);

%9
\draw(vs-4-1.center) to[out=180,in=-170] (top);
\draw(vs-4-1.center) to[out=180,in=150] (bot);
\end{tikzpicture}
\caption{An illustration of the 3-state Purgatory Duel $m=2$. The {\em non-dashed} edges have probability~$\frac{1}{2}$ each. The order of the actions is $(1,1),(1,2),(2,1),(2,2)$. The actions (i.e., $(2,2)$ for player~1 and $(1,1)$ for player~2) with white background cannot be played in a restricted strategy.\label{fig:3sp}}
\end{figure}
\begin{comment}
1 2 3 4
5 1 6 3
7 8 1 2
9 7 5 1
\end{comment}

Given a strategy $\tau_i$ for player~$i$ in the 3-state Purgatory Duel we define the strategy $\sigma_i$ in the Purgatory Duel with $N=5$ 
which is the projection of $\tau_i$ and vice versa (note that the other direction maps to a set of strategies). 

\begin{definition}
Given a strategy $\tau_i$ for player~$i$ in the 3-state Purgatory Duel, let $\sigma_i^{\tau_i}$ be the stationary strategy for 
player~$i$ in the Purgatory Duel with $N=5$ where 
\[
\sigma_i^{\tau_i}(v_1^1)(a_1^1)=\sum_{a_1^2}\tau_i(v_s')(a_1^1,a_1^2)
\]
and 
\[
\sigma_i^{\tau_i}(v_1^2)(a_1^2)=\sum_{a_1^1}\tau_i(v_s')(a_1^1,a_1^2) \enspace .
\]
Also, for any stationary strategy $\sigma_i$ in the Purgatory Duel with $N=5$, let $\Tau_i^{\sigma_i}$ be the set of stationary strategies in the 3-state Purgatory Duel such that $\tau_i\in\Tau_i^{\sigma_i}$ implies that $\sigma_i^{\tau_i}=\sigma_i$.
\end{definition}

\begin{lemma}\label{lem:projection}
Consider any $\eps\geq 0$.
Let $G$ be the Purgatory Duel with $N=5$ and $G'$ be the 3-state Purgatory Duel.
For any $\eps$-optimal stationary strategy $\tau_i$ for player~$i$ in $G'$, we have that $\sigma_i^{\tau_i}$ is $\eps$-optimal starting in $v_s$ in $G$. Similarly, for any $\eps$-optimal stationary strategy  $\sigma_i$ in $G$ starting in $v_s$ each strategy in $\Tau_i^{\sigma_i}$ is $\eps$-optimal in $G'$.
Also, $\val(v_s')=\frac{1}{2}$.
\end{lemma}
\begin{proof}
Consider some pair of strategies $\tau_i$ and $\sigma_i^{\tau_i}$ for player~$i$ in $G'$ and $G$, respectively. 
Fixing $\tau_i$ and $\sigma_i^{\tau_i}$ as the strategy for player~$i$ we get two MDPs $H'$ and $H$, respectively. 
We will argue that $\val(H',v_s')=\val(H,v_s)$. Let $\overline{v'}^t$ and $\overline{v}^t$ be the vector of values for the value iteration algorithm in iteration $t$ when run on $H'$ and $H$ respectively (i.e., the values of $H'$ and $H$ with time limit $t$). We have that $\overline{v}^{2t}_{v_s}=\overline{v'}^{t}_{v_s'}$ by definition of the value-iteration algorithm and the transition function in the 3-state Purgatory Duel. Hence, since $\overline{v}^{2t}_{v_s}$ and $\overline{v'}^{t}_{v_s'}$ converges to the value of state $v_s$ and $v_s'$ in $H$ and $H'$ respectively, they have the same value. We know that the value of $v_s$ is $\frac{1}{2}$ and thus that is also the value of $v_s'$.
\end{proof}

\begin{corollary}
The patience of $\eps$-optimal stationary strategies for both players, 
for $0<\eps<\frac{1}{3}$, in the 3-state Purgatory Duel is at least 
$2^{\Omega(m)}$, where $m^2$ is the number of actions in state $v_s$.
\end{corollary}
\begin{proof}
The patience of $\eps$-optimal strategies, for $0<\eps<\frac{1}{3}$, in the Purgatory Duel with $N=5$ is $2^{\Omega(m)}$ from Theorem~\ref{thm:big n}. 
Thus, by Lemma~\ref{lem:projection}, the patience of the 3-state Purgatory Duel is $2^{\Omega(m)}$.
\end{proof}

\smallskip\noindent{\bf The restricted 3-state Purgatory Duel.} 
The above corollary only shows that the for the 3-state Purgatory Duel, in which one state have 
$m^2$ actions and others have 1, the patience is at least $2^{\Omega(m)}$. 
We now show how to decrease the number of actions from quadratic down to 
linear, while keeping the same patience.

From Lemma~\ref{lem:precise} and Lemma~\ref{lem:mat_wh} we see that for any optimal strategy $\sigma_1$ for player~1 (resp., $\sigma_2$ for player~2) in the Purgatory Duel with $N=5$, we have that $\sigma_1(v_1^1)(1)>\frac{1}{2}$ and that $\sigma_1(v_1^2)(1)>\frac{1}{2}$ (resp., $\sigma_2(v_1^1)(m)>\frac{1}{2}$ and that $\sigma_2(v_1^2)(m)>\frac{1}{2}$). Hence, there exists an optimal strategy for player~1 in the 3-state Purgatory Duel that only plays actions on the form $(1,a_1^2)$ and $(a_1^1,1)$ with positive probability. More precisely, the strategy $\tau_1$ where (1)~$\tau_1(v_s)((1,a_1^2))=\sigma_1(v_1^2)(a_1^2)$; and (2)~$\tau_1(v_s)((a_1^1,1))=\sigma_1(v_1^1)(a_1^1)$; and (3)~has the remaining probability mass on $(1,1)$ is optimal in the 3-state Purgatory Duel, since $\sigma_1^{\tau_1}$ is $\sigma_1$. Similarly for player~2 and the actions $(m,a_2^2)$ and $(a_2^1,m)$. Let \[R_1=\{(i,j)\mid i=1\vee j=1, 1 \leq i,j \leq m\}\] and \[R_2=\{(i,j)\mid i=m\vee j=m, 1 \leq i,j \leq m\}\enspace .\]
Observe that $|R_1|=|R_2|=2m-1$.
We say that a strategy for player~$i$, for each $i$, is {\em restricted} if the strategy uses only actions in $R_i$.
The sub-matrix corresponding to the restricted 3-state Purgatory Duel for $m=2$ is depicted as the grey sub-matrix in Figure~\ref{fig:3sp}.
This suggests the definition of the {\em restricted 3-state Purgatory Duel}, which is like the 3-state Purgatory Duel, 
except that the strategies for the players are restricted. 
We next show that $\eps$-optimal strategies in the restricted 3-state Purgatory Duel also have high patience (note, that while this is perhaps not surprising, it does not follow directly from the similar result for the 3-state Purgatory Duel, since it is possible that the restriction removes the optimal best reply to some strategy which would otherwise not be $\eps$-optimal). 
The key idea of the proof is as follows: (i)~we show that the patience of player~$i$ 
in the 3-state Purgatory Duel remains unchanged even if only the opponent is enforced to
use restricted strategies; and (ii)~each player has a restricted strategy that 
is optimal in the 3-state Purgatory Duel as well as in the restricted 3-state Purgatory Duel. 
%%% KRISH ADDED ABOVE SENTENCE. THE REST CAN BE REMOVED FOR THE SHORT VERSION.

\begin{lemma}
The value of state $v_s'$ in the restricted 3-state Purgatory Duel is $\frac{1}{2}$
\end{lemma}
\begin{proof}
Each player has a restricted strategy which is optimal in the 3-state Purgatory Duel and ensures value $\frac{1}{2}$. 
Thus, these strategies must still be optimal in the restricted 3-state Purgatory Duel and still ensure value $\frac{1}{2}$.
\end{proof}

The next lemma is conceptually similar to Lemma~\ref{lem:eps_opt_partial_mixed} for $N=5$ (however, it does not follow from 
Lemma~\ref{lem:eps_opt_partial_mixed}, since the strategies for player~1 are restricted here).

\begin{lemma}\label{lem:res._3-state_partial_mixed}
Let $\tau_2$ be an $\eps$-optimal stationary strategy for player~2 in the restricted 3-state Purgatory Duel, 
for $0<\eps<\frac{1}{2}$. Then, $\sum_{i=1}^m\tau_2(v_s')(i,j)>0$, for each $j$.
\end{lemma}

\begin{proof}
Fix $0<\eps<\frac{1}{2}$. 
Let $\tau_2$ be a stationary strategy in the 3-state Purgatory Duel (note, we do not require that $\tau_2$ is restricted), such that there exists an $a_2$ for which $\sum_{a_1}\tau_2(v_s')((a_1,a_2))=0$. Let $a'$ be smallest such $a_2$.

Fix $0<\eta<\frac{1}{2}-\eps$.
We show that there exists a restricted stationary strategy $\tau_1$ for player~1, ensuring that the payoff is at least $1-\eta>\frac{1}{2}+\eps$.
There are two cases. Either (i)~$a'=1$ or (ii)~not. 

In case (i), let $\sigma_1(v_s')$ be an $\eta$-optimal strategy for player~1 in the {\em Purgatory} with parameters $(3,m)$. Then consider the strategy $\tau_1(v_s')$, where $\tau_1(v_s')((a,1))=\sigma_1(v_s')(a)$, for each $a$. Observe that $\tau_1$ is a restricted strategy. Consider what happens if $\tau_1$ is played against $\tau_2$: 
In each round~$i$, as long as $v_i=v_s'$, the next state is either defined by the first or the second component of the actions of the players. If it is defined by the second component, then the next state $v_{i+1}$ is always $v_s'$, because player~1's first component is $1$ and player~2's first component greater than~1. Consider the rounds where the next state is defined by the first component. In such rounds 
$\top$ is reached with probability $(1-\eta)\cdot p$, for some $p>0$ and $\bot$ is reached with probability at most $\eta\cdot p$, because player~1 follows an $\eta$-optimal strategy in Purgatory on the first component. 
But in expectation, in every second round the first component is used and thus $\top$ is reached with probability at least 
$1-\eta$, which shows that $\sigma_2$ is not $\eps$-optimal.

In case~(ii), consider the strategy $\tau_1$, such that $\tau_1(v_s')((1,a'))=1$. Observe that $\tau_1$ is a restricted strategy. Consider what happens if $\tau_1$ is played against $\tau_2$: In each round~$i$, as long as $v_i=v_s'$, the next state is either defined by the first or the second component of the players choice. If it is defined by the first component, then the next state $v_{i+1}$ is always $v_s'$ or $\top$, because the choice of player~1 is $1$. Consider the rounds where the next state is defined by the second component. In each such round either $\top$ or $v_s'$ is reached and $\top$ is reached with positive probability, since player~1 plays $a'>1$ and player~2 always plays something else and $1$ with positive probability. But in expectation, in every second round the second component is used and hence~$\top$ is reached with probability~1 eventually, which shows that $\sigma_2$ is not $\eps$-optimal.
\end{proof}

We will now define how to mirror strategies in the restricted 3-state Purgatory Duel.
\begin{definition}
Given a stationary strategy $\tau_i$ for player~$i$ in the restricted 3-state Purgatory Duel, for either $i$, 
let $\tau_{\widehat{i}}^{\tau_i}$ be the stationary strategy for player~$\widehat{i}$ (referred to as the mirror
strategy of $\tau_i$) in the restricted 3-state Purgatory Duel 
where $\tau_{\widehat{i}}^{\tau_i}(v_s')((a_1,a_2))=\tau_i(v_s')((a_2,a_1))$ for each $a_1$ and $a_2$.
\end{definition}

We next show that each $\eps$-optimal stationary strategy for player~2 can be mirrored to an $\eps$-optimal stationary for player~1. 
The statement and the proof idea are similar to Lemma~\ref{thm:epsilon_mirror}, but since the strategies for the players are restricted here, 
there are some differences.

\begin{lemma}\label{lem:change players}
For all $0<\eps<\frac{1}{2}$, each $\eps$-optimal stationary strategy $\tau_2$ for player~2 in the restricted 3-state Purgatory Duel is such that the mirror 
strategy $\tau_1^{\tau_2}$ is $\eps$-optimal for player~1 in the restricted 3-state Purgatory Duel.
\end{lemma}
\begin{proof}
Fix $\eps$, such that $0<\eps<\frac{1}{2}$. 
Consider some $\eps$-optimal stationary strategy $\tau_2^*$ for player~2 in the restricted 3-state Purgatory Duel. 
Let $\tau_1^*=\tau_1^{\tau_2^*}$ be the mirror strategy for player~1 given $\tau_2^*$ 
and let $\tau_2$ be an optimal best reply to $\tau_1^*$.
Let $\tau_1=\tau_1^{\tau_2}$ be the mirror strategy for player~1 given $\tau_2$.  
Observe that eventually either $\top$ or $\bot$ is reached with probability~1, when playing $\tau_1^*$ against $\tau_2$, 
by Lemma~\ref{lem:res._3-state_partial_mixed} and the construction of the game (since there is a positive probability 
that the second component matches in every round in which the play is in $v_s'$). 
We have that $u(v_s',\tau_1,\tau_2^*)\leq\frac{1}{2}+\eps$, since $\tau_2^*$ is $\eps$-optimal. 
This indicates that $\top'$ is reached with probability at most $\frac{1}{2}+\eps$ when playing $\tau_1$ 
against $\tau_2^*$. 
Hence, by symmetry $\bot'$ is reached with probability at most $\frac{1}{2}+\eps$ when playing $\tau_1^*$ against $\tau_2$. 
Thus, since $\bot'$ or $\top'$ is reached with probability~1, we have that $u(v_s',\tau_1^*,\tau_2)\geq\frac{1}{2}-\eps$, 
showing that $\tau_1^*$ is $\eps$-optimal.
\end{proof}

We next show that $\eps$-optimal stationary strategies for player~1 requires high (exponential) patience. 
The statement and the proof idea are similar to Lemma~\ref{lem:player_1_eps_patience}, 
but since the players strategies are restricted here, there are some differences.

\begin{lemma}\label{lem:p1_in_rest}
For all $0<\eps<\frac{1}{3}$, each $\eps$-optimal stationary strategy $\sigma_1$ for player~1 in the restricted 3-state Purgatory Duel has patience $2^{\Omega(m)}$.
\end{lemma}
\begin{proof}
Fix some $0<\eps<\frac{1}{3}$ and some $\eps$-optimal stationary strategy $\sigma_1$ for player~1 in the restricted 3-state Purgatory Duel. The restricted 3-state Purgatory Duel then turns into an MDP $M$ for player~2 and we can apply Lemma~\ref{lem:change_child}(2). 
We have that $p=\sum_{a_1^1}\sigma_1(v_s')(a_1^1,a_2^2)/2$ is the probability that player~1 plays an action with second component $a_2^2$ and the next state is defined by the second component. 
Let $\dist(a_1^2,a_2^2)$ be the probability distribution over successors if player~2 plays  $(a_1^2,a_2^2)$ in $v_s'$. 
Observe that the play would go to $\bot$ if both players played $a_2^2$ and the next state is defined by the second component 
and thus 
\[
\dist(a_1^2,a_2^2)(\bot)-p\geq 0\enspace . 
\]
Let 
\[
\dist'(a_1^2,a_2^2)(v)=\begin{cases}\dist(a_1^2,a_2^2)(v_s')+p&\text{if }v=v_s'\\
\dist(a_1^2,a_2^2)(\bot)-p&\text{if }v=\bot\\
\dist(a_1^2,a_2^2)(\top)&\text{if }v=\top\enspace . \\
\end{cases} 
\]
Consider the MDP $M'$, which is equal to $M$, except that it uses the distribution $\dist'(a_1^2,a_2^2)$ instead of $\dist(a_1^2,a_2^2)$. By Lemma~\ref{lem:change_child}(2) we have that \[\val(M')\geq \val(M)\geq \frac{1}{2}-\eps\geq \frac{1}{6}\enspace .\]
It is clear that player~2 has an optimal positional strategy in $M'$ that plays $(a_1^2,m)$ for some $a_1^2$ (this strategy is restricted), since playing $(a_1^2,a_2^2)$, for some $a_2^2<m$, just increases the probability to reach $\top$ in one step (because player~1 might play some action $a_2^1>a_2^2$ and otherwise the play will go back to $v_s'$).
But $M'$ corresponds to the MDP obtained by playing $\sigma_1$ in the Purgatory with $N=3$ (where $v_s'$ corresponds to $v_1$), except that with probability $\frac{1}{2}$ the play goes  from $v_s'$ back to $v_s'$ in the restricted 3-state Purgatory Duel no matter the choice of the players. This difference clearly does not change the value. Hence, $\sigma_1$ ensures payoff at least $\frac{1}{6}$ in the Purgatory with $N=3$ and hence has patience $2^{\Omega(m)}$ by~\cite{HIM11}.
\end{proof}

We are now ready for the main result of this section.

\begin{theorem}\label{thm:small n}
For all $0<\eps<\frac{1}{3}$, every $\eps$-optimal stationary strategy, for either player, 
in the restricted 3-state Purgatory Duel (that has three states, two of which are absorbing, and 
the non-absorbing state has $O(m)$ actions for each player) has patience $2^{\Omega(m)}$.
\end{theorem}
\begin{proof}
By Lemma~\ref{lem:p1_in_rest}, the statement is true for every $\eps$-optimal stationary strategy for player~1. 
By Lemma~\ref{lem:change players}, every $\eps$-optimal stationary strategy for player~2 corresponds to an $\eps$-optimal stationary strategy for player~1, 
with the same patience, and thus every $\eps$-optimal stationary strategy for player~2 has patience $2^{\Omega(m)}$.
\end{proof}

\section{Zero-sum Concurrent Stochastic Games: Patience Upper Bound}
In this section we give upper bounds on the patience of optimal and
$\eps$-optimal stationary strategies in a zero-sum concurrent
reachability game $G$ for the safety player. Our exposition here
makes heavy use of the setup of Hansen~et~al.\ \cite{HKLMT11} and will
for that reason not be fully self-contained. We assume for
concreteness that the player 1 is the reachability player and player 2 the
safety player.

Hansen~et~al.\ showed \cite[Corollary~42]{HKLMT11} for the more
general class of Everett's recursive games~\cite{Eve57} that each
player has an $\eps$-optimal stationary strategy of doubly-exponential
patience. More precisely, if all probabilities have bit-size at most
$\tau$, then each player has an $\eps$-optimal strategy of patience bounded
by $(\frac{1}{\eps})^{\tau m^{O(N)}}$. For zero-sum concurrent
reachability games the safety player is guaranteed to have an optimal
stationary strategy~\cite{BAMS:Parthasarathy71,PAMS:HPRV76}. Using
this fact one may use directly the results of Hansen~et~al.\ to show
that the safety player has an optimal strategy of patience bounded by
$(\frac{1}{\eps})^{\tau m^{O(N^2)}}$. We shall below refine this
latter upper bound in terms of the number of value classes of the
game. The overall approach in deriving this is the same, namely we use
the general machinery of real algebraic geometry and semi-algebraic
geometry~\cite{BasuPollackRoy2006} to derive our bounds. In order to
do this we derive a formula in the first order theory of the real
numbers that uniquely defines the value of the game, and from the
value of the game we can express the optimal strategies. The improved bound
is obtained by presenting a formula where the number of variables depend
only on the number of value classes rather than the number of
states.

Let below $N$ denote the number of non-absorbing states, and $m \geq 2$
the maximum number of actions in a state for either player. Assume
that all probabilities are rational numbers with numerators and
denominators of bit-size at most $\tau$, where the bit-size of a positive
integer $n$ is given by $\lfloor \lg n \rfloor+1$. We let $K$ denote
the number of value classes. We number the non-absorbing states
$1,\dots,N$ and assume that both players have the actions
$\{1,\dots,m\}$ in each of these states. For a non-negative integer
$z$, define $\bit(z)=\lceil \lg z \rceil$.

Given valuations $v_1,\dots,v_N$ for the non-absorbing states, we
define for each state $k$ a $m \times m$ matrix game $A^k(v)$ letting
entry $(i,j)$ be $s_{ij}^k + \sum_{\ell=1}^N p_{ij}^{k\ell}v_\ell$,
where $p_{ij}^{k\ell} = \delta(k,i,j)(\ell)$ and $s_{ij}^k$ is the
probability of a transition to a state where the reachability player
wins, given actions $i$ and $j$ in state $k$. The \emph{value mapping}
operator $M:\R^N \rightarrow \R^N$ is given by
$M(v)=\left(\val(A^1(v),\dots,\val(A^N(v)))\right)$. Everett showed
that the value vector of his recursive games are given by the unique
\emph{critical vector}, which in turn is defined using the value
mapping. We will instead for concurrent reachability games use the
characterization of the value vector as the coordinate-wise least
fixpoint of the value mapping. The value vector $v$ is thus
characterized by the formula
\begin{equation}
\label{EQ:MinFixedpoint}
M(v)=v \wedge \left( \forall v' : M(v')=v' \Rightarrow v \leq v'\right ) \enspace .
\end{equation}
Similarly to \cite[proof of Theorem 13]{HKLMT11} we obtain the
following statement.
\begin{lemma}
  There is a quantifier free formula with $N$ variables $v$ that
  expresses $M(v)=v$. The formula uses at most $N(m+2)4^m$ different
  polynomials, each of degree at most $m+2$ and having coefficients of
  bit-size at most $2(N+1)(m+2)^2\bit(m)\tau$.
\end{lemma}
Now, if we instead introduce a variable for each value class, we can
express $M(v)=v$ using only $K$ free variables, by identifying
variables of the same value class. For $w \in \R^K$, let $v(w) \in
\R^N$ denote the vector obtained by letting the coordinates
corresponding to value class $j$ be assigned $w_j$. We thus simply
express $M(v(w))=v(w)$ instead. Combining this with
(\ref{EQ:MinFixedpoint}) we obtain the final formula.
\begin{corollary}
\label{COR:FirstOrderValueClasses}
  There is a quantified formula with $K$ free variables that describes
  whether the vector $v(w)$ is the value vector of $G$. The
  formula has a single block of quantifiers over $K$
  variables. Furthermore the formula uses at most $2N(m+2)4^m+K$ different
  polynomials, each of degree at most $m+2$ and having coefficients of
  bit-size at most $2(N+1)(m+2)^2\bit(m)\tau$.
\end{corollary}
We shall now apply the \emph{quantifier elimination}
\cite[Theorem~14.16]{BasuPollackRoy2006} and \emph{sampling}
\cite[Theorem~13.11]{BasuPollackRoy2006} procedures to the formula of
Corollary~\ref{COR:FirstOrderValueClasses}.

First we use Theorem~14.16 of Basu, Pollack, and
Roy~\cite{BasuPollackRoy2006} obtaining a quantifier free formula with
$K$ variables, expressing that $w(v)$ is the value of $G$. Next
we use Theorem~13.11 of \cite{BasuPollackRoy2006} to obtain a
univariate representation of $w$ such that $v(w)$ is the value vector
of $G$. That is, we obtain univariate real polynomials $f,
g_0,\dots,g_K$, where $f$ and $g_0$ are co-prime, such that
$w=\left(g_1(t)/g_0(t),\dots,g_K(t)/g_0(t)\right)$, where $t$ is a
root of $f$. These polynomial are of degree $m^{O(K^2)}$ and their
coefficients have bit-size $\tau m^{O(K^2)}$. Our next task is to
recover from $w$ an optimal strategy for the safety player. For this
we just need to select optimal strategies for the column player in
each of the matrix games $A^k(v(w))$. Such optimal strategies
correspond to basic feasible solutions of standard linear programs for
computing the value and optimal strategies of matrix games
(cf.~\cite[Lemma~3]{HKLMT11}). This means that there exists
$(m+1)\times(m+1)$ matrices $M^1(w),\dots,M^N(w)$, such that
$(q^k_1(w),\dots,q^k_m(w))$ is an optimal strategy for the column
player in $A^k(v(w))$ where $q^k_i(w) =
\det((M^k(w))_i)/\det(M^k(w))$, where $(M^k(w))_i$ denotes the matrix
obtained from $M^k(w)$ by replacing column $i$ with the $(m+1)$th unit
vector $e_{m+1}$. As the matrices $M^1(w),\dots,M^k(w)$ are obtained
from the matrix games $A^1(v(w)),\dots,A^N(v(w))$, the entries are
degree 1 polynomial in $w$ and having rational coefficients with
numerators and denominators of bit-size at most $\tau$ as well.  Using
a simple bound on determinants
\cite[Proposition~8.12]{BasuPollackRoy2006}, and substituting the
expression $g_j(t)/g_0(t)$ for $w_j$ for each $j$, we obtain a
univariate representation of $(q^k_1(w),\dots,q^k_m(w))$ for each $k$
given by polynomials of degree $m^{O(K^2)}$ and their coefficients
have bit-size $\tau m^{O(K^2)}$.  Substituting the root $t$ using
resultants (cf.~\cite[Lemma~15]{HKLMT11}) we finally obtain the
following result.
\begin{theorem}
\label{THM:OptimalStrategy-RealAlgebraic}
  Let $G$ be a zero-sum concurrent reachability game with $N$
  non-absorbing states, at most $m\geq 2$ actions for each player in
  every non-absorbing state, and where all probabilities are rational
  numbers with numerators and denominators of bit-size at most
  $\tau$. Assume further that $G$ has at most $K$ value
  classes. Then there is an optimal strategy for the safety player
  where each probability is a real algebraic number, defined by a
  polynomial of degree $m^{O(K^2)}$ and maximum coefficient bit-size
  $\tau m^{O(K^2)}$.
\end{theorem}
By a standard root separation bounds (e.g.~\cite[Chapter~6,
equation~(5)]{Yap2000}) we obtain a patience upper bound.
\begin{corollary}
  Let $G$ be as in
  Theorem~\ref{THM:OptimalStrategy-RealAlgebraic}. Then there is an
  optimal strategy for the safety player of patience at most
  $2^{\tau m^{O(K^2)}}$.
\end{corollary}
In general the probabilities of this optimal strategy will be
irrational numbers. However we may employ the rounding scheme as
explained in Lemma~14 and Theorem~15 of Hansen, Kouck\'y, and
Miltersen~\cite{HKM09} to obtain a rational $\eps$-optimal
strategy. Letting $\eps=2^{-\ell}$ we may round each probability,
except the largest, upwards to
$L=\lg\frac{1}{\eps}+\lg\lg\frac{1}{\eps}+N\tau m^{O(K^2)}$ binary
digits, and then rounding the largest probability down by the total
amount the rest were rounded up. Here we use that by fixing the above
strategy of patience at most $2^{\tau m^{O(K^2)}}$ for the safety
player and an pure strategy for the reachability player one obtains a
Markov chain where each non-zero transition probability is at least
$(2^{\tau m^{O(K^2)}})^{-1}$. We thus have the following.
\begin{corollary}
\label{COR:RoundedEpsOptimal}
  Let $G$ be as in
  Theorem~\ref{THM:OptimalStrategy-RealAlgebraic}. Then there is an
  $\eps$-optimal strategy for the safety player where each
  probability is a rational number with a common denominator of
  magnitude at most $\frac{1}{\eps}\lg{\frac{1}{\eps}}2^{N\tau m^{O(K^2)}}$.
\end{corollary}
We now address the basic decision problem. Let $s$ be a state and let
$\lambda$ be a rational number with numerator and denominator of
bit-size at most $\kappa$, and consider the task of deciding whether
$v_2(s) \geq \lambda$. An equivalent task is to decide whether
$v_2(s)-\lambda \geq 0$. Since $v_2(s)$ is a real algebraic number
defined by a polynomial of degree $m^{O(K^2)}$ and maximum coefficient
bit-size $\tau m^{O(K^2)}$ it follows that $v_2(s)-\lambda$ is a real
algebraic number defined by a polynomial of degree $m^{O(K^2)}$ and
maximum coefficient bit-size $(\kappa+\tau) m^{O(K^2)}$. This can be
seen by subtracting $\lambda$ from the univariate representation of
$v_2(s)$ and substituting for the root $t$ using a resultant. By
standard root separation bounds this means that either is
$v_2(s)-\lambda=0$ or $\abs{v_2(s)-\lambda} > \eta$, for some
$\eta$ of the form $d=2^{-(\kappa+\tau)m^{O(K^2)}}$. Given an
$\eta/2$-optimal strategy $\sigma_2$ for the safety player, by
fixing the strategy $\sigma_2$ we obtain an MDP for player 1, 
where we can find the value $\widetilde{v}_2(s)$ of state
$s$ using linear programming, and the computed estimate
$\widetilde{v}_2(s)$ for $v_2(s)$ is within $\eta/2$ of the true
value. Thus if $\widetilde{v}_2(s) \geq \lambda - \eta/2$ we
conclude that $v_2(s) \geq \lambda$ (and similarly if $\widetilde{v}_2(s)
\geq \lambda + \eta/2$ we conclude that $v_2(s) > \lambda$).
Now, if we fix $K$ to be a constant and consider the promise problem
that $G$ has at most $K$ value classes, then a rational
$\eta/2$-optimal strategy $\sigma_2$ exists with numerators and 
denominators of polynomial bit-size by Corollary~\ref{COR:RoundedEpsOptimal}. 
Now, by simply guessing non-deterministically the strategy $\sigma_2$ and
verifying as above we have the following result.
\begin{theorem}
  For a fixed constant $K$, the promise problem of deciding whether
  $v_1(s) \geq \lambda$ given a zero-sum concurrent stochastic game with 
  at most $K$ value classes is in $\coNP$ if player~1 has reachability
  objective and in $\NP$ if player~1 has safety objective.
\end{theorem}
Note that interestingly it does not follow similarly that the promise problem is in $(\coNP\cap \NP)$, because the games are not symmetric.

\begin{remark}[Complexity of approximation for constant value classes]
As a direct consequence we have that for a game $G$ promised to
have at most $K$ value classes, the value of a state can be
approximated in $\FP^{\NP}$. This improves on the $\FNP^{\NP}$ bound
of Frederiksen and Miltersen~\cite{ISAAC:FrederiksenM13} (that holds
in general with no restriction on the number of value classes).
\end{remark}

\section{Non-Zero-sum Concurrent Stochastic Games: Bounds on Patience and Roundedness}
In this section we consider non-zero-sum concurrent stochastic games
where each player has either a reachability or a safety objective. 
We first present a remark on the lower bound in the presence of even a 
single player with reachability objective, and then for the rest of the 
section focus on non-zero-sum games where all players have safety objectives.

\begin{remark}\label{rem:multiple}
In non-zero-sum concurrent stochastic games, with at least two players, even if
there is one player with reachability objectives, then at least doubly-exponential
patience is required for $\eps$-Nash equilibrium strategies.
We have the property if $k=2$ and one player is a reachability player and the other is a safety player, 
from Section~\ref{sec:eps_patience}. 
It is also easy to see that Lemma~\ref{lem:partial_mixed} together with Lemma~\ref{lem:eps_opt_partial_mixed} 
imply that if player~1 is identified with the objective $(\reach,\{\top\})$ and player~2 is identified with the 
objective $(\reach,\{\bot\})$ and they are playing the Purgatory Duel, then  each strategy profile $\sigma$, 
that forms a $\eps$-Nash equilibrium, for any $0<\eps<\frac{1}{3}$, in the Purgatory Duel, has patience $2^{m^{\Omega(n)}}$. 
%%% KRISH: PART BELOW CAN BE OMITTED FOR SHORT VERSION.
This is because player~2 has a harder objective (a subset of the plays satisfies it) than in Section~\ref{sec:eps_patience}, 
but can still ensure the same payoff (by using an optimal strategy for player~2 in the concurrent reachability variant, 
which ensures that $\bot$ is reached with probability at least $\frac{1}{2}$).
In this case, we say that a strategy is optimal (resp., $\eps$-optimal) for a player, if it is optimal (resp., $\eps$-optimal) 
for the corresponding player in the concurrent reachability version. 
It is clear that only if both strategies are optimal (resp., $\eps$-optimal), then the strategies forms a Nash 
equilibrium (resp., $\eps$-Nash equilibrium).
Thus the doubly-exponential lower bound follows even for non-zero-sum games with two reachability players.
The key idea to extend to more players, of which at least one is a reachability player, is as follows: 
Consider some reachability player~$i$. The game for which the lower bound holds can be described as follows. First player~$i$ picks another player~$j$ and they then proceed to play the Purgatory Duel with parameters $n,m$ against each other. This can be captured by a game with $k(2n+1)+3$ states, where each matrix has size at most $\max(m,k)$. Each player must then use doubly-expoential patience in every strategy profile that forms an $\eps$-Nash equilibrium, for sufficently small $\eps>0$. 
First consider a player $j$ that is different from $i$, and a strategy for player~$j$ with low patience.
It follows that player~$i$ would then simply play against player~$j$ and win with good probability. 
Second, consider a strategy for player~$i$ with low patience and there are two cases.
Either player~$i$ gets a payoff close to $\frac{1}{2}$ or not. If he gets a payoff close to $\frac{1}{2}$, 
then the player he is most likely to play against can deviate to an optimal strategy and increase his payoff 
by an amount close to $\frac{1}{2k}$, which player~$i$ loses. 
On the other hand, if player~$i$ gets a payoff far from $\frac{1}{2}$, then he can deviate to 
an optimal strategy and then he gets payoff $\frac{1}{2}$.
\end{remark}

The rest of the section is devoted to non-zero-sum concurrent stochastic games with 
safety objectives for all players, and first we establish an exponential upper bound 
on patience and then an exponential lower bound for $\eps$-Nash equilibrium strategies,
for $\eps>0$.

\subsection{Exponential upper bound on roundedness}
In this section we consider non-zero-sum concurrent safety games, with $k\geq 2$ players, 
and such games are also called stay-in-a-set games, by~\cite{SS02}.
We will argue that, for all $0<\eps<\frac{1}{4}$, in any such game, there exists a strategy profile 
$\sigma$ that forms an $\eps$-Nash equilibrium and have roundedness
at most 
\[
\frac{-32\cdot k^2 \cdot \ln(\eps)\cdot n\cdot (\delta_{\min})^{-n}\cdot m}{\eps}\enspace .
\] 
Note that the roundedness is only exponential, as compared to the doubly-exponential patience when there is at least 
one reachability player (Remark~\ref{rem:multiple}).
Note that the bound is polynomial in $m$ and $k$; and also polynomial in $n$ if 
$\delta_{\min}=1$.

\smallskip\noindent{\bf Players already lost, and all winners.}
For a prefix of a play $P_s^{\ell'}$, for a starting state $s$, play $P_s$ and length $\ell'$, 
let $\widehat{L}(P_s^{\ell'})$ be the set of players that have not lost already in $P_s^{\ell'}$ 
(note that for each $i$, player~$i$ has lost in a play prefix if a state not in $S^i$ has been visited in the prefix). 
Let $P_s^{\ell'}$ be some prefix of a play and we define $W(P_s^{\ell'})$ as the event that each 
player in $\widehat{L}(P_s^{\ell'})$ wins with probability~1.

\smallskip\noindent{\bf Player-stationary strategies.} 
As shown by~\cite{SS02}, there exists a strategy profile $\sigma=(\sigma_i)_i$ that forms a Nash equilibrium. 
They show that the strategy $\sigma_i$, for any player~$i$, in the witness Nash equilibrium strategy profile has the following properties:
For each set of players $\Pi$ and state $s$, there exists a probability distribution $\widehat{\sigma}_i(\Pi,s)$, such that for each prefix of a play $P_s^{\ell'}$, play $P_s$ and length $\ell'$, if $P_s^{\ell'}$ ends in $s'$, we have that $\sigma_i(P_s^{\ell'})=\widehat{\sigma}_i(\widehat{L}(P_s^{\ell'}),s')$ (i.e., the strategy only depends on the players who have not lost yet and the current state). 
Also, there exists some positional strategy $\sigma_i'$, such that $\widehat{\sigma}_i(\Pi,s)=\sigma_i'(s)$, for all $i\not\in \Pi$ (i.e., players who have lost already play some fixed positional strategy). This allows them to only consider the sub-game $G^{\Pi}$, which is the game in which each player~$i$ not in $\Pi$ plays $\sigma_i'$.
Also, if there is a strategy profile which ensures that each player in $\Pi$ wins with probability~1 if the play starts in $s$ of $G^{\Pi}$, then 
the probability distribution $\widehat{\sigma}_i(\Pi,s)$ is pure\footnote{it is not explicitly mentioned in~\cite{SS02} that the distributions are pure, but it follows from the fact that if all players can ensure their objectives with probability~1, then there exists a positional strategy profile ensuring so, by just considering an MDP (with all players together) with a conjunction of safety objectives} and it ensures that the players in $\Pi$ wins with probability~1. 
We call strategies with these properties {\em player-stationary strategies}.

\smallskip\noindent{\bf The real number $\eps$ and the length $\ell$.} 
In the remainder of this section, fix $0<\eps<\frac{1}{4}$ and fix the length $\ell$, 
such that  
\[
\ell=-n\cdot k\cdot \ln(\eps/(4k))\cdot (\delta_{\min})^{-n}\enspace .
\]
We will, in Lemma~\ref{lem:short_dist}, argue that any player-stationary strategy is such that with probability $1-\eps$ 
no player loses after $\ell$ steps.
Also several lemmas in this section will use $\ell$ and $\eps$.

\smallskip\noindent{\bf The event $E(P_s^{\ell'})$.} 
Given a play $P_s$, starting in state $s$ for some $s$ and any $\ell'$, let $E(P_s^{\ell'})$ be the event that either the event  $(\widehat{L}(P_s^{\ell'})\subsetneq \widehat{L}(P_s^{\ell'-1}))$ (i.e., some player lost at the $\ell'$-th step) or the event $W(P_s^{\ell'})$ (i.e., the remaining players win with probability~1) happens.
In \cite[2.1 Lemma]{SS02} they show\footnote{they do not explicitly show that the constant is $1-(\delta_{\min})^{n}$, but it follows easily from an inspection of the proof}:
\begin{lemma}\label{lem:event_happens}
Fix a player-stationary strategy profile $\sigma$.
Let $T\geq 0$ denote a round (or a step of plays).
Let $Y^{T,s}$ be the set of plays, where for all plays $P_s$ in $Y^{T,s}$, 
either the remaining players win with probability~1 in round $T$ 
(i.e., the event $W(P_s^{T})$ happens) 
or some player loses in round $T$ (i.e., the event 
$\widehat{L}(P_s^{T})\subsetneq \widehat{L}(P_s^{T-1})$ happens). 
For a constant $c$ and length $\ell'$, let 
$y_{c,\ell'}=\Pr_{\sigma}[\exists T:\ell'<T\leq \ell'+cn\ \wedge \ P_s\in Y^{T,s}]$ denote
the probability that event $Y^{T,s}$ happens for some $T$ between $\ell'$ and $\ell'+cn$. 
Then, for all constants $c$ and length $\ell'$, we have that 
\[
%\Pr_{\sigma}[\exists T:\ell'<T\leq \ell'+cn\ \wedge \ P_s\in Y^{T,s}]
y_{c,\ell'} \geq 1-(1-(\delta_{\min})^{n})^c\enspace .\]
\end{lemma}
Note that $T$ above depends on the play $P_s$.
It is straightforward that players can lose at most $k$ times in any play $P_s$, 
simply because there are at most $k$ players, and if the remaining players win with 
probability~1 in round $T$, 
then they also win with probability~1 in round $T+1$, by construction of $\sigma$.

\smallskip\noindent{\bf Proof overview.} 
Our proof will proceed as follows. Consider the game, while the players play some player-stationary strategy profile that forms a Nash equilibria. 
First, we show that it is unlikely (low-probability event) that the players do not play positional 
(like they do if the event $W(P_s^{\ell'})$ has happened) after some exponential number of steps.
Second, we show that if we change each of the probabilities used by an exponentially 
small amount as compared to the Nash equilibria, then it is unlikely that 
that there will be a large difference in the first exponentially many steps. 
This allows us to round the probabilities to exponentially small probabilities while the players only 
lose little.

\begin{lemma}\label{lem:short_dist}
 Fix some player-stationary strategy profile $\sigma$. 
Consider the set $P$ of plays $P_s$, under $\sigma$, such that $W(P_s^{\ell})$ does not happen. 
Then, the probability $\Pr_{\sigma}[P]$ is less than $\eps/4$.
\end{lemma}
\begin{proof}
Fix $0<\eps<\frac{1}{2}$ and a player-stationary strategy profile $\sigma$. 
Let $c=-\ln(\eps/(4k))\cdot (\delta_{\min})^{-n}>1$. We will argue that the event $E(P_s^{\ell'})$ happens at least $k$ times with probability at least $1-\eps/4$ over $c\cdot n\cdot k=\ell$ steps.

We consider two cases, either $\delta_{\min}=1$ or $0<\delta_{\min}<1$.
If $\delta_{\min}=1$, the event $\exists 1\leq T\leq  n:E(P_s^{\ell'+T})$ always happens (otherwise, in case it did not in some play, then a deterministic cycle 
satisfying the safety objectives of all players who have not lost yet is executed, and then the players could win by playing whatever they did the last time they were in a given state). If $0<\delta_{\min}<1$, we see that $c\geq c'=\frac{\ln(\eps/(4k))}{\ln(1-(\delta_{\min})^{-n})}$, since $1+x\leq e^x$ and that $\exists 1\leq T\leq  c'\cdot n:E(P_s^{\ell'+T})$ happens with probability at least $1-\eps/(4k)$ by Lemma~\ref{lem:event_happens}. In either case, we have that the event $\exists 1\leq T\leq  c\cdot n:E(P_s^{\ell'+T})$ happens with probability at least $1-\eps/(4k)$.

Next, split the plays up in epochs of length $c\cdot n$ each, and we get that the event $E(P_s^{T})$ happens at least once for $T$ ranging over the steps of an epoch with probability at least $1-\eps/(4k)$ and hence happens at least once in each of the first $k$ epochs with probability at least $1-\eps/4$ using union bound. At that point the remaining players win with probability~1. The first $k$ epochs have length $c\cdot k\cdot n=\ell$ and the lemma follows.
\end{proof}

%%Similar to ideas in~\cite{CI14}, 
We use the above lemma to show that any strategy profile 
close to a Nash equilibrium ensures payoffs close to that equilibrium. 
To do so, we use coupling (similar to~\cite{CI14}).

\smallskip\noindent{\bf Variation distance.} 
The {\em variation distance} is a measure of the similarity between two distributions.
Given a finite set $Z$, and two distributions $d_1$ and $d_2$ over $Z$, 
the variation distance of the distributions is 
\[
\var(d_1,d_2)=\frac{1}{2}\cdot \sum_{z\in Z} |d_1(z)-d_2(z)| \enspace .
\]
We will extend the notion of variation distances to strategies as follows:
Given two strategies $\sigma_i$ and $\sigma_i'$ for player~$i$ the variation distance between the strategies is \[
\var(\sigma_i,\sigma_i')=\sup_{P_s^\ell}\var(\sigma_i(P_s^\ell),\sigma_i'(P_s^\ell)) \enspace ;
\]
i.e., it is the supremum over the variation distance of the distributions used by the strategies for finite-prefixes of plays.

\smallskip\noindent{\bf Coupling and coupling lemma.}
Given a pair of distributions, a coupling is a probability distribution over the joint set of possible outcomes. 
Let $Z$ be a finite set. For distributions $d_1$ and $d_2$ over the finite set $Z$, 
a {\em coupling} $\omega$ is a distribution over $Z\times Z$, such that for all $z\in Z$ we have 
$\sum_{z'\in Z} \omega(z,z')=d_1(z)$ and also for all $z'\in Z$ we have 
$\sum_{z\in Z} \omega(z,z')=d_2(z')$. 
One of the most important properties of coupling is the coupling lemma~\cite{aldous} of which we only mention and use the second part:
\begin{itemize}
\item {\bf (Coupling lemma).}
For a pair of distributions 
$d_1$ and $d_2$, there exists a coupling $\omega$ of $d_1$ and $d_2$, such that for a random variable 
$(X,Y)$ from the distribution $\omega$, we have that $\var(d_1,d_2)=\Pr[X\neq Y]$.
\end{itemize}

\smallskip\noindent{\bf Smaller support.}
Fix a pair of strategies $\sigma_i$ and $\sigma_i'$ for player~$i$ for some $i$. 
We say that $\sigma_i'$ has {\em smaller support} than $\sigma_i$, if for all $P_s^\ell$ we have that \[
\supp(\sigma_i'(P_s^\ell))\subseteq \supp(\sigma_i(P_s^\ell))\enspace .\]

\begin{lemma}\label{lem:close_outcome}
Let $\sigma=(\sigma_i)_i$ and $\sigma'=(\sigma_i')_i$ be player-stationary strategy profiles, such that 
\[\var(\sigma,\sigma')\leq \frac{\eps}{\ell\cdot k\cdot 4}\enspace ,\] and such that $\sigma_i'$ has smaller support than $\sigma_i$, for all $i$.
 Then $\sigma'$ is such that \[u(G,s,\sigma',i)\in [u(G,s,\sigma,i)-\eps/2,u(G,s,\sigma,i)+\eps/2]\] for each player~$i$ and state~$s$.
\end{lemma}

\begin{proof}
Fix $\sigma$ and $\sigma'$ according to the lemma statement. 
For any prefix of a play $P_s^{\ell'}$, for any state $s$ and length $\ell'$ and player~$i$, we have that $\var(\sigma_i(P_s^{\ell'}),\sigma_i'(P_s^{\ell'}))\leq \frac{\eps}{\ell\cdot k\cdot 4}$ and thus, we can create a coupling $\omega=(X_i^{P_s^{\ell'}},Y_i^{P_s^{\ell'}})$ between the two distributions $\sigma_i(P_s^{\ell'})$ and $\sigma_i'(P_s^{\ell'})$, i.e.,  $X_i^{P_s^{\ell'}}\sim \sigma_i(P_s^{\ell'})$ and $Y_i^{P_s^{\ell'}}\sim\sigma_i'(P_s^{\ell'})$ is such that $\Pr[X_i^{P_s^{\ell'}}\neq Y_i^{P_s^{\ell'}}]\leq \frac{\eps}{\ell\cdot k\cdot 4}$. Then, consider some state $s$ and consider a play $P_s$, picked using the random variables $X_i^{P_s^{\ell'}}$, and a play $Q_s$, picked using the random variables $Y_i^{P_s^{\ell'}}$ (where, if the players uses the same action in $P_s^{\ell'}$ and $Q_s^{\ell'}$, then the next state is also the same, using an implicit coupling).  
Then according to Lemma~\ref{lem:short_dist}, the probability that $W(P_s^\ell)$ occurs is at least $1-\eps/4$. In that case, we are interested in the probability that $Q_s=P_s$. Observe that we just need to ensure that $P_s^\ell$ and $Q_s^\ell$ are the same, since at that point the players play according to the same positional strategy, because of the smaller support. For each $\ell''\leq \ell$, if the first $\ell''$ steps match, then the next step match with probability at least $1-\frac{\eps}{\ell\cdot k\cdot 4}\cdot k$, since each of the $k$ players has a probability of $\frac{\eps}{\ell\cdot k\cdot 2}$ to differ in the two plays. Hence, all $\ell$ steps match with probability at least $1-\frac{\eps}{\ell\cdot k\cdot 4}\cdot \ell \cdot k=1-\eps/4$. Hence, with probability at least $1-\eps/2$ we have that $P_s$ equals $Q_s$ and thus, especially, the payoff for each player must be the same in that case. But observe that $P_s$ is distributed like plays under $\sigma$ and $Q_s$ is distributed like plays under $\sigma'$ and the statement follows.\end{proof}

We will next show that we only need to consider deviations to player-stationary strategies for the purpose of player-stationary equilibria.

\begin{lemma}\label{lem:player-stationary_deviation}
For all player-stationary strategy profiles $\sigma$ and each player~$i$, there exists a pure player-stationary strategy $\sigma_i'$ for player~$i$ maximizing $u(G,s,\sigma[\sigma_i'],i)$.
\end{lemma}
\begin{proof}
Observe first that it does not matter what player~$i$ does if he has already lost, and we can consider him to play some fixed positional strategy in that case. Also, when the remaining players play according to $\sigma$, we can view the game as being an MDP, in the games $G^\Pi$. The objective of player~$i$ is then to reach a sub-game of $G^\Pi$ and a state in that sub-game, from which he cannot lose. But it is well-known that such reachability objectives have positional optimal strategies in MDPs. Hence, this strategy forms a pure player-stationary strategy in the original game.
\end{proof}

We will use Lemma~3 from \cite{CI14}. The proof only appears in \cite{CI14Full}, where the lemma is Lemma~4.

\begin{lemma}\label{lem:round}{\bf\em (Lemma 3, \cite{CI14}).} Let $Z$ be a set of size $\ell$.
Let $d_1$ be some distribution over $Z$ and let $q\geq\ell$ be some integer.
Then there exists some distribution $d_2$, such that for each $z\in Z$, there exists an integer $p$ such that $d_2(z)=\frac{p}{q}$ and such that $|d_1(z)-d_2(z)|<\frac{1}{q}$.
\end{lemma}

We are now ready to show the main theorem of this section.

\begin{theorem}\label{thm:safety_only_upperbound}
For all concurrent stochastic games with all $k$ safety players, for all $0<\eps<\frac{1}{4}$, 
there exists a player-stationary strategy profile $\sigma$ that forms an $\eps$-Nash equilibrium 
and has roundedness at most 
\[4 n\cdot k^2\cdot m\cdot \eps^{-1} \cdot \ln(4k/\eps)\cdot (\delta_{\min})^{-n}\enspace .\]
\end{theorem}
\begin{proof}
Fix some player-stationary strategy profile $\sigma$ that forms a Nash-equilibrium and some $0<\eps<\frac{1}{4}$ and let  \[
\ell:=-n\cdot k\cdot \ln(\eps/(4k))\cdot (\delta_{\min})^{-n}\enspace .\]

Consider some distribution $d_1$ over some set $Z$. Observe that for each distribution $d_2$ with smaller support than $d_1$ and such that $|d_1(z)-d_2(z)|<\frac{1}{q}$, for each $z\in \supp(d_1)$, we have $\var(d_1,d_2)\leq \frac{|\supp(d_1)|}{q}$. 
Then, applying Lemma~\ref{lem:round}, for $q=\frac{\ell\cdot k\cdot 4\cdot m}{\eps}$ and $Z=\supp(d)$, to each probability distribution $d$ defining $\sigma$, we see that there exists a player-stationary strategy profile $\sigma'=(\sigma_i')_i$, such that (1) \[\var(\sigma,\sigma')\leq \frac{m}{q}=\frac{\eps}{\ell\cdot k\cdot 4}\enspace ;\] and (2)~$\sigma_i'$ has smaller support than $\sigma_i$; and (3)~$\sigma_i'(P_s^\ell)$ is a fraction with denominator $q$. Observe that the strategy has roundedness $q$.

We now argue that $\sigma'$ is an $\eps$-Nash equilibrium.
Consider some player~$i$ and a player-stationary strategy $\sigma_i''$ maximizing the probability that player~$i$ wins when the remaining players play according to $\sigma'$, which is known to exists by Lemma~\ref{lem:player-stationary_deviation}. From Lemma~\ref{lem:close_outcome}, we have that \[u(G,s,\sigma[\sigma_i''],i)\geq u(G,s,\sigma'[\sigma_i''],i)-\eps/2\] and \[u(G,s,\sigma,i)\leq u(G,s,\sigma',i)+\eps/2\enspace .\] Thus, $u(G,s,\sigma',i)\geq  u(G,s,\sigma'[\sigma_i''],i)-\eps$. This completes the proof.
\end{proof}

\begin{remark}[Finding an $\eps$-Nash equilibria in $\TFNP$]\label{rem:find_safety_only_nash} 
We explain how the results of this section imply that for non-zero-sum concurrent 
stochastic games with safety objectives for all players, if the number $k$ of players is 
only a constant or logarithmic, then we can compute an $\eps$-Nash equilibria in $\TFNP$, 
where $\eps>0$ is given in binary as part of the input. 
Note that there is a polynomial-size witness (to guess) for a stationary strategy with 
exponential roundedness. Observe that a player-stationary strategy for a player is 
defined by $2^{k-1}+1$ stationary strategies, one used in case that the respective player has 
lost, and one for each subset of other players. 
Thus, we can guess polynomial-size witnesses of $k$ player-stationary strategies with exponential roundedness, 
given that the number of players is at most logarithmic in the size of the input. 
Hence, according to Theorem~\ref{thm:safety_only_upperbound}, we can guess a candidate strategy profile 
$\sigma$ that forms an $\eps$-Nash equilibrium in non-deterministic polynomial time. 
For each player~$i$, constructing the (polynomial-sized) MDP described in the proof of 
Lemma~\ref{lem:player-stationary_deviation} and then solving it using linear programming 
gives us the payoff of playing the strategy maximizing the value for player~$i$ 
while the remaining players follows $\sigma$. 
If, for each player~$i$, the payoff only differs at most $\eps$ from what achieved by 
player~$i$ when all players follows $\sigma$, then the strategy profile $\sigma$ is an 
$\eps$-Nash equilibrium.
It follows that the approximation of some $\eps$-Nash equilibria can be achieved in $\TFNP$,
given that the number of players is at most logarithmic.
\end{remark}

\subsection{Exponential lower bound on patience}
In this section, we show that $\Omega((\delta_{\min})^{-(n-3)/6})$ patience is required, 
for each strategy profile that forms an $\eps$-Nash equilibrium, for any $0<\eps<\frac{1}{6}$, 
in a family of games $\{G_c^{(\delta_{\min})}\mid c\in \N\wedge \delta_{\min}<6^{-3}\}$ with two safety players. 

\smallskip\noindent{\bf Game family $G_c^{\delta_{\min}}$.}
For a fixed number $c\geq 1$ and $0<\delta_{\min}<6^{-3}$, the game $G_c^{\delta_{\min}}$ is defined as follows:
There are $n=4\cdot c+3$ states, namely, $S=\{v_s,v_1,v_2,\top,\bot\}\cup\{v_j^\ell\mid j\in\{1,2\}\wedge \ell\in\{1,\dots,2\cdot c-1\}\}$.
For player~$i$ in state $v_j$, for $j=1,2$, there are two actions, called $a_i^{j,1}$ and $a_i^{j,2}$, respectively. 
For each other state $s$ and each player~$i$, there is a single action, $a$.
For simplicity, for each pair of states $s,s'$ we write $d(s,s')$ for the probability distribution, where $d(s,s')(s)=1-\delta_{\min}$ and $d(s,s')(s')=\delta_{\min}$. Also, we define $v_1^0$ as $\top$ and $v_2^0$ as $\bot$.
The states $\bot$ and $\top$ are absorbing.
The state $v_s$ is such that\footnote{recall that $\U(s,s')$ is the uniform distribution over $s$ and $s'$} $\delta(v_s,a,a)=\U(v_1,v_2)$.
For each $j\in \{1,2\}$, the transition function of state $v_j$ is 
\[
\delta(v_j,a_1^{j,\ell},a_2^{j,\ell'})=\begin{cases}
d(v_s,v_j^{c-1})&\text{if }\ell=\ell'\\ 
d(v_s,v_{\widehat{j}}^{2c-1})&\text{if }\ell<\ell'\\
v_{\widehat{j}}^0&\text{if }\ell>\ell'\\
\end{cases}
\]
For each other state $v_j^\ell$, the transition function is $
\delta(v_j^\ell,a,a)=
d(v_s,v_j^{\ell-1})$.
The objective of player~$1$ is $(\safety,S\setminus \{\bot\})$ and the objective of player~2 
is $(\safety,S\setminus \{\top\})$. See Figure~\ref{fig:gcd} for an illustration of $G_2^{\delta_{\min}}$.

\begin{figure}
\center
\begin{tikzpicture}[node distance=3cm,-{stealth},shorten >=2pt]
\ma{top}[$\top$]{1}{1};

\ma[shift={($(top)+(2cm,0)$)}]{v11}[$v_1^1$]{1}{1};
\ma[shift={($(v11)+(2cm,0)$)}]{v12}[$v_1^2$]{1}{1};
\ma[shift={($(v12)+(2cm,0)$)}]{v13}[$v_1^3$]{1}{1};

\ma[shift={($(top)+(0,-4cm)$)}]{v1}[$v_1$]{2}{2};

\ma[shift={($(v1)+(3cm,0)$)}]{vs}[$v_s$]{1}{1};

\ma[shift={($(v13)+(0,-4cm)$)}]{v2}[$v_2$]{2}{2};

\ma[shift={($(v2)+(0,-4cm)$)}]{bot}[$\bot$]{1}{1};
\ma[shift={($(bot)+(-2cm,0)$)}]{v21}[$v_2^1$]{1}{1};
\ma[shift={($(v21)+(-2cm,0)$)}]{v22}[$v_2^2$]{1}{1};
\ma[shift={($(v22)+(-2cm,0)$)}]{v23}[$v_2^3$]{1}{1};

\draw[dashed] (v1-1-1.center) to[out=90,in=135] (vs);
\draw[dotted] (v1-1-1.center) to[out=90,in=-135] (v11);

\draw[dashed] (v1-1-2.center) to[out=0,in=155] (vs); 
\draw[dotted] (v1-1-2.center) to[out=0,in=45,min distance=2.2cm] (v23);

\draw[dashed] (v1-2-2.center) to[out=0,in=-155] (vs);
\draw[dotted] (v1-2-2.center) to[out=0,in=-90] ($(v11.south)!.5!(v11.south east)$);

\draw (v1-2-1.center) to[out=-90,in=135] (bot);

\draw (vs.center) to[out=-90,in=0] node[pos=0.8,anchor=center,fill=white,draw=black,inner sep=0.2pt,circle] {$\frac{1}{2}$}(v1-2-2.south east);

\draw (vs.center) to[out=-90,in=180] node[pos=0.7,anchor=center,fill=white,draw=black,inner sep=0.2pt,circle] {$\frac{1}{2}$}(v2-2-1.south west);

\draw[dashed] (v2-1-1.center) to[out=180,in=15] (vs);
\draw[dotted] (v2-1-1.center) to[out=180,in=135] (v21);

\draw[dashed] (v2-1-2.center) to[out=90,in=45] (vs); 
\draw[dotted] (v2-1-2.center) to[out=90,in=-90] ($(v13.south)!0.5!(v13.south east)$);

\begin{scope}[yscale=-1,xscale=-1]
\draw[dashed] (v2-2-2.center) arc (180:0:1.75cm);
\end{scope}
\draw[dotted] (v2-2-2.center) to[out=-90,in=45] (v21);
\draw (v2-2-1.center) to  (top);

\nloop{top}[xscale=-1,yscale=-1]
\nloop{bot}

\begin{scope}[yscale=-1,xscale=-1]
\draw[dotted] (v11.center) to (v11.south) arc (180:0:0.85cm);
\end{scope}
\draw[dashed] (v11.center) to[out=-90,in=110] (vs);

\begin{scope}[yscale=-1,xscale=-1]
\draw[dotted] (v12.center) to (v12.south) arc (180:0:0.75cm);
\end{scope}
\draw[dashed] (v12.center) to[out=-90,in=70] (vs);

\begin{scope}[yscale=-1,xscale=-1]
\draw[dotted] (v13.center) to (v13.south) arc (180:0:0.85cm);
\end{scope}
\draw[dashed] (v13.center) to (v13.south) to[out=-90,in=57.5] (vs);

\draw[dotted] (v21.center) to (v21.north) arc (180:0:0.85cm);
\draw[dashed] (v21.center) to[out=90,in=-75] (vs);

\draw[dotted] (v22.center) to (v22.north) arc (180:0:0.85cm);
\draw[dashed] (v22.center) to[out=90,in=-105] (vs);

\draw[dotted] (v23.center) to (v23.north) arc (180:0:0.85cm);
\draw[dashed] (v23.center) to (v23.north) to[out=90,in=-130] (vs);
\end{tikzpicture}
\caption{An illustration of the game $G_2^{\delta_{\min}}$. The probabilities are as follows: The probability of each dashed edge is $1-\delta_{\min}$; and the probability of each dotted edge is $\delta_{\min}$; and the probability of each solid edge is~1. The only exception is the edges from $v_s$, where the probability is written on each edge (it is $\frac{1}{2}$ in each case).\label{fig:gcd}}
\end{figure}
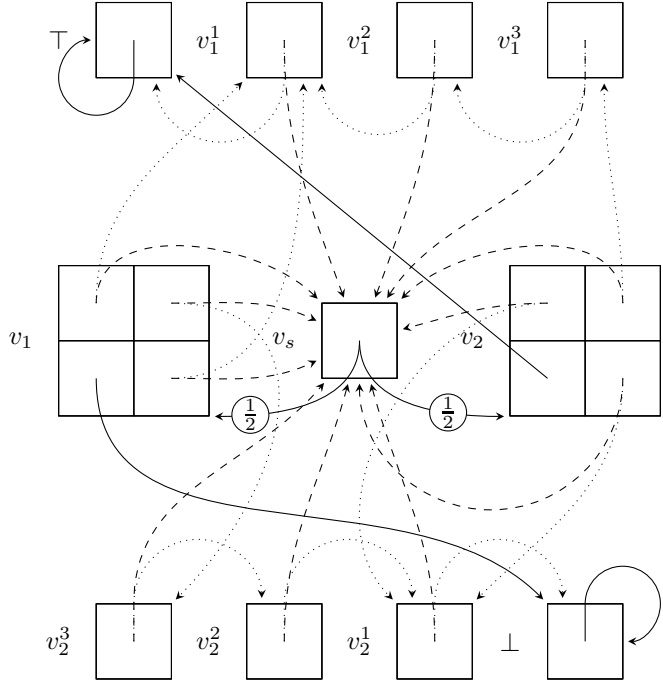

\smallskip\noindent{\bf Near-zero-sum property.} 
Observe that either $\bot$ or $\top$ is reached with probability~1 (and once $\top$ or $\bot$ is reached, the game stays there). 
The reasoning is as follows: there is a probability of at least $(\delta_{\min})^{2c}$ to reach either $\top$ or $\bot$ 
within the next $2c+1$ steps from any state.
If the current state is $v_s$, then the next state is either $v_1$ or $v_2$, and 
from $v_1$ or $v_2$  through $v_j^{\ell}$ for each $\ell$ from $1$ to $2c-1$, for some $j$, 
either $\top$ or $\bot$ is reached, and each of the steps from $v_1$ or $v_2$ onward happens 
with probability at least~$\delta_{\min}$, no matter the choice of the players. 
Hence, the game is in essence zero-sum, since with probability~1 precisely one player wins. 

\smallskip\noindent{\bf Proof overview.}
Our proof has two parts. We show that there is a strategy for player~$i$, for each~$i$, that ensures that against all strategies for the other player, the payoff is at least $\frac{1}{2}$ for player~$i$. Also, we show that for each strategy of player~$i$ with patience at most~$(\delta_{\min})^{-2/3\cdot c}$, there is a strategy for the other player such that the payoff is less than $\frac{1}{6}$ for player~$i$. 
This then allows us to show that no strategy profile that forms a $\frac{1}{6}$-Nash equilibrium has patience less than 
$(\delta_{\min})^{-2/3\cdot c}$.

\begin{lemma}\label{lem:safe_optimal}
For each $i$, player~$i$ has a strategy $\sigma_i$ such that \[\inf_{\sigma_{\widehat{i}}}u(G,v_s,\sigma_1,\sigma_2,i)=\frac{1}{2}\enspace .\]
\end{lemma}
\begin{proof}
Consider the stationary strategy $\sigma_1$, where \[\sigma_1(v_1)(a_1^{1,1})=\sigma_1(v_2)(a_1^{2,1})=\frac{1+(\delta_{\min})^{-c}}{2+(\delta_{\min})^{-c}+(\delta_{\min})^{c}}\]
and
\[\sigma_1(v_1)(a_1^{1,2})=\sigma_1(v_2)(a_1^{2,2})=\frac{1+(\delta_{\min})^{c}}{2+(\delta_{\min})^{-c}+(\delta_{\min})^{c}}\enspace .\]
Observe that fixing $\sigma_1$ as the strategy for player~1, the game turns into an MDP for player~2. Such games have a positional strategy ensuring that the payoff for player~2 is as large as possible.  Going through all four candidates for $\sigma_2$, one can see that $\max_{\sigma_2}u(G,v_s,\sigma_1,\sigma_2,2)=\frac{1}{2}$. Because of the near-zero-sum property, this minimizes the payoff for player~1 (since $u(G,v_s,\sigma_1,\sigma_2,1)+u(G,v_s,\sigma_1,\sigma_2,2)=1$), 
which is then $\inf_{\sigma_2}u(G,v_s,\sigma_1,\sigma_2,1)=\frac{1}{2}$. 
The strategy for player~2 follows from $\sigma_1$ and the symmetry of the game.
\end{proof}

We next argue that if player~$i$ uses a low-patience strategy, then the opponent can ensure low payoff 
for player~$i$.
\begin{lemma}\label{lem:low_outcome}
Let $\sigma_i$ be a strategy for player~$i$ with patience at most $(\delta_{\min})^{-2/3\cdot c}$. Then there exists a pure strategy $\sigma_{\widehat{i}}$ such that $u(G,v_s,\sigma_1,\sigma_2,\widehat{i})> 1-\frac{1}{6}$.
\end{lemma}
\begin{proof}
Consider first player~$1$ (the argument for player~2 follows from symmetry). Let $\sigma_1$ be some strategy with patience at most $(\delta_{\min})^{-(n-3)/6}=(\delta_{\min})^{-2/3\cdot c}$.

The pure strategy $\sigma_2$ is defined given $\sigma_1$ as follows. For plays $P_s^{\ell}$ ending in state $v_1$ or $v_2$ we have that 
\[
\sigma_2(P_s^{\ell})=\begin{cases}
a_2^{j,j} &\text{if }\sigma_1(P_s^{\ell})(a_2^{j,2})>0\\
a_2^{j,\widehat{j}} &\text{if }\sigma_1(P_s^{\ell})=a_2^{j,1}\enspace .
\end{cases}
\]
To argue that $u(G,v_s,\sigma_1,\sigma_2,2)> 1-\frac{1}{6}$, 
we consider a play $P_{v_s}$ picked according to $(\sigma_1,\sigma_2)$, 
such that either $\bot$ or $\top$ is eventually reached. 
This is true with probability~1. 
Consider the last round $\ell$, such that $v_\ell=v_j$, for some $j=1,2$. 
 We now consider four cases: 
Either we have that
\begin{enumerate}
\item $j=1$ and $\sigma_1(P_s^{\ell})(a_2^{j,2})>0$ or
\item $j=1$ and $\sigma_1(P_s^{\ell})=a_2^{j,1}$ or
\item $j=2$ and $\sigma_1(P_s^{\ell})(a_2^{j,2})>0$ or
\item $j=2$ and $\sigma_1(P_s^{\ell})=a_2^{j,1}$.
\end{enumerate} 
 The probability to eventually reach $\bot$ is then at least the minimum probability to eventually reach $\bot$ in each of the four cases. In case (2) and case (4), we see that player~2 wins with probability~1. In case (1) observe that from a round $\ell'$ where $\sigma_1(P_s^{\ell'})(a_2^{1,2})>0$ player~1 wins (i.e., reaches $\top$ before entering $v_s$ again) with probability $(1-(\delta_{\min})^{2/3\cdot c})\cdot (\delta_{\min})^{c}<(\delta_{\min})^{c}$ and player~2 wins (i.e., reaches $\bot$ before entering $v_s$ again) with probability  $(\delta_{\min})^{2/3\cdot c}$. Hence, the probability that player~1 wins if such a round is round $\ell$ is at most 
\begin{align*}
\frac{(\delta_{\min})^{c}}{(\delta_{\min})^{2/3\cdot c}+(\delta_{\min})^{c}} < \frac{(\delta_{\min})^{c}}{(\delta_{\min})^{2/3\cdot c}}=(\delta_{\min})^{c/3}<\frac{1}{6}\enspace ,
\end{align*} 
where the last inequality comes from that $c\geq 1$ and $\delta_{\min}<6^{-3}$. In case (3) observe that from a round $\ell'$ where $\sigma_1(P_s^{\ell'})(a_2^{2,2})>0$ player~1 wins (i.e., reaches $\top$ before entering $v_s$ again) with probability at most $(1-(\delta_{\min})^{2/3\cdot c})\cdot (\delta_{\min})^{2c}<(\delta_{\min})^{2c}$ and player~2 wins (i.e., reaches $\bot$ before entering $v_s$ again) with probability  at least $(\delta_{\min})^{2/3\cdot c}\cdot (\delta_{\min})^{c}=(\delta_{\min})^{5/3\cdot c}$. Hence, the probability that player~1 wins if such a round is round $\ell$ is at most 
\begin{align*}
\frac{(\delta_{\min})^{2\cdot c}}{(\delta_{\min})^{5/3\cdot c}+(\delta_{\min})^{2\cdot c}} < \frac{(\delta_{\min})^{2\cdot c}}{(\delta_{\min})^{5/3\cdot c}}=(\delta_{\min})^{c/3}<\frac{1}{6}\enspace ,
\end{align*} 
where the last inequality comes from that $c\geq 1$ and $\delta_{\min}<6^{-3}$.
The desired result follows.
\end{proof}

We now prove the main result that no strategy with patience only $(\delta_{\min})^{-2/3\cdot c}$ can be a part 
of a $\frac{1}{6}$-Nash equilibrium.

\begin{theorem}
For all $c\in \N$ and all $0<\delta_{\min}<6^{-3}$, 
consider the game $G_c^{\delta_{\min}}$ (that has $n=4c+3$ states
and at most two actions for each player at all states).
Each strategy profile $\sigma=(\sigma_i)_i$ that forms an $\frac{1}{6}$-Nash equilibrium 
has patience at least $(\delta_{\min})^{-(n-3)/6}$.
\end{theorem}
\begin{proof}
Fix some $c\in \N$ and $0<\delta_{\min}<6^{-3}$. 
The proof will be by contradiction. Consider first player~$1$ (the argument for player~2 follows from symmetry). Let $\sigma_1$ be some strategy with patience at most $(\delta_{\min})^{-(n-3)/6}=(\delta_{\min})^{-2/3\cdot c}$.

Consider some strategy $\sigma_2$ for player~2.
We consider two cases, either \[u(G,v_s,\sigma_1,\sigma_2,2)\leq \frac{1}{2}+\frac{1}{6}=\frac{2}{3}\] or not.
 If \[u(G,v_s,\sigma_1,\sigma_2,2)\leq \frac{2}{3}\enspace ,\] then player~2 can play a strategy $\sigma_2'$, shown to exist in Lemma~\ref{lem:low_outcome}, instead and get payoff strictly above $1-\frac{1}{6}=\frac{5}{6}$, showing that $(\sigma_1,\sigma_2)$ is not an $\frac{1}{6}$-Nash equilibrium. On the other hand, if  
\[u(G,v_s,\sigma_1,\sigma_2,2)> \frac{2}{3}\enspace ,\]
then $u(G,v_s,\sigma_1,\sigma_2,1)< \frac{1}{3}$ and player~1 can play a strategy $\sigma_1'$, shown to exist in Lemma~\ref{lem:safe_optimal}, for which $u(G,v_s,\sigma_1',\sigma_2,1)\geq \frac{1}{2}$. Hence, $(\sigma_1,\sigma_2)$ does not form an $\frac{1}{6}$-Nash equilibrium in this case either.
The desired result follows.
\end{proof}

%%% KRISH: CAN BE OMITTED FOR SHORT VERSION.
\begin{remark}
Using ideas similar to Remark~\ref{rem:multiple} we can construct a game with $k\geq 3$ 
safety players in which the patience is at least $(\delta_{\min})^{-(n-3)/(6k)}$ for all strategy profiles that forms an $\frac{1}{6k}$-Nash equilibrium.
\end{remark}

\section{Discussion and Conclusion}\label{sec:discussion}
In this section, we discuss some important features and interesting technical 
aspects of our results.
Finally we conclude with some remarks.

\subsection{Important features of results}\label{subsec:imp-features}

\begin{figure}
\center
\resizebox{0.9\textwidth}{!}{
\begin{tikzpicture}[node distance=3cm,-{stealth},shorten >=2pt]
\ma{top}[$\top$]{1}{1};

\node[rectangle, minimum height=7 cm,minimum width=5 cm,draw] (pn1m) at (1.5cm,-5.5cm){Purgatory $(n+1,m)$};

\ma[shift={($(pn1m.south)-(1.5cm,2cm)$)}]{bot}[$\bot$]{1}{1};

\nloop{top};
\nloop{bot}[yscale=-1];
\draw (pn1m.260) to[out=-90,in=90] (bot);
\draw (pn1m.100) to[out=90,in=-90] (top);

\ma[shift={($(pn1m.south)+(2cm,-2cm)$)}]{v21}[$v_2^1$]{2}{2};
\ma[shift={($(v21)+(0,-4cm)$)}]{v11}[$v_1^1$]{2}{2};

\draw (v11-1-1.center) to (v21);
\draw (v11-2-1.center) to[out=180,in=-90] (bot.250);

\nloop{v11-1-2};
\draw (v11-2-2.center) to[out=90,in=-90] (v21);
\draw (pn1m.270) to[out=-90,in=135] (v11);

\draw (v21-1-1.center) to[out=90,in=-90] (pn1m.298);
\draw (v21-2-1.center) to[out=180,in=0] (bot.25);
\draw (v21-1-2.center) to[out=0,in=0] (v11.0);

\draw (v21-2-2.center) to[out=90,in=-90] (pn1m.301);
\draw (pn1m.270) to[out=-90,in=135] (v11);

\draw (5cm,-5cm) to node[midway,above] {Simplify} (7cm,-5cm) ;
\begin{scope}[shift={(9cm,0)}]

\ma{top}[$\top$]{1}{1};

\node[rectangle, minimum height=7 cm,minimum width=5 cm,draw] (pn1m) at (1.5cm,-5.5cm){Simplified purgatory $(n,m)$};

\ma[shift={($(pn1m.south)-(1.5cm,2cm)$)}]{bot}[$\bot$]{1}{1};

\nloop{top};
\nloop{bot}[yscale=-1];
\draw (pn1m.260) to[out=-90,in=90] (bot);
\draw (pn1m.100) to[out=90,in=-90] (top);

\ma[shift={($(pn1m.south)+(2cm,-2cm)$)}]{v21}[$v_2^1$]{2}{2};
\ma[shift={($(v21)+(0,-4cm)$)}]{v11}[$v_s$]{1}{1};
\draw (v11.center) to (v21);

\draw (v21-1-1.center) to[out=90,in=-90] (pn1m.298);
\draw (v21-2-1.center) to[out=180,in=0] (bot.25);
\draw (v21-1-2.center) to[out=0,in=0] (v11.0);

\draw (v21-2-2.center) to[out=90,in=-90] (pn1m.301);
\draw (pn1m.270) to[out=-90,in=135] (v11);

\draw (5cm,-5cm) to node[midway,above] {Exchange $\top\leftrightarrow\bot$} (7cm,-5cm) ;
\end{scope}

\begin{scope}[shift={(18cm,0)}]

\ma{top}[$\bot$]{1}{1};

\node[rectangle, minimum height=7 cm,minimum width=5 cm,draw] (pn1m) at (1.5cm,-5.5cm){Simplified purgatory dual $(n,m)$};

\ma[shift={($(pn1m.south)-(1.5cm,2cm)$)}]{bot}[$\top$]{1}{1};

\nloop{top};
\nloop{bot}[yscale=-1];
\draw (pn1m.260) to[out=-90,in=90] (bot);
\draw (pn1m.100) to[out=90,in=-90] (top);

\ma[shift={($(pn1m.south)+(2cm,-2cm)$)}]{v21}[$v_2^2$]{2}{2};
\ma[shift={($(v21)+(0,-4cm)$)}]{v11}[$v_s$]{1}{1};
\draw (v11.center) to (v21);

\draw (v21-1-1.center) to[out=90,in=-90] (pn1m.298);
\draw (v21-2-1.center) to[out=180,in=0] (bot.25);
\draw (v21-1-2.center) to[out=0,in=0] (v11.0);

\draw (v21-2-2.center) to[out=90,in=-90] (pn1m.301);
\draw (pn1m.270) to[out=-90,in=135] (v11);

\end{scope}

\begin{scope}[shift={(13.5cm,-18cm)}]

\ma{top}[$\top$]{1}{1};

\node[rectangle, minimum height=7 cm,minimum width=5 cm,draw] (pn1m) at (1.5cm,-5.5cm){Purgatory duel $(n,m)$};

\ma[shift={($(pn1m.south)-(1.5cm,2cm)$)}]{bot}[$\bot$]{1}{1};

\nloop{top};
\nloop{bot}[yscale=-1];
\draw (pn1m.260) to[out=-90,in=90] (bot);
\draw (pn1m.100) to[out=90,in=-90] (top);

\ma[shift={($(pn1m.south)+(2cm,-2cm)$)}]{v12}[$v_2^1$]{2}{2};
\ma[shift={($(pn1m.south)+(-5cm,-2cm)$)}]{v22}[$v_2^2$]{2}{2};

\ma[shift={($(bot)+(0,-4cm)$)}]{vs}[$v_s$]{1}{1};

\draw (v22-1-1.center) to[out=90,in=180] (pn1m.210);
\draw (v22-2-1.center) to[out=180,in=180] (top.210);
\draw (v22-1-2.center) to[out=-30,in=180] (vs.150);
\draw (v22-2-2.center) to[out=0,in=-90] (pn1m.240);

\draw (v12-1-1.center) to[out=90,in=-90] (pn1m.298);
\draw (v12-2-1.center) to[out=180,in=0] (bot.25);
\draw (v12-1-2.center) to[out=-20,in=0] (vs.south east);

\draw (v12-2-2.center) to[out=90,in=-90] (pn1m.301);

\draw (vs.center) to[out=90,in=-45] node[above] {$\frac{1}{2}$} (v22-2-1);
\draw (vs.center) to[out=90,in=-135] node[above] {$\frac{1}{2}$} (v12);

\draw (pn1m.270) to[out=-90,in=45] (vs);
\end{scope}
\draw[-,shorten >=0] (13.7cm,-10cm) -- (15.2cm,-14cm);
\draw[-,shorten >=0] (16.7cm,-10cm) -- (15.2cm,-14cm);
\draw (15.2cm,-14cm) to node[above,sloped] {Merge $v_s,\top,\bot$} (15.2cm,-19cm);

\end{tikzpicture}
}
\caption{An illustration of the Purgatory Duel with $m=n=2$. The two dashed edges have probability~$\frac{1}{2}$ each.\label{fig:pd}}
\end{figure}
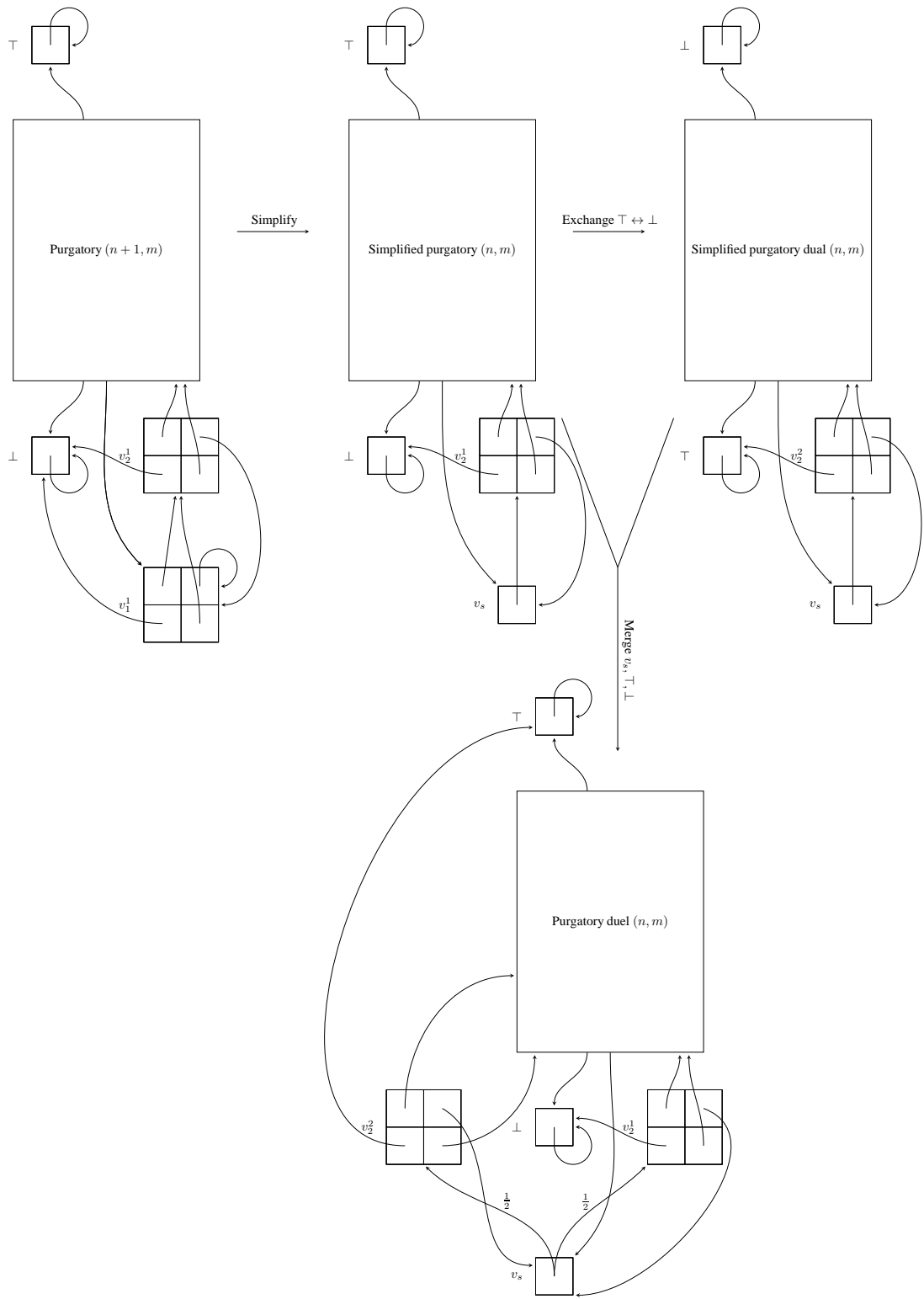

We now highlight two important features of our results, namely, the 
surprising aspects and the significance of the results.

\smallskip\noindent{\em Surprising aspects of our results.}
We discuss three surprising aspects of our result.
\begin{compactenum}

\item {\em The doubly-exponential lower bound on patience.}
For concurrent safety games, the properties of strategies resemble that
of concurrent discounted games.
In both cases, (1)~optimal strategies exist, (2)~there exist stationary 
strategies that are optimal, and (3)~locally optimal strategies (that play
optimally in every state with respect to the matrix games with values) 
are optimal.
The other class of concurrent games where optimal stationary strategies
exist are concurrent ergodic mean-payoff games, however, in contrast to 
safety and discounted games, in concurrent ergodic mean-payoff games not
all locally optimal strategies are optimal.
However, though for concurrent discounted games as well for concurrent ergodic
mean-payoff games, the optimal bound on the patience of $\epsilon$-optimal 
stationary strategies, for $\epsilon>0$, is exponential, we show a doubly-exponential 
lower bound on patience of $\epsilon$-optimal strategies for concurrent 
safety games, for $\epsilon>0$.

\item {\em The lower bound example.}
The second surprising aspect of our result is the lower bound example itself,
which had been elusive for safety games. 
The closer the lower bound example is to known examples, the greater is its value, 
as it is easier to understand, and illustrates the simplicity of our elusive example. 
Our example is obtained as follows:
We consider the Purgatory games $(n+1,m)$, which has two value classes, and in this 
game positional (pure memoryless) optimal strategies exist for the safety player.
We simplify the game by making the start state a deterministic state with one action 
for each player that with probability one goes to the next state.
We call this simplified Purgatory, and strategies in simplified Purgatory corresponds
to strategies in Purgatory $(n,m)$.
Then we consider the dual of the simplified Purgatory, which is basically a mirror
of the simplified Purgatory, with roles of the players exchanged.
In effect the dual is obtained by exchanging $\top$ and $\bot$.
Both in the simplified Purgatory and the dual of simplified Purgatory, 
there are two value classes, and positional optimal strategies exist for the 
safety player.
The Puragatory duel is obtained by simply merging the start states of the 
simplified Purgatory and the dual of the simplified Puragatory, 
thus from the start state we go to the first state of the Purgatory $(n,m)$ 
and the first state of the dual of Purgatory $(n,m)$, each with 
probability half; see Figure~\ref{fig:pd}.
Quite surprisingly we show that this simple merge operation gives a game 
where each state has a different value (i.e., that has linear number of value 
classes instead of two value classes), and the patience of optimal strategies 
increases from~1 (positional) to doubly-exponential 
(even for $\epsilon$-optimal strategies) for the safety player.

\item {\em From reachability to safety.}
The third surprising aspect is that we transfer a lower bound result from 
concurrent reachability to concurrent safety games.
Typically, the behavior of strategies of concurrent reachability and safety games are different,
e.g., for reachability games optimal strategies do not exist in general, whereas they exist
for concurrent safety games; and even in concurrent reachability games where optimal strategies
exist, not all locally optimal strategies are optimal, whereas in concurrent safety games
all locally optimal strategies are optimal.
Yet we show that a lower bound example for concurrent reachability games can be modified
to obtain a lower bound for concurrent safety games.
Moreover, we show that the strategy complexity results  with respect to the 
number of value classes in concurrent safety games is different and 
much more refined as compared to reachability games 
(see Table~\ref{tab:strategy-complexity}).

\end{compactenum}

\smallskip\noindent{\em Significance of our result.} There are several significant aspects
of our result.
\begin{compactenum}

\item {\em Roundedeness and patience.} As a measure of strategy complexity 
there are two important notions:
(a)~roundedness, which is more relevant from the computational aspect; and (b)~patience, which is
the traditional game theoretic measure.
The roundedness is always at least the patience, and in this work we present matching bounds for
patience and roundedness (i.e., our upper bounds are for roundedness which are matched with lower
bounds for patience).
Thus our results present a complete picture of strategy complexity with respect
to both well-known measures.

\item {\em Computational complexity.} In the study of stochastic games, the most well-studied way 
to obtain computational complexity result is to explicitly guess strategies and then verify
the resulting game obtained after fixing the strategy.
The lower bound for concurrent reachability games by itself did not rule out that improved computational
complexity bounds can be achieved through better strategy complexity for safety games.
Indeed, for constant number of value classes, we obtain a better complexity result due to the
exponential bound on roundedness.
Our doubly-exponential lower bound shows that in general the method of explicitly guessing strategies
would require exponential space, and would not yield \NP\ or \coNP\ upper bounds.
In other words, our results establish that to obtain \NP\ or \coNP\ upper bound for concurrent
safety games in general completely new techniques are necessary.

\item {\em Lower bound for algorithm.}
One of the most well-studied algorithm for games is the strategy-iteration algorithm that explicitly 
modifies strategies. 
Our result shows that any natural variant of the strategy-iteration algorithm for the safety player 
which explicitly compute strategies require exponential space in the worst-case.

\item {\em Complexity of strategies.} 
While the decision problem for games of whether the value is at least a threshold is the 
most fundamental question, along with values, witness (close-to-)optimal strategies are required. 
Our results present a tight bound on the complexity of strategies (which are as important as values).
\end{compactenum}
In summary, our main contributions are optimal bounds on strategy complexity, and our lower bounds 
have significant implications: it provides worst-case lower bound for a natural class of algorithms,
as well rules out a traditional method to obtain computational complexity results.

\subsection{Interesting technical aspects}\label{subsec:imp-tech}
\begin{remark}[Difference of exponential bounds]
In this work we present two different exponential bound on patience.
The first for zero-sum concurrent stochastic games, and the second for 
non-zero-sum concurrent stochastic games with safety objectives for all 
players.
However, note that the nature of the lower bounds are very different.
The first lower bound is exponential in the number of actions, and the size
of the state space is constant. 
In contrast, for non-zero-sum concurrent stochastic games with 
safety objectives for all players, if the size of the state space is constant,
then our upper bound on patience is polynomial. 
The second lower bound in contrast to the first lower bound is exponential
in the number of states (and the upper bound is polynomial in $m$ and also
the number of players). 
\end{remark}

\begin{remark}[Concurrent games with deterministic transitions]
We now discuss our results for concurrent games with deterministic transitions.
It follows from the results of~\cite{CDGH10} that for zero-sum games, there is 
a polynomial-time reduction from concurrent stochastic games to concurrent 
games with deterministic transitions. Hence, all our lower bound results for
zero-sum games also hold for concurrent deterministic games.
Observe that this is also true for our lower bound on non-zero sum games with at least one reachability player, since we reduce the problem to the zero-sum case.
However, in general for non-zero-sum games polynomial-time reductions 
from concurrent stochastic games to concurrent deterministic games are not 
possible. 
For example, for concurrent stochastic games with safety objectives for all 
players we establish an exponential lower bound on patience of strategies 
that constitute an $1/6$-Nash equilibrium, whereas in contrast, our upper
bound on patience shows that if the game is deterministic 
(i.e., $\delta_{\min}=1$) and $\epsilon$ is constant, then there always 
exists an $\eps$-Nash equilibrium that requires only polynomial patience.
\end{remark}

\begin{remark}[Nature of strategies for the reachability player]
Another important feature of our result is as follows: for zero-sum concurrent 
stochastic games, the characterization of~\cite{FM13} of $\epsilon$-optimal 
strategies as \emph{monomial} strategies for reachability objectives, 
separates the description of the strategies as a part that is a function of 
$\epsilon$, and a part that is independent $\epsilon$. 
The previous double-exponential lower bound on patience 
from~\cite{HKM09,HIM11} shows that the part dependent on $\epsilon$ requires 
double-exponential patience, whereas the part that is independent only 
requires linear patience. 
A witness for $\epsilon$-optimal strategies in Purgatory (as described 
in~\cite{TCS:AlfaroHK07} for the value-1 problem for general zero-sum 
concurrent stochastic game) can be obtained as a ranking function on states 
and actions, such that the actions with rank~0 are played with uniform 
probability (linear patience); and an action of  rank $i$ at a state of rank 
$j$ is played with probability roughly proportional to $\epsilon^{i^{j}}$.
In contrast, since we show lower bound for optimal strategies (and the 
strategies are symmetric) in Purgatory Duel, our lower bound implies that also 
the part that is independent of $\epsilon$ requires double-exponential 
patience in general (i.e., the probability description of $\epsilon$-optimal 
strategies needs to be doubly exponentially precise).
\end{remark}

\subsection{Concluding remarks}
In this work, we established the strategy complexity of zero-sum and non-zero-sum 
concurrent games with safety and reachability objectives.
Our most important result is the doubly-exponential lower bound on patience
for $\epsilon$-optimal strategies, for $\epsilon>0$, for the safety player
in concurrent zero-sum games.
Note that roundedness is at least patience, and we present upper bounds for 
roundedness that match our lower bound for patience, and thus we establish 
tight bounds both for roundedness and patience. 
Our results also imply tight bounds on ``granularity'' of strategies (i.e., the minimal difference between 
two probabilities). 
Since patience is the minimum positive probability, and some actions can be played with 
probability~0, a lower bound on patience is a lower bound on granularity, and 
an upper bound on roundedness is an upper bound on granularity.
Finally, there are many interesting directions of future work. 
The first question is the complexity of the value problem for concurrent safety games.
While our results show that explicitly guessing strategies does not yield desired complexity
results, an interesting question is whether new techniques can be developed to show 
that concurrent safety games can be decided in \coNP\ in general.
A second interesting question is whether variants of strategy-iteration algorithm can be 
developed that does not explicitly modify strategies, and does not have 
worst-case exponential-space complexity.

\bibliographystyle{abbrv}
\bibliography{diss}

\end{document}